\documentclass[prodmode,acmtoplas]{acmsmall}

\usepackage{mdframed}
\usepackage{longtable}
\usepackage{graphicx}
\usepackage{amssymb, amsmath, amsxtra}
\usepackage{stmaryrd}
\usepackage{bbm}
\usepackage{galois}
\usepackage{listings}
\usepackage{color}
\usepackage{hyperref}
\usepackage[ruled,onelanguage,noline,noend,titlenotnumbered]{algorithm2e}

\SetAlFnt{\small}
\SetAlCapFnt{\small}
\SetAlCapNameFnt{\small}
\SetAlCapHSkip{0pt}
\IncMargin{-\parindent}

\usepackage{tikz}
\usetikzlibrary{arrows,backgrounds,shapes}
\newcommand{\comment}[1]{}

\lstdefinelanguage{while}{morekeywords={while,if,do,to,bail,then,else,put,
  not, guard}}
\lstdefinelanguage{javascript}{morekeywords={for,guard,if,else,bail,function, var,while, goto}}

\lstdefinelanguage{JavaScript}{
  keywords={typeof, new, true, false, catch, function, return, null,
    catch, switch, var, if, in, while, do, for, goto, else, case, break, guard,
    bail, to},
  identifierstyle=\color{black},
  sensitive=false,
  comment=[l]{//},
  morecomment=[s]{/*}{*/},
  morestring=[b]',
  morestring=[b]"
}

\DeclareMathAlphabet{\mathpzc}{OT1}{pzc}{m}{it}
\newcommand{\bott}{{\bot_{\!\mathrm{T}}}}
\newcommand{\topt}{{\top_{\!\mathrm{T}}}}
\newcommand{\guard}{{\cdo{guard}}}
\newcommand{\undeff}{{\textit{undef}}}

\newcommand{\trh}{{\textrm{\rm tr}_{hp}}}
\newcommand{\trt}{{\textrm{\rm tr}}_{hp}^{\mathit{in}}}
\newcommand{\trf}{{\textrm{\rm tr}}_{hp}^{\mathit{out}}}

\newcommand{\rtr}{{\textrm{\rm rtr}}_{hp}}

\newcommand{\bbl}{\mathbbm{l}}
\newcommand{\bbh}{\mathbbm{h}}
\newcommand{\bbj}{\mathbbm{j}}
\newcommand{\elb}{{\textbf{\em l}}}
\newcommand{\cd}[1]{\text{\lstinline$#1$}}
\newcommand{\cdo}[1]{\mathit #1}

\newcommand{\cC}{\mathcal{C}}
\newcommand{\cX}{\mathcal{X}}
\newcommand{\ab}[2]{\langle #1, #2 \rangle}
\newcommand{\set}[2]{\lbrace #1 \mid #2 \rbrace}

\newcommand{\bsem}{\mathbf{B}}

\newcommand{\asem}{\mathbf{A}}

\newcommand{\ssem}{\mathbf{S}}
\newcommand{\esem}{\mathbf{E}}
\newcommand{\esemt}{\mathbf{E}^{\mathit{t}}}
\newcommand{\tsem}{\mathbf{T}}

\newcounter{commandposition}

\newcommand{\GP}{{\mathit{GP}}}
\newcommand{\vars}{{\mathit{vars}}}
\newcommand{\stateC}{{\mathcal{C}^s}}
\newcommand{\traceC}{{\mathcal{C}^t}}
\newcommand{\traceD}{{\mathcal{D}^t}}

\newcommand{\comC}{{\mathsf{C}}}

\def\ok#1{\mbox{\raisebox{0ex}[1ex][1ex]{$#1$}}}
\def \tuple#1{\langle #1 \rangle}
\def \suff#1{#1^{^{\!\shortrightarrow}}}
\newcommand{\Lra}{\Leftrightarrow}
\newcommand{\Ra}{\Rightarrow}
\newcommand{\La}{\Leftarrow}
\newcommand{\ra}{\rightarrow}
\newcommand{\kto}{k^{^{\!\shortrightarrow}}}
\newcommand{\rat}{\stackrel{\tau}{\rightarrow}}
\newcommand{\rad}{\stackrel{\delta}{\rightarrow}}

\newcommand{\raa}{\stackrel{a}{\rightarrow}}

\newcommand{\Raac}{\stackrel{\hat{a}}{\Rightarrow}}

\newcommand{\Php}{P_{hp}}

\DeclareMathOperator{\Type}{\mathrm{Type}}
\DeclareMathOperator{\Exp}{Exp}
\DeclareMathOperator{\FV}{FV}
\DeclareMathOperator{\Act}{\mathit{Act}}
\DeclareMathOperator{\Stm}{Stm}
\DeclareMathOperator{\Cmd}{Cmd}
\DeclareMathOperator{\BExp}{BExp}

\DeclareMathOperator{\State}{State}
\DeclareMathOperator{\TState}{tState}

\DeclareMathOperator{\Store}{Store}

\DeclareMathOperator{\Storeav}{{Store^\sharp_{\mathit{value}}}}
\DeclareMathOperator{\Storet}{Store^{\mathrm{t}}}

\DeclareMathOperator{\Trace}{Trace}

\DeclareMathOperator{\Valueu}{Value_{\!\,\mathrm{u}}}
\DeclareMathOperator{\Valueush}{Value^\sharp}
\DeclareMathOperator{\Value}{Value}

\DeclareMathOperator{\Var}{Var}
\DeclareMathOperator{\CP}{CP}
\DeclareMathOperator{\CPst}{CP_{\!\mathit{st}}}
\DeclareMathOperator{\Program}{Program}
\newcommand{\ud}{\triangleq}

\DeclareMathOperator{\full}{{\mathit{full}}}
\DeclareMathOperator{\tdo}{\mathit{tdo}}
\DeclareMathOperator{\tod}{\mathit{to}}
\DeclareMathOperator{\out}{\mathit{out}}
\DeclareMathOperator{\sch}{\mathit{sc}}
\DeclareMathOperator{\osch}{\mathit{osc}}

\DeclareMathOperator{\sloop}{\mathit{loop}}
\DeclareMathOperator{\hot}{\mathit{hot}}
\DeclareMathOperator{\outerhot}{\mathit{outerhot}}
\DeclareMathOperator{\scount}{\mathit{count}}

\DeclareMathOperator{\extr}{\mathit{extr}}
\DeclareMathOperator{\stitch}{\mathit{stitch}}
\DeclareMathOperator{\hotcut}{\mathit{hotcut}}
\def\grasse#1{{\llbracket #1 \rrbracket}}

\newcommand{\avalues}{{\alpha_{\mathit{value}}^{\sqsubseteq}}}
\newcommand{\gvalues}{{\gamma_{\mathit{value}}^{\sqsubseteq}}}

\newcommand{\avalue}{\alpha_{\mathit{value}}}
\newcommand{\gvalue}{\gamma_{\mathit{value}}}
\newcommand{\astore}{\alpha_{\mathit{store}}}
\newcommand{\gstore}{\gamma_{\mathit{store}}}

\newcommand{\ahotn}{\alpha_{\mathit{hot}}^N}

\newcommand{\atype}{\alpha_{\mathit{type}}}
\newcommand{\gtype}{\gamma_{\mathit{type}}}
\newcommand{\type}{\mathit{type}}
\newcommand{\leqt}{\leq_{\mathrm{t}}}

\newcommand{\st}{\mathit{st}}

\DeclareMathOperator{\Int}{Int}
\DeclareMathOperator{\Bool}{Bool}
\DeclareMathOperator{\Boolu}{Bool_{\mathrm{u}}}
\DeclareMathOperator{\String}{String}
\DeclareMathOperator{\Char}{Char}
\DeclareMathOperator{\Undef}{Undef}

\DeclareMathOperator{\Types}{\mathrm{Types}}

\DeclareMathOperator{\ts}{{\text{\rm ts}}}
\DeclareMathOperator{\cf}{{\text{\rm cf}}}

\begin{document}

\markboth{S.~Dissegna et al.}{An Abstract Interpretation-based Model of Tracing Just-In-Time Compilation}

\title{An Abstract Interpretation-based Model of Tracing Just-In-Time Compilation}
\author{STEFANO DISSEGNA
           \affil{University of Padova}
FRANCESCO LOGOZZO
           \affil{Facebook Inc.}
FRANCESCO RANZATO
           \affil{University of Padova}
}           

\begin{abstract}
Tracing just-in-time compilation is a  popular compilation technique 
for the efficient implementation of dynamic languages, 
which is commonly used for JavaScript, Python and PHP.
It relies on two key ideas.
First, it monitors program execution in order to detect so-called hot paths, i.e., the most frequently executed program paths.
Then, hot paths are optimized by exploiting some information on program stores which is available and therefore gathered at runtime. 
The result is a residual program where the optimized hot paths are guarded by  sufficient conditions ensuring 
some form of equivalence with the original program.
The  residual program is persistently mutated during its execution, e.g., to  add new optimized hot paths or to merge existing paths. 
Tracing compilation is thus fundamentally different from traditional static compilation.
Nevertheless, despite the practical 
success of tracing compilation, very little is known about its theoretical foundations.
We provide a formal model of tracing compilation of programs using abstract interpretation.
The monitoring  phase (viz., hot path detection) corresponds to an abstraction of the trace semantics 
of the program that captures the most frequent occurrences of sequences 
of program points together with 
an abstraction of their corresponding stores, e.g., a type environment.
The optimization  phase (viz., residual program generation) 
corresponds to a transform of the original program that preserves its trace 
semantics up to a given observation as modeled by some abstraction.
We provide a generic framework to express dynamic optimizations along hot paths and to prove them correct. 
We instantiate it to prove the correctness of dynamic type specialization and constant variable folding.
We show that our framework is more general  than the 
model of tracing compilation introduced 
by Guo and Palsberg [2011] which is based on operational 
bisimulations.
In our model we can naturally express hot path reentrance and 
common optimizations like dead-store elimination, which are either excluded or 
unsound in Guo and Palsberg's framework.
\end{abstract}



\begin{bottomstuff}
The work of Francesco Logozzo was carried out while being affiliated with Microsoft Research, Redmond, WA, USA.
The work of Francesco Ranzato was partially supported 
by Microsoft Research Software Engineering Innovation Foundation~2013 Award (SEIF~2013) and 
by the University of
Padova under the PRAT projects 
BECOM and ANCORE. 

Author's addresses: S.~Dissegna {and} F.~Ranzato, Dipartimento di Matematica, University
of Padova, Padova, Italy; 
F.~Logozzo, Facebook Inc., Seattle, WA, USA.
\end{bottomstuff}

\maketitle

\section{Introduction}
Efficient traditional static compilation of popular dynamic languages like JavaScript, Python and 
PHP is very hard if not impossible.
In particular, these languages present so many dynamic features which make all traditional static analyses used for program optimization very imprecise.
Therefore, practical implementations of dynamic languages should rely on dynamic information in order to produce an optimized version of the program.
Tracing just-in-time (JIT) compilation (TJITC) \cite{bala2000,bauman2015,bebenita2010spur,Bohm2011,bolz2009,bolz2011,gal-vee2006,gal2009,pall,schilling2013} has emerged as a valuable implementation and optimization technique for dynamic languages (and not only, e.g.\ Java \cite{Haubl2011,Haubl2014,Inoue2011}).
For instance, the Facebook HipHop virtual machine for PHP and the V8 JavaScript engine of Google Chrome use some form of tracing compilation~\cite{Adams2014,hiphop,chrome}.
The Mozilla Firefox JavaScript engine used to have a tracing engine, called TraceMonkey, which has been 
later substituted by whole-method just-in-time compilation engines 
(initially J\"{a}gerMonkey and then IonMonkey) \cite{tracemonkey,ionmonkey}.  

\subsubsection*{\textbf{The Problem}}
Tracing JIT compilers leverage runtime profiling of programs to detect and record often executed paths, called hot paths, and then they optimize and compile only these paths at runtime. 
A path is a linear sequence (i.e., no loops or join points are allowed) of instructions through the program. 
Profiling may also collect information about the values that the program 
variables may assume during the execution of that path, which is then used to specialize/optimize the code of the hot path. 
Of course,  this information is not guaranteed to hold for all the subsequent executions of the hot path. 
Since optimizations rely on that information, the hot path is augmented with guards that check the profiled conditions, such as, for example, 
variable types and constant variables. 
When a guard fails, execution jumps back to the old,  non-optimized code.
The main hypotheses of tracing compilers, confirmed by the practice, are: (i)~loop bodies are the 
most interesting code to optimize, so they only consider paths inside program 
loops; and (ii)~optimizing straight-line code is easier than a whole-method 
analysis (involving loops, goto, \emph{etc.}).

Hence, tracing JIT compilers look quite different than traditional compilers. 
These differences raise some natural questions  on trace compilation: (i)~what 
is a viable formal model, which is generic yet realistic enough to capture the behavior of real optimizers? (ii)~which optimizations are sound? (iii)~how can one  prove their soundness?
In this paper we answer these questions.

Our formal model is based on program trace semantics~\cite{cousot2002} and abstract interpretation \cite{CC77,cousot2002systematic}.
Hot path detection is modeled just as an abstraction of the trace semantics of the program, which only retains: 
(i)~the sequences of program points which are repeated more than some 
threshold; (ii)~an abstraction of the possible program stores, 
e.g., the type of the variables instead of their concrete values.
As a consequence, a hot path does not contain loops nor join points.
Furthermore, in the hot path, all the correctness conditions (i.e., guards) are explicit, for instance before performing  integer addition, we should check that the operands are  integers. 
If the guard condition is not satisfied then the execution leaves the hot path, reverting to the non-optimized code.
Guards are essentially elements of some abstract domain, which is then 
left as a parameter in our model. 
The hot path is then optimized using standard compilation techniques---we only require the optimization to be sound.

We define the correctness of the residual (or extracted) program in terms of an abstraction of its trace semantics: the residual program is correct  if it is indistinguishable, up to some abstraction
of its trace semantics, from the original program.
Examples of abstractions are the program store at the exit of a method, or the stores at loop entry and loop exit points.

\subsubsection*{\textbf{Main Contributions}}
This paper puts forward a formal model of TJITC whose key features are as follows: 

\begin{enumerate}
\item[--] We provide the first model of tracing compilation based on
  abstract interpretation of trace semantics of programs.
\item[--] We provide a more general and realistic framework than the model of TJITC by
\citeN{palsberg}  based on program bisimulations: we employ a less restrictive correctness criterion that enables  the correctness proof of
  practically implemented optimizations; hot paths can be annotated with
  runtime information on the stores, notably type information; 
  optimized hot loops can be re-entered.
  \item[--] We formalize and prove the correctness of type specialization of hot paths.
\end{enumerate}

Our model focusses on source-to-source
program transformations and optimizations of a low level imperative language with untyped 
global
variables, which may play the role of intermediate language of some
virtual machine. Our starting point is that 
program optimizations can be seen as transformations that lose some information
on the original program,  
so that optimizations can be viewed as
approximations and in turn can be formalized as abstract interpretations. 
More precisely, we rely on the insight by \citeN{cousot2002systematic} that 
a program source can be seen as an abstraction of its
trace semantics, i.e.\ the set of all possible execution sequences,  
so that a source-to-source optimization can be viewed as 
an abstraction of a transform of the program trace semantics.
In our model, soundness of program optimizations is defined as
program equivalence w.r.t.\ an observational abstract interpretation of the program trace
semantics. Here, 
an observational abstraction induces a correctness criterion by
describing what is observable about
program executions, so that program equivalence means that two programs are
indistinguishable by looking only at their observable behaviors. 

A crucial part of tracing compilation is the selection of the hot path(s) to
optimize. This choice is made at runtime based on  program executions, so it
can be seen once again as an abstraction of  trace semantics. 
Here, a simple trace abstraction selects cyclic
instruction sequences, i.e.
loop paths, that appear at least $N$
times within a single execution trace. These instruction sequences are
recorded together with some property of the values assumed by program variables at that point,  
which is represented as an abstract store belonging to a suitable 
store abstraction, which in general depends on the successive
optimizations to perform. 

A program optimization can be seen as an abstraction of a semantic
transformation of program execution traces, as described by
\citeN{cousot2002systematic}. The advantage of this approach is that
optimization properties, such as their correctness, are easier to prove
at a semantic level. The optimization itself can be defined on the whole
program or, as in the case of real tracing JIT compilers, can be
restricted to the hot path. This latter restriction is achieved by transforming the
original program so that the hot path is extracted, i.e.\ made explicit: 
the hot path is added to the program as a
path with no join points that jumps back to the original code when
execution leaves it. A guard is placed before each command in this hot path that
checks if the necessary conditions, as selected by the store
abstraction, are satisfied.
A program optimization can then be confined to the hot path only, making it
linear, by ignoring the parts of the program outside it. The guards added to the hot path
allows us to retain precision.

We apply our TJITC model to type specialization.  
Type specialization is definitely the key optimization for
dynamic languages such as Javascript \cite{gal2009}, as they make available
generic operations whose execution depends on the type of runtime 
values of their operands. Moreover, as a further application of our model, we consider 
the constant variable folding optimization along hot paths, which relies on the standard
constant propagation abstract domain \cite{WZ91}.

\subsubsection*{\textbf{Related Work}}
A formal model for tracing JIT compilation has been put forward  by 
\citeN{palsberg} at POPL symposium. It is based on operational bisimulation \cite{milner95} 
to describe the
equivalence between source and optimized programs. We show 
how this model can be
expressed within our framework through the following steps: Guo and
Palsberg's language is compiled into ours; we then exhibit an observational
abstraction which is equivalent to Guo and Palsberg's correctness criterion;
finally,  after some minor changes that address a few differences in path
selection, the transformations performed on the source program turn out to be the
same.  
Our framework overcomes some significant limitations in Guo and Palsberg's model.
The bisimulation equivalence model used in~\cite{palsberg}  implies that the optimized program has to match every change to the store made by the original program, whereas in practice we only need this   match to hold in certain program points and for some  variables, such as in output instructions. 
This limits the number of real optimizations that can be modeled in this framework. 
For instance,  dead store elimination is proven unsound in \cite{palsberg}, while it is implemented in actual tracing compilers \cite[Section~5.1]{gal2009}.
Furthermore,  their formalization fails to model some important features of actual TJITC implementation: (i)~traces are mere linear paths of instructions, i.e., they cannot be annotated with store properties;
(ii)~hot path selection is completely non-deterministic, since they  do not model a selection criterion; and, (iii)~once execution leaves an
optimized hot path the program will not be able to re-enter it. 

It is also worth citing that  abstract interpretation of program trace semantics roots at the foundational work by Cousot \citeyear{Cousot1997,cousot2002} and has been widely used as a technique for defining a range of static program analyses \cite{Barbuti1999,Colby1996,handjieva1998,Logozzo2009,Rival2007,Schmidt1998,Spoto2003}. 
Also, \citeN{Rival04} describes various program optimizations as the trace abstractions they preserve.
In the Cousot and Cousot terminology~\cite{cousot2002systematic}, Rival's approach corresponds to offline transformations whereas tracing compilation is an online transformation.

\subsubsection*{\textbf{Structure}}
The rest of the paper is organized as follows. Sections~\ref{language}
and~\ref{abs-sec} contain some necessary background: the language considered in the paper and 
its operational trace semantics are defined in Section~\ref{language}, while Section~\ref{abs-sec} 
recalls some basic notions of abstract interpretation, in particular for defining abstract domains of program stores. 
Hot paths are formally defined in Section~\ref{hp-sect} as a suitable abstract interpretation of program traces, while
Section~\ref{trace-ext-sec} defines the program transform for extracting a given hot path. The correctness of
the hot path extraction transform is defined and proved correct in Section~\ref{obs}, which also introduces in Subsection~\ref{chpo-sec}
program optimizations along hot paths together with a methodology for proving their correctness. Section~\ref{type_specialization} 
applies our model of hot path optimization to type specialization of untyped program commands, while Section~\ref{cf-sec} 
describes an application to constant variable folding along hot paths. Nested hot paths and the corresponding program
transform for their extraction are the subject of Section~\ref{mte-sec}. Section~\ref{GP-sec} provides a thorough formal
comparison of our model with \citeN{palsberg}'s framework for tracing compilation. Finally, Section~\ref{concl-sec} concludes,
also discussing some directions for future work.

This is an expanded and revised version of the POPL symposium article \cite{DLR14} including all the proofs.

\section{Language and Concrete Semantics}
\label{language}
\subsection{Notation}
Given a finite set $X$ of objects, we will use the following notation concerning sequences: 
$\epsilon$ is the empty sequence;
$X^+$ is the set of nonempty finite sequences of 
objects of $X$; $X^* \ud X^+ \cup \{\epsilon\}$; if $\sigma \in X^*$ then $|\sigma|$ denotes
the length of $\sigma$; indices of objects  in a nonempty sequence $\sigma\in X^+$ start from
$0$ and thus range in the interval $[0,|\sigma|)\ud [0,|\sigma|-1]$; if $\sigma\in X^+$
and $i\in [0,|\sigma|)$ then $\sigma_i\in X$ (or $\sigma(i)$) denotes the $i$-th object in $\sigma$;
if $\sigma\in X^*$ and $i,j\in [0,|\sigma|)$ then $\sigma_{[i,j]} \in X^*$ denotes the subsequence 
$\sigma_i \sigma_{i+1} \ldots \sigma_{j}$, which is therefore the empty sequence if $j< i$, while if
$k\in \mathbb{N}$ then $\sigma_{\kto}\in X^*$ denotes the
suffix $\sigma_k \sigma_{k+1}\ldots \sigma_{|\sigma|-1}$, which is the empty sequence when $k\geq |\sigma|$.  

If $f:X\ra Y$ is any function then its collecting version $f^c:\wp(X)\ra \wp(Y)$ is defined pointwise 
by $f^c(S)\ud \{f(x)\in Y~|~ x\in S\}$, and when clear from the context, 
by a slight abuse of notation, it is sometimes denoted by $f$ itself.  

\subsection{Syntax}
We 
consider a basic low level language with untyped global
variables, a kind of elementary dynamic language, which is defined through the notation used
in \cite{cousot2002systematic}. 
Program commands range in $\mathbb{C}$ and consist of 
a labeled action which specifies 
a next label ($\L$ is the undefined label, where the execution becomes stuck: it is used
for defining final commands). 
\begin{equation*}
  \begin{array}{rl}
  \text{Labels:~~} & L \in \mathbb{L} \quad\quad \L\not \in \mathbb{L} \\
  \text{Values:~~} & v\in \Value\\
\text{Variables:~~} & x \in \Var  \quad\quad  \\
 \text{Expressions:~~} & \Exp\ni E ::= v \mid x \mid E_{1} + E_{2}\\
 \text{Boolean Expressions:~~}   & \BExp\ni B ::= \cd{tt} \mid \cd{ff} \mid
    E_{1} \leq E_{2} \mid  \neg B \mid B_1 \wedge B_2 \\
\text{Actions:~~}    &  \mathbb{A}\ni A ::= x := E \mid B \mid \cd{skip}\\
\text{Commands:~~}    &  \mathbb{C}\ni C ::= L : A \rightarrow L' \quad(\text{with~}L'\in \mathbb{L}\cup\!\{\L\})
  \end{array}
\end{equation*}

\noindent
For any command $L : A \rightarrow L'$, we use the following notation: 
$$lbl(L : A \rightarrow L') \ud L,\quad act(L : A \rightarrow L') \ud A,\quad suc(L : A \rightarrow L')\ud L'.$$
Commands $L : B \rightarrow L'$ whose action is a Boolean expression are called
conditionals. 
A program $P\in  \wp(\mathbb{C})$ is a (possibly infinite, at least in theory) 
set of commands.  In order to be well-formed, if a program $P$ includes a 
conditional $C\equiv L : B \rightarrow L'$ then $P$ must also include 
a unique complement conditional  $L : \neg B  \rightarrow L''$,
which is denoted by $cmpl(C)$ or $C^c$, where $\neg \neg B$ is taken to be equal to 
$B$, so that $cmpl(cmpl(C))=C$. The set of well-formed programs is denoted by $\Program$.
In our examples, programs $P$ will be deterministic,  
i.e., for any $C_1,C_2\in P$ such that $lbl(C_1)=lbl(C_2)$:
(1)~if $act(C_1)\neq act(C_2)$ then $C_1 = cmpl(C_2)$; (2) if $act(C_1)= act(C_2)$ then $C_1=C_2$.
We say that two programs $P_1$ and $P_2$ are equal up to label renaming, denoted by $P_1\cong P_2$, 
when there exists a suitable
renaming for the labels of $P_1$ that makes $P_1$ equal to $P_2$.

\subsection{Transition Semantics}\label{sem-sec}
The language semantics relies 
on values ranging in $\Value$, possibly undefined values ranging in $\Valueu$, truth values in  $\Bool$, 
possibly undefined truth values ranging in $\Boolu$ and
type names ranging in $\Types$, which are defined as follows:
\[
\begin{array}{c}
\Value \ud \mathbb{Z}  \cup \Char^* \qquad \Valueu \ud \mathbb{Z}  \cup \Char^* \cup \{\undeff\}\\[5pt]
  \Bool \ud \{ \textit{true},\, \textit{false} \}\qquad \Boolu \ud \{ \textit{true},\, \textit{false},\undeff \}\\[5pt]
  \Types \ud \{\Int,  \String, \Undef, \topt, \bott\} 
\end{array}
\]
where $\Char$ is a nonempty finite set of characters and $\undeff$ is a distinct symbol. 
The mapping $\type:\Valueu \ra \Types$ 
provides the type of any possibly undefined value:
\[
\type(v) \ud 
\begin{cases}
\Int & \text{if } v\in \mathbb{Z}\\
\String & \text{if } v\in \Char^*\\
\Undef & \text{if } v=\undeff
\end{cases}
\] 
The type names $\bott$ and $\topt$ will be used in Section~\ref{type_specialization} as, respectively,  
top and bottom type, 
that is,  subtype and supertype of all types. 

\begin{figure*}
\begin{mdframed} 
\begin{equation*}
 \begin{aligned}
  &\esem: \Exp \ra \Store \ra \Valueu\\
  &\esem\grasse{v}\rho \ud v\quad \esem\grasse{x}\rho \ud \rho(x)\\
    &\esem\grasse{E_{1} + E_{2}}\rho \ud
     \begin{cases}
       \esem\grasse{E_1}\rho +_{\mathbb{Z}} \esem\grasse{E_2}\rho &
      \text{if } \type(\esem\grasse{E_i}\rho)=\Int\\
       \esem\grasse{E_1}\rho \cdot \esem\grasse{E_2}\rho   &
      \text{if }  \type(\esem\grasse{E_i}\rho)=\String \\
       \undeff & \text{otherwise} 
         \end{cases} \\[10pt]
    &\bsem: \BExp \ra \Store \ra \Boolu\\
    &\bsem\grasse{\cd{tt}}\rho \ud \textit{true}\quad \bsem\grasse{\cd{ff}}\rho \ud \textit{false}\\
    & \bsem\grasse{E_1 \leq E_2}\rho \ud
         \begin{cases}
       \esem\grasse{E_1}\rho \leq_{\mathbb{Z}} \esem\grasse{E_2}\rho &
      \text{if }  \type(\esem\grasse{E_i}\rho)=\Int\\
       \exists \sigma\in \String.\, \esem\grasse{E_2}\rho = (\esem\grasse{E_1}\rho)\!\cdot\! \sigma&
      \text{if }  \type(\esem\grasse{E_i}\rho)=\String\\
       \undeff &
      \text{otherwise }   
         \end{cases} \\
    &\bsem\grasse{\neg B}\rho \ud  \neg \bsem\grasse{B}\rho \quad
    \bsem\grasse{B_1 \wedge B_2}\rho \ud  \bsem\grasse{B_1}\rho \wedge \bsem\grasse{B_2}\rho\\[10pt]
    &\asem: \mathbb{A} \ra \Store \ra \Store\cup \,\{\bot\}\\
    & \asem\grasse{\cd{skip}}\rho \ud  \rho\\
     &\asem\grasse{x := E}\rho \ud  
     \begin{cases} \rho[x/\esem\grasse{E}\rho] &\text{if } \esem\grasse{E}\rho \neq  \undeff\\
                                            \bot & \text{if } \esem\grasse{E}\rho =  \undeff
                                            \end{cases}\\
      &\asem\grasse{B}\rho \ud \begin{cases} \rho &\text{if } \bsem\grasse{B}\rho = \textit{true}\\
                                            \bot & \text{if } \bsem\grasse{B}\rho \in \{\textit{false},\undeff\}
                                            \end{cases}
     \end{aligned}
 \end{equation*}
\end{mdframed}
\caption{Semantics of program expressions and actions.}
\label{fig:semantics}
\end{figure*}

Let   $\Store \ud \Var  \rightarrow \Valueu$ denote the set of possible stores on variables in $\Var$, where
$\rho(x)=\undeff$ means that the store $\rho$ is not defined on a program variable $x\in \Var$.  
Hence, let us point out that the symbol
$\undeff$ will be used to represent both store undefinedness and 
a generic error when evaluating an expression (e.g., additions and comparisons between integers and strings), two situations 
which are not distinguished in our semantics.  
A store $\rho\in \Store$ will be denoted by $[x/\rho(x)]_{\rho(x)\neq \undeff}$, thus
omitting undefined variables, while $[\,]$ will denote the totally undefined store. 
If $P\in \Program$ then $\vars(P)$ denotes
the set of variables in $\Var$ that occur in $P$, so that
$\Store_P \ud \vars(P)  \rightarrow \Valueu$ is the set of possible stores for $P$.

The semantics of expressions $\esem$, Boolean expressions $\bsem$ and  program actions $\asem$ 
is standard and goes as defined 
in Fig.~\ref{fig:semantics}. 
Let us remark that: 
\begin{itemize}
\item[{\rm (i)}]
the binary function $+_{\mathbb{Z}}$ denotes integer addition, $\leq_{\mathbb{Z}}$ denotes integer comparison, 
while $\cdot$ is string concatenation;  
\item[{\rm (ii)}]
logical negation and conjunction in $\Boolu$
are extended in order to handle 
$\undeff$ as follows: $\neg \undeff = \undeff$ and $\undeff \wedge b = \undeff = b \wedge \undeff$; 
\item[{\rm (iii)}]
$\rho[x/v]$ denotes a store update for the variable $x$ with $v\in \Value$; 
\item[{\rm (iv)}]
the distinct symbol $\bot\not\in \Valueu$ is used to denote the result of: 
$\asem\grasse{x := E}\rho$ when the evaluation of the expression $E$ for $\rho$ generates an error; 
$\asem\grasse{B}\rho$ when  the evaluation of the Boolean expression $B$ for $\rho$ is either $\textit{false}$ or generates an error. 
\end{itemize}
With  a slight abuse of  notation  we also consider the
collecting versions of the semantic functions in Fig.~\ref{fig:semantics}, which are defined as follows:
\[
  \begin{array}{lll}
   &\esem: \Exp \ra \wp(\Store) \ra \wp(\Valueu) \\
   & \esem\grasse{E}S \ud \{\esem\grasse{E}\rho \in \Valueu~|~ \rho \in S\}\\[7.5pt]
  &\bsem: \BExp \ra \wp(\Store) \ra \wp(\Store)  \\
  & \bsem\grasse{B}S \ud \{\rho \in S ~|~ \bsem\grasse{B}\rho =\textit{true}\}\\[7.5pt]
   &\asem: \mathbb{A} \ra \wp(\Store) \ra \wp(\Store)  \\
 & \asem\grasse{A}S\ud \{\asem\grasse{A}\rho ~|~ \rho \in S,\, \asem\grasse{A}\rho \in \Store\}
     \end{array}
\]
Let us point out that, in the above collecting versions, if $\esem\grasse{E}\rho =\undeff$ then $\esem\grasse{E}\{\rho\}=\{\undeff\}$ and
$\asem\grasse{x:=E}\{\rho\}=\varnothing$, while
if $\bsem\grasse{B}\rho \in \{\textit{false},\undeff\}$
 then $\bsem\grasse{B}\{\rho\}=\varnothing$ and 
$\asem\grasse{B}\{\rho\} = \varnothing$.

Generic program states are pairs  of stores and commands: $\State\ud \Store\times\, \mathbb{C}$. We extend  the 
previous functions $lbl$, $act$ and $suc$ to be defined on states, meaning that they are defined on the command
component of a state. Also, $\mathit{store}(s)$ and $\mathit{cmd}(s)$ 
return, respectively, the store and the command of a state $s$. 
The transition semantics 
${\ssem: \State \rightarrow \wp(\State)}$  is a relation between generic states defined as follows:
$$
  \ssem\tuple{\rho,C}\ud \{\tuple{\rho',C'}\in  \State~|~ \rho' \in \asem\grasse{act(C)}\{\rho\},\: 
suc(C)=lbl(C') \}.
$$

If $P$ is a program then  $\State_P\ud \Store_P\times P$ is the set of possible states of $P$.  
Given $P\in \Program$,  the program transition relation  
$\ssem\grasse{P}: \State_P \rightarrow \wp(\State_P)$  between states of $P$ is defined as:
$$
  \ssem\grasse{P}\tuple{\rho,C}\ud \{\tuple{\rho',C'}\in  \State_P~|~ \rho' \in \asem\grasse{act(C)}\{\rho\},\: 
C' \in P,\: suc(C)=lbl(C') \}.
$$

\noindent
Let us remark that, according to the above definition, 
if $C\equiv L:A\ra L'$, $C_1 \equiv L' : B\ra L''$ and $C_1^c\equiv L': \neg B
\ra L'''$ are all commands in $P$ and $\rho'\in \asem\grasse{A}\rho$ then 
we have that $\ssem\grasse{P}\tuple{\rho,C} = \{\tuple{\rho',C_1}, \tuple{\rho',C_1^c}\}$.

A state $s\in \State_P$ is stuck for $P$ when $\ssem\grasse{P} s = \varnothing$. 
Let us point that:
\begin{itemize}
\item[{\rm (i)}]
If the  conditional command of a
state  $s=\tuple{\rho, L: B \ra L'}\in \State_P$ is
such that $\bsem\grasse{B}\rho = \textit{false}$ then $s$ is stuck for $P$  
because there exists no store $\rho' \in \asem\grasse{B}\{\rho\}=\varnothing$. 
\item[{\rm (ii)}]
If the command of a state $s=\tuple{\rho, L:A\ra \L}\in \State_P$ has  
the undefined label $\L$ as next label then $s$ is stuck for $P$.  
\item[{\rm (iii)}] We have a stuck state $s$ when 
an error happens. E.g., this is the case for an undefined evaluation of an addition as in $s=\tuple{[y/3,z/\texttt{foo}], L:x:=y+z\ra L'}$
and for an undefined evaluation of a Boolean expression as in $s=\tuple{[y/3,z/\texttt{foo}], L:y\leq x \ra L'}$.
\end{itemize}
  
Programs typically have an entry point, and this is modeled through
a distinct initial label $L_{\iota}\in \mathbb{L}$ from which execution starts.
$\State_P^\iota \ud \{\tuple{\rho,C}~|~lbl(C) = L_\iota\}$ 
denotes the set of possible initial states for $P$.

\subsubsection{Trace Semantics}
A partial trace is any nonempty finite sequence of generic program 
states which are related by the transition relation
$\ssem$. Hence, the set $\Trace$ of partial traces is defined as follows:
$$
  \Trace \ud \{\sigma \in \State^+ ~|~ \forall i\in [1,|\sigma| ).\: \sigma_i \in \ssem \sigma_{i-1}\}. 
$$
The partial trace semantics of $P\in \Program$ is in turn defined 
as follows:
$$
  \tsem\grasse{P} = \Trace_P \ud \{\sigma \in (\State_P)^+ ~|~ 
  \forall i\in [1,|\sigma|).\: \sigma_{i} \in \ssem\grasse{P} \sigma_{i-1}\}.
$$

A trace $\sigma\in \Trace_P$ is complete if for any state $s\in \State_P$, $\sigma s\not \in \Trace_P$ and
$s\sigma \not\in \Trace_P$. Observe that $\Trace_P$ contains all the possible partial traces
of $P$, complete traces included. 
Let us remark that a trace $\sigma\in \Trace_P$ does not necessarily 
begin with an initial state, namely it may happen that $\sigma_0\not\in \State_P^\iota$. 
Traces of $P$ starting from initial states are denoted by
$$\tsem^\iota\grasse{P} = \Trace_P^\iota \ud \{\sigma \in \Trace_P ~|~ \sigma_0 \in \State_P^\iota\}.$$ 
Also, a complete trace
$\sigma\in\Trace^\iota_P$ such that $suc (\sigma_{|\sigma|-1}) = \L$ 
corresponds to a terminating run of the program $P$.

\begin{example}\label{example-init}\rm
Let us consider the program $Q$ below written in some while-language:

\medskip
{ 
$x:=0$;

\While{~$(x\leq 20)$}{
   \Indp
   \noindent
   $x:=x+1$;\\
   \lIf{$(x\%3=0)$}{$x:=x+3$;}
}
}

\medskip
\noindent
Its translation as a program $P$ in our language is given below, with $L_\iota = L_0$, where,
with a little abuse, we assume an extended syntax that  
allows expressions like ${x \% 3 =0}$.
\begin{equation*}
\begin{aligned}
    P= \big\{ &C_0 \equiv L_0: x:=0 \ra L_1, \\
    & C_1 \equiv L_1: x\leq 20 \ra L_2,\, C_1^c \equiv L_1: \neg(x\leq 20) \ra L_5,\\
    & C_2 \equiv L_2: x:=x+1 \ra L_3,\\
    & C_3 \equiv L_3: (x \% 3 =0) \ra L_4,\, C_3^c \equiv L_3: \neg (x\% 3 =0) \ra L_1\\
    & C_4 \equiv L_4: x:=x+3 \ra L_1,\, C_5 \equiv L_5: \cd{skip} \ra \L\big\}
\end{aligned}
\end{equation*}

\noindent
Its trace semantics  from initial states $\Trace^\iota_P$ includes the following complete traces, 
where $[\,]$ is the initial totally undefined store. 
\begin{equation*}
{%
\begin{array}{l}
\tuple{[\,],C_0}\tuple{[x/0],C_1^c} \\
\tuple{[\,],C_0}\tuple{[x/0],C_1}\tuple{[x/0],C_2}\tuple{[x/1],C_3} \\
\tuple{[\,],C_0}\tuple{[x/0],C_1}\tuple{[x/0],C_2}\tuple{[x/1],C_3^c}\tuple{[x/1],C_1^c}\\
 \;\cdots\\
  \;\cdots\\
\tuple{[\,],C_0}\tuple{[x/0],C_1}\tuple{[x/0],C_2}\tuple{[x/1],C_3^c}\tuple{[x/1],C_1}\cdots \tuple{[x/21],C_4}\tuple{[x/24],C_1}\\
\tuple{[\,],C_0}\tuple{[x/0],C_1}\tuple{[x/0],C_2}\tuple{[x/1],C_3^c}\tuple{[x/1],C_1}\cdots \tuple{[x/21],C_4}\tuple{[x/24],C_1^c}\tuple{[x/24],C_5}
\end{array}
}
\end{equation*}
Observe that the last trace corresponds to a terminating run of $P$.
\qed
\end{example}

\section{Abstractions}\label{abs-sec}
\subsection{Abstract Interpretation Background}
\label{intro_abs_int}

In standard abstract interpretation~\cite{CC77,CC79}, 
abstract domains, also called abstractions, 
are specified by Galois connections/insertions
(GCs/GIs for short) or, equivalently, adjunctions. 
Concrete and abstract domains, $\tuple{C,\leq_C}$ and
$\tuple{A,\leq_A}$, are assumed to be complete lattices 
which are related by abstraction and concretization maps
$\alpha:C\ra A$ and $\gamma:A \ra C$ such that,
for all $a$ and $c$,
$\alpha(c) \leq_A a \Lra c \leq_C \gamma(a)$. 
A GC is a GI when $\alpha\comp\gamma=\lambda x.x$.
It is well known that a join-preserving $\alpha$ uniquely 
determines, by adjunction,  $\gamma$
as follows: $\gamma(a) = \vee\{c\in C~|~\alpha(c) \leq_A a\}$; 
conversely, a meet-preserving $\gamma$ uniquely determines, by adjunction, $\alpha$ as follows:
$\alpha(c) = \wedge\{a\in A~|~ c\leq_C \gamma(a)\}$.

Let $f:C\ra C$ be some 
concrete monotone function---for simplicity, 
we consider 1-ary functions---and let
 $\ok{f^\sharp:A \ra A}$ be a corresponding monotone abstract function
 defined on some abstraction $A$ related to $C$ by a GC. Then,
$\ok{f^\sharp}$ is a correct abstract interpretation of $f$ on $A$
when $\ok{\alpha \circ f \sqsubseteq f^\sharp\circ \alpha}$ holds, where $\sqsubseteq$
denotes the pointwise ordering between functions.
Moreover, the abstract function
$\ok{f^A \ud \alpha \circ f \circ \gamma: A\rightarrow A}$ is called the best
correct approximation of $f$ on $A$ because any abstract
function
$\ok{f^\sharp}$ is correct iff $\ok{f^A \sqsubseteq f^\sharp}$.
Hence, for any $A$, 
$\ok{f^A}$ plays the role of the 
best possible approximation of $f$ on the abstraction $A$.

\subsection{Store Abstractions}\label{store-abstractions-sec}

As usual in abstract interpretation \cite{CC77}, a 
store property is modeled by some  abstraction $\Store^\sharp$ of $\wp(\Store)$ which is 
formalized through a Galois connection:
$$(\astore,{\tuple{\wp(\Store),\subseteq}},\tuple{\Store^\sharp,\leq},\gstore).$$

Given a program $P$, 
when $\Store^\sharp$ is viewed as an abstraction of $\tuple{\wp(\Store_P),\subseteq}$ 
we emphasize it by adopting the notation $\Store_P^\sharp$.   
A store abstraction $\Store_P^\sharp$ also induces a
state
abstraction $\State_P^\sharp \ud \Store_P^\sharp\times P$ and, in turn, a trace abstraction defined by
$\Trace_P^\sharp \ud (\State_P^\sharp)^*$.

\subsubsection{Nonrelational Abstractions}\label{abstract-lifting}
Nonrelational store abstractions (i.e., relationships between program variables are not 
taken into account) 
can be easily designed by a standard pointwise lifting of some value abstraction.  
Let $\Valueush$ be an abstraction of sets of possibly undefined values in $\wp(\Valueu)$ as formalized by a Galois connection
$$(\avalue,{\tuple{\wp(\Valueu),\subseteq}},\tuple{\Valueush,\leq_{\Valueush}},\gvalue).$$
The abstract domain $\Valueush$ induces a nonrelational 
store abstraction 
$$\rho^\sharp\in \Storeav \ud {\tuple{\Var \ra \Valueush,\sqsubseteq}}$$
where $\sqsubseteq$ is the pointwise ordering induced by $\leq_{\Valueush}$:
$\rho_1^\sharp \sqsubseteq \rho_2^\sharp$ iff for all $x\in \Var$, $\rho_1^\sharp(x)
\leq_{\Valueush} \rho_2^\sharp(x)$.
Hence, the bottom and top abstract stores
are, respectively, $\lambda x.\bot_{\Valueush}$ and $\lambda x.\top_{\Valueush}$. 
The abstraction map $\avalues:\wp(\Store) \ra \Storeav$ is defined as follows:
$$\avalues(S) \ud \lambda x.\avalue(\{\rho(x)\in \Valueu~|~\rho \in S\})$$
The corresponding concretization map $\gvalues: \Storeav \ra \wp(\Store) $ is defined,
as recalled in Section~\ref{intro_abs_int}, 
by adjunction  from the abstraction map $\ok{\avalues}$ 
and it is easy to check that it can be given as follows:
$$
\gvalues(\rho^\sharp) = \{\rho\in \Store~|~\forall x\in \Var.\: \rho(x) \in \gvalue(\rho^\sharp(x))\}.
$$
Let us observe that:
\begin{itemize}
\item[{\rm (i)}] $\avalues(\varnothing) = \lambda x.\avalue(\varnothing)  =\lambda x.\bot_{\Valueush}$ because
$\avalue(\varnothing) = \bot_{\Valueush}$ always holds in a GC;
\item[{\rm (ii)}] $\avalues(\{[\,]\}) = \lambda x.\avalue(\{\undeff\})$; 
\item[{\rm (iii)}] if $\gvalue(\bot_{\Valueush})=\varnothing$, $\rho^\sharp \in \Storeav$ and
$\rho^\sharp(x) = \bot_{\Valueush}$ then 
$\gvalues(\rho^\sharp) = \varnothing$;
\item[{\rm (iv)}] if $\gvalue(\bot_{\Valueush})=\{\undeff\}$ then
$\gvalues(\lambda x.\bot_{\Valueush}) = \{[\,]\}$.

\end{itemize}

\begin{example}[\textbf{The constant propagation abstraction}]\label{cp-ex}\rm
The constant propagation (see \cite{WZ91}) lattice $\tuple{\CP,\preceq}$ is depicted below. 
\begin{center}
    \begin{tikzpicture}[scale=1]
\small
     \tikzstyle{arrow}=[shorten >=-2pt, shorten <=-2pt]
      \draw (0,0) node[name=1] {{$\bot$}};
      \draw (-2.5,1) node[name=2] {{$\cdots$}};
      \draw (-2,1) node[name=3] {{$v_{-2}$}};
      \draw (-1,1) node[name=4] {{$v_{-1}$}};
      \draw (0,1) node[name=5] {{{$v_{0}$}}};
      \draw (1,1) node[name=6] {{$v_{1}$}};
      \draw (2,1) node[name=7] {{$v_2$}};
      \draw (2.5,1) node[name=8] {{$\cdots$}};
      \draw (0,2) node[name=9] {{$\top$}};

      \draw[semithick] (1) -- (3);
      \draw[semithick] (1) -- (4);
      \draw[semithick] (1) -- (5);
      \draw[semithick] (1) -- (6);
      \draw[semithick] (1) -- (7);
      \draw[semithick] (9) -- (3);
      \draw[semithick] (9) -- (4);
      \draw[semithick] (9) -- (5);
      \draw[semithick] (9) -- (6);
      \draw[semithick] (9) -- (7);

\end{tikzpicture}
\end{center}
where $\{v_i\}_{i\in \mathbb{Z}}$ is any enumeration of $\Valueu$, thus $\undeff$ is included.  
Abstraction ${\alpha_{cp}: \wp(\Valueu) \ra \CP}$ and concretization $\gamma_{cp}:\CP \ra \wp(\Valueu)$ functions are defined as follows:
\begin{align*}
 \alpha_{cp}(S) \ud 
 \begin{cases}
 	\bot & \text{if } S=\varnothing\\
     v_i & \text{if } S=\{v_i\}\\
     \top & \text{otherwise}
 \end{cases}
 &\qquad
\gamma_{cp}(a)\ud 
	\begin{cases}    
      \varnothing & \text{if } a=\bot\\
       \{v_i\} & \text{if } a=v_i \\
   \Valueu      & \text{if } a=\top
     \end{cases}
\end{align*}
and give rise to a GI $(\alpha_{cp}, \tuple{\wp(\Valueu),\subseteq},\tuple{\CP,\preceq},\gamma_{cp})$. 
The corresponding nonrelational 
store abstraction is denoted by $\CPst \ud \tuple{\Var \ra \CP,\dot{\preceq}}$, where
$\alpha_{\CP}: \wp(\Store) \ra \CPst$ and $\gamma_{\CP}:\CPst \ra \wp(\Store)$ denote the 
abstraction and concretization maps. For example, for $\Var=\{x,y,z,w\}$ and omitting the bindings $v/\undeff$ also in 
abstract stores,
we have that: 
\[
\begin{array}{l}
\alpha_{\CP} (\{ [x/2,y/\mathtt{foo},z/1], [x/2, y/\mathtt{bar}]\}) =
[x/2,y/\top, z/\top],
\\[5pt]
\gamma_{\CP} ([x/2,y/\top,w/\mathtt{foo}]) = 
\{\rho\in \Store ~|~ \rho(x) = 2, \rho(y)\in \Valueu, \rho(z)=\undeff, \rho(w)=\mathtt{foo}\}, \\[5pt]
\gamma_{\CP} ([x/2,y/\top,w/\bot]) = \varnothing. \qquad \qed
\end{array}
\]
\end{example}

\section{Hot Path Selection}\label{hp-sect}
A loop path is a sequence of program commands 
which is repeated in some execution of a program loop, together
with a store property which is valid at the entry of each command in the path. 
A loop path becomes \emph{hot} when, during the execution, 
it is repeated at least a fixed number $N$ of times. In a TJITC, hot path selection 
is performed by a loop path monitor that also records store properties (see, e.g., 
\cite{gal2009}).
Here, hot path selection is not operationally defined, it is instead 
semantically modeled as an abstraction map over program traces, i.e., 
program executions.

Given a program $P$ and therefore its trace semantics $\Trace_P$, 
 we first define a mapping $\sloop:\Trace_P 
\ra \wp(\Trace_P)$ that
returns all the loop paths in some execution trace of  $P$. 
More precisely, a loop path is 
a proper substring (i.e., a segment) $\tau$ of a program trace $\sigma$ such that: 
\begin{enumerate}
\item[(1)] the successor 
command in $\sigma$ of the last state in $\tau$ exists and coincides with 
the command~--~or its complement, when this is the last loop iteration~--~of the first
state in $\tau$; 
\item[(2)] there is no other such command within $\tau$ (otherwise
the sequence $\tau$ would contain multiple iterations);
\item[(3)] the last state of $\tau$ performs a backward jump in the program $P$. 
\end{enumerate}
To
recognize backward jumps, we consider a topological order on the control flow graph
of commands in $P$, denoted by $\lessdot$. This leads to the following formal
definition:
\begin{multline*}
    \sloop(\ab{\rho_{0}}{C_{0}} \cdots
    \ab{\rho_{n}}{C_{n}}) \ud
    \big\{\tuple{\rho_i,C_{i}}\tuple{\rho_{i+1},C_{i+1}} \cdots \tuple{\rho_j,C_j}~|~
    0\leq i \leq j < n,\, C_{i} \lessdot C_{j},\\
     suc(C_{j})=lbl(C_i),
       \forall k\in (i,j].\,
       C_{k} \not\in\{ C_{i},cmpl(C_{i})\}
       \big\}.
 \end{multline*}

\noindent
Let us remark that a loop path 
$$\tuple{\rho_i,C_{i}} \cdots \tuple{\rho_j,C_j}\in \sloop(\ab{\rho_{0}}{C_{0}} \cdots
    \ab{\rho_{n}}{C_{n}})$$ may contain some sub-loop path, namely
it may happen that $\sloop(\tuple{\rho_i,C_{i}} \cdots \tuple{\rho_j,C_j})\neq \varnothing$ so that
some commands $C_k$, with $k\in [i,j]$, may occur more than once in $\tuple{\rho_i,C_{i}} \cdots \tuple{\rho_j,C_j}$; for example,
this could be the case of a while loop whose body includes a nested while loop. 

We abuse notation by using $\astore$ to denote a
map $\astore: \Trace_P \ra \ok{\Trace_P^\sharp}$ which 
``abstracts'' a program trace $\tau$ into $\ok{\Trace^\sharp_P}$ by abstracting the sequence of stores occurring in $\tau$: 
\[
  \astore ( \tuple{\rho_0,C_0}\cdots \tuple{\rho_n,C_n}) \ud
  \tuple{\astore(\{\rho_0\}),C_0} \cdots \tuple{\astore(\{\rho_n\}),C_n}.
\]
Given a static integer parameter $N>0$, we define a
function $$\hot^N:\Trace_P \ra  \wp(\Trace_P^\sharp)$$
 which returns the set of $\Store^\sharp$-abstracted loop paths appearing at least
$N$ times in some program trace.  
In order to
count the number of times a loop path appears within a trace we need an
auxiliary function $\scount: \Trace_P^\sharp \times \Trace_P^\sharp \rightarrow \mathbb{N}$
such that $\scount(\sigma,\tau)$ yields the number of times an abstract path $\tau$ occurs in an abstract 
trace $\sigma$:
\begin{multline*}    
    \scount(\tuple{a_0,C_0} \cdots \tuple{a_n,C_n}, \tuple{b_0,C'_0} \cdots \tuple{b_m,C'_m})
    \ud \\
    \sum_{i=0}^{n-m} \begin{cases}
      1 & \text{if } \tuple{a_i,C_i} \cdots \tuple{a_{i+m},C_{i+m}} = \tuple{b_0,C'_0} \cdots \tuple{b_m,C'_m}\\
      0 & \text{otherwise}
    \end{cases}
    \end{multline*}     
Hence, $\hot^N$ can be defined as follows: 
\begin{align*}
    \hot^N(\sigma\equiv\ab{\rho_{0}}{C_{0}} \cdots
    \ab{\rho_{n}}{C_{n}}) \ud
    \big\{& \tuple{a_i,C_i} \cdots \tuple{a_j,C_j} ~|~
    \exists \tuple{\rho_i,C_i}\cdots\tuple{\rho_j,C_j}\in \sloop(\sigma) \text{ s.t. }\\
    & %
    \, \astore(\tuple{\rho_i,C_i}\cdots\tuple{\rho_j,C_j}) = \tuple{a_i,C_i} \cdots \tuple{a_j,C_j},\\
    & \!\, \scount(\astore(\sigma),\tuple{a_i,C_i} \cdots \tuple{a_j,C_j} )\geq N\big\}.
\end{align*}

\noindent
Finally, an abstraction map $\ahotn: \wp(\Trace_P) \ra \wp(\Trace_P^\sharp)$ 
collects the results of applying $\hot^N$ to a set of traces:
$$
\alpha_{hot}^N(T) \ud \bigcup_{\sigma \in T} \hot^N(\sigma).
$$

A $N$-hot path $hp$ in a program $P$ is therefore any
$hp \in \alpha_{hot}^N(\Trace_P)$ and
is  compactly denoted as $hp = \tuple{a_0,C_0,...,a_n,C_n}$. Let us observe that 
if the hot path corresponds to the body of some while loop then its first command $C_0$ is a conditional, namely 
$C_0$ is the Boolean guard of the while loop. 
We define  the successor function $\textit{next}$ for indices in a hot path $\tuple{a_0,C_0,...,a_n,C_n}$ as follows:
$\textit{next} \ud \lambda i.\ i = n\,?\, 0 : i+1$. 
For a $N$-hot path $\tuple{a_0,C_0,...,a_n,C_n}\in \ahotn(\Trace_P)$,
for any $i\in [0,n]$, if $C_i$ is a conditional
command $L_i: B_i \ra L_{\textit{next}(i)}$ then throughout the paper its
complement $C_i^c=cmpl(C_i)$ will be also denoted by $L_i : \neg B_i \ra L^{c}_{\textit{next}(i)}$.

\begin{example}\label{ex-zero}
\rm 
Let us consider the program $P$ in Example~\ref{example-init} and
a trivial 
one-point store abstraction $\Store^\sharp = \{\top\}$,
where all the stores are abstracted 
to the same abstract store $\top$, i.e., $\astore = \lambda S.\top$.
Here, we have two $2$-hot paths in $P$, that is, it turns out that 
$\alpha_{\mathit{hot}}^2 (\Trace_P)=\{hp_1,hp_2\}$ where: 
\begin{align*}
hp_1=\tuple{ & \top, C_1\equiv L_1:x\leq 20 \ra L_2,
\top, C_2\equiv L_2: x:=x+1 \ra L_3 ,\\
&\top,C_3^c\equiv L_3: \neg(x\%3 =0) \ra L_1};\\[7.5pt]
hp_2=\tuple{ & \top, C_1\equiv L_1:x\leq 20 \ra L_2, \top, C_2\equiv L_2: x:=x+1 \ra L_3 ,\\
&\top,C_3\equiv L_3: (x\%3 =0) \ra L_4, \top, C_4\equiv x:=x+3 \ra L_1}.
\end{align*}
Therefore, the hot paths $hp_1$ and $hp_2$ correspond, respectively, 
to the cases where the Boolean test $(x\%3 =0)$ fails and succeeds.  Observe that the
maximal sequence of different values assumed by the program variable $x$ is as follows: 
$$?\mapsto 0 \mapsto 1 \mapsto 2 \mapsto 3 \mapsto 6 \mapsto 7 \mapsto 8 \mapsto 9 \mapsto 12 \mapsto 13 \mapsto 14 \mapsto 15 \mapsto 18 \mapsto 19 \mapsto 20 \mapsto 21 \mapsto 24$$  
Hence, if $\sigma$ is the complete terminating trace of $P$ in Example~\ref{example-init} 
then it turns out that $\scount(\astore(\sigma),hp_1)=8$ 
and $\scount(\astore(\sigma),hp_2)=4$. 
\qed
\end{example}   
   
\section{Trace Extraction}\label{trace-ext-sec}

For any abstract store $a\in \Store^\sharp$, a corresponding 
Boolean expression denoted by $\cdo{guard}\: E_a$ $\in \BExp$ is defined (where the notation 
$E_a$ should hint at an expression which is induced by the abstract store $a$), whose semantics
is as follows: for any $\rho \in \Store$,
\begin{equation*}
\bsem\grasse{\cdo{guard}\:E_a}\rho \ud 
\begin{cases}
\textit{true} &\text{if } \rho \in \gstore(a)\\
\textit{false} &\text{if } \rho \not\in \gstore(a)\\
\end{cases}
\end{equation*}
In turn, we also have program actions $\cdo{guard}\:E_a$ $\in \mathbb{A}$ such that: 
\[
\asem\grasse{\cdo{guard}\:E_a}\rho \ud 
\begin{cases}
\rho & \text{if }\rho \in \gstore(a)\\
\bot & \text{if }\rho \not\in \gstore(a)
\end{cases}
\]

\begin{figure*}
\newdimen\largeur\settowidth{\largeur}{\quad}\advance\largeur by -\columnwidth \largeur-\largeur
\begin{center}
\begin{tikzpicture}[scale=0.41] 

\scriptsize
\draw (4,22.5) node{$\vdots$};
\draw[thick,->] (4,21.5) -- (4,20);
\draw[color=black!100] (4.75,20.5) node {{${L_0}$}};
\filldraw[thick,draw=black!100,fill=red!25] (4,18) -- (2,19) -- (4,20) -- (6,19) -- (4,18);
\draw[color=black!100] (4,19) node{{${B_0}$}};
\draw[thick,->] (6,19) -- (7.5,19);
\draw[color=black!100] (7,19.5) node {{${L_1^c}$}};
\draw[thick,->] (4,18) -- (4,16);
\draw[color=black!100] (4.75,16.5) node {{${L_1}$}};
\filldraw[thick,draw=black!100,fill=red!25] (2,14) rectangle +(4,2);
\draw[color=black!100] (4,15) node{{${A_1}$}};
\draw[thick,->] (4,14) -- (4,12.5);
\draw[color=black!100] (4.75,13) node {{${L_2}$}};
\draw (4,12) node{$\pmb{\vdots}$};
\draw[thick,->] (4,10.5) -- (4,9);
\draw[color=black!100] (4.75,9.5) node {{${L_n}$}};
\filldraw[thick,draw=black!100,fill=red!25] (2,7) rectangle +(4,2);
\draw[color=black!100] (4,8) node{{${A_n}$}};
\draw[thick,-] (4,7) -- (4,6) -- (0,6) -- (0,21);
\draw[thick,->] (0,21) -- (4,21);


\draw (18,28.5) node{$\pmb{\vdots}$};
\draw[thick,->] (18,27.5) -- (18,26);
\draw[color=black!100] (18.75,26.5) node {{${L_0}$}};
\filldraw[thick,draw=black!100,fill=blue!50] (18,24) -- (16,25) -- (18,26) -- (20,25) -- (18,24);
\draw[color=black!100] (18,25) node{{${\guard ~E_{a_0}}$}};
\draw[thick,-] (20,25) -- (26,25);
\draw[thick,->] (26,25) -- (26,24); \draw[color=black!100] (26.75,24.5) node {{${\overline{L_0}}$}};
\filldraw[thick,draw=black!100,fill=red!25] (26,22) -- (24,23) -- (26,24) -- (28,23) -- (26,22);
\draw[color=black!100] (26,23) node{{${B_0}$}};
\draw[thick,-] (26,22) -- (26,17); 
\draw[thick,->] (28,23) -- (29.5,23); 
\draw[color=black!100] (29,23.5) node {{${L_1^c}$}};
\draw[thick,->] (18,24) -- (18,22);
\draw[color=black!100] (18.75,22.5) node {{${\ell_0}$}};
\filldraw[thick,draw=black!100,fill=blue!50] (18,20) -- (16,21) -- (18,22) -- (20,21) -- (18,20);
\draw[color=black!100] (18,21) node{{${B_0}$}};
\draw[thick,->] (20,21) -- (21.5,21); \draw[color=black!100] (21,21.5) node {{${L_1^c}$}};
\draw[thick,->] (18,20) -- (18,18); \draw[color=black!100] (18.75,18.5) node {{${\bbl_1}$}};

\filldraw[thick,draw=black!100,fill=blue!50] (18,16) -- (16,17) -- (18,18) -- (20,17) -- (18,16);
\draw[color=black!100] (18,17) node{{${\guard ~E_{a_1}}$}};
\draw[thick,-] (20,17) -- (26,17); 
\draw[thick,->] (26,17) -- (26,16); \draw[color=black!100] (26.75,16.5) node {{${L_1}$}};
\filldraw[thick,draw=black!100,fill=red!25] (24,14) rectangle +(4,2);
\draw[color=black!100] (26,15) node{{${A_1}$}};
\draw[thick,->] (26,14) -- (26,12.5); \draw[color=black!100] (26.75,13) node {{${L_2}$}};
\draw (26,12) node {{$\pmb{\vdots}$}};

\draw[thick,->] (18,16) -- (18,14); \draw[color=black!100] (18.75,14.5) node {{${\ell_1}$}};
\filldraw[thick,draw=black!100,fill=blue!50] (16,12) rectangle +(4,2);
\draw[color=black!100] (18,13) node{{${A_1}$}};
\draw[thick,->] (18,12) -- (18,10.5); \draw[color=black!100] (18.75,11) node {{${\bbl_2}$}};
\draw (18,10) node{$\pmb{\vdots}$};
\draw[thick,->] (18,8.5) -- (18,7); \draw[color=black!100] (18.75,7.5) node {{${\bbl_n}$}};
\filldraw[thick,draw=black!100,fill=blue!50] (18,5) -- (16,6) -- (18,7) -- (20,6) -- (18,5);
\draw[color=black!100] (18,6) node{{${\guard ~E_{a_n}}$}};
\draw[thick,-] (20,6) -- (26,6); 
\draw[thick,->] (26,7) -- (26,5); \draw[color=black!100] (26.75,5.5) node {{${L_n}$}};
\filldraw[thick,draw=black!100,fill=red!25] (24,3) rectangle +(4,2);
\draw[color=black!100] (26,4) node{{${A_n}$}};
\draw[thick,-] (26,3) -- (26,0) -- (18,0); 

\draw[thick,->] (18,5) -- (18,3); \draw[color=black!100] (18.75,3.5) node {{${\ell_n}$}};
\filldraw[thick,draw=black!100,fill=blue!50] (16,1) rectangle +(4,2);
\draw[color=black!100] (18,2) node{{${A_n}$}};
\draw[thick,-] (18,1) -- (18,0) -- (14,0) -- (14,27);
\draw[thick,->] (14,27) -- (18,27);

\end{tikzpicture}
\end{center}
\caption{An example of trace extraction transform: on the left, a hot path $hp$ with
commands in pink (in black/white: loosely dotted) 
shapes; on the right, the corresponding trace transform $\extr_{hp}(P)$ 
with new commands in blue (in black/white: densely dotted) shapes.}\label{tet-fig}
\end{figure*}
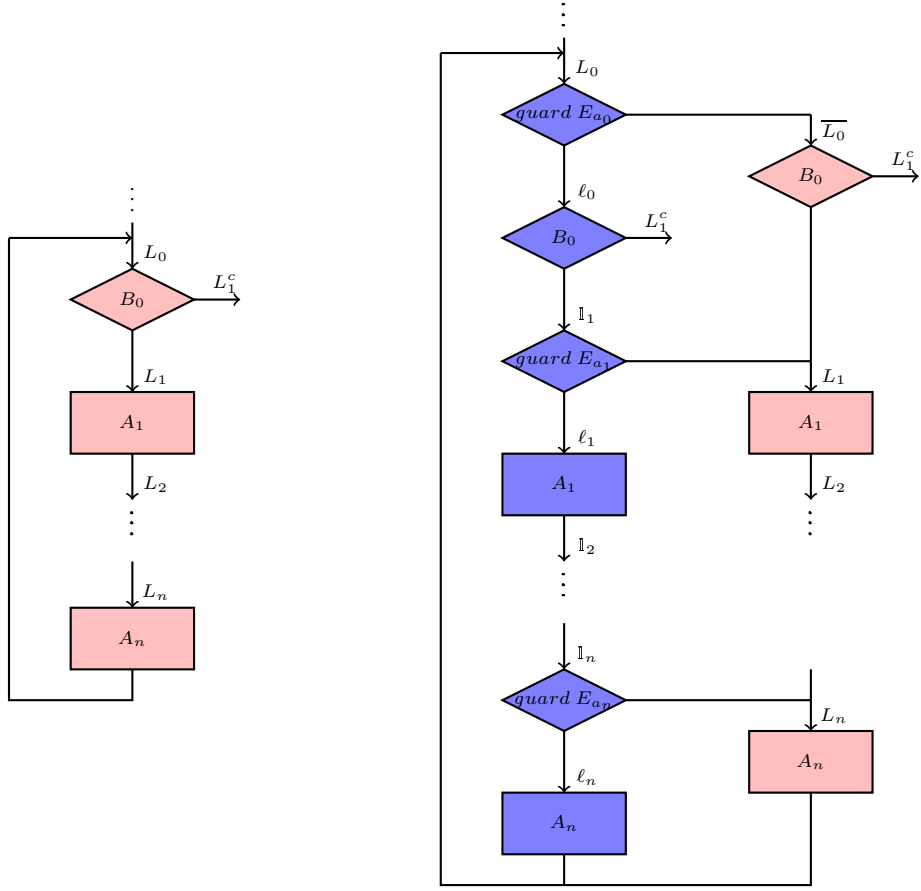

\noindent
Let $P$ be a program and $hp= \tuple{a_0,C_0,...,a_n,C_n}\in \ahotn(\Trace_P)$ be a
hot path on some store abstraction $\Store^\sharp$. 
We define a syntactic transform of $P$ 
where the hot path $hp$ is explicitly extracted from $P$. This is achieved by a
suitable relabeling of each command $C_i$ in $hp$ which is in turn preceded by
the conditional $\cdo{guard}\:E_{a_i}$ induced by the corresponding store property $a_i$. 
To this aim, we consider three \emph{injective} relabeling functions 
$$
\ell:[0,n] \ra \mathbb{L}_{1}\qquad\quad \bbl:[1,n] \ra \mathbb{L}_{2}\qquad\quad \overline{(\cdot)}: \mathbb{L} \ra  \overline{\mathbb{L}}
\eqno(*)
$$
where $\mathbb{L}_{1}$, $\mathbb{L}_{2}$ and $\overline{\mathbb{L}}$ are pairwise disjoint sets of 
fresh labels, so that 
$\mathit{labels}(P)\cap (\mathbb{L}_{1}\cup \mathbb{L}_{2}\cup \overline{\mathbb{L}})=\varnothing$. 
The transformed program $\extr_{hp}(P)$ for the hot path $hp$
is defined as follows and a graphical example of this  transform is 
depicted in Fig.~\ref{tet-fig}.

\begin{definition}[\textbf{Trace extraction transform}]
\label{tet-def}
\rm
The \emph{trace extraction transform} of $P$ for the hot path $hp=\tuple{a_0,C_0,...,a_n,C_n}$ is:
\begin{align*}
    &\extr_{hp}(P) \ud  P \smallsetminus \big(\{ C_{0}\} \cup \{ cmpl(C_{0})~|~cmpl(C_0)\in P \}\big)\\ 
     &\;\;\quad\quad\qquad \qquad\cup\{ \overline{L_0} : act(C_0) \ra L_1\} \cup \{ \overline{L_0} : \neg act(C_0) \ra L_1^c ~|~ cmpl(C_0)\in P\}
    \cup \stitch_{P}(hp)
\end{align*}
where the stitch of $hp$ into $P$ is defined as follows:  
\belowdisplayskip=-10pt
\begin{equation*}
  \begin{aligned}
    \stitch_{P}(hp) \ud  & \;\{ L_{0} : \cdo{guard}\: E_{a_0} \rightarrow
    \ell_0,\, L_{0} : \neg \cdo{guard}\:E_{a_0} \rightarrow
    \overline{L_0}\}\\   
    &\cup \{\ell_i : act(C_{i}) \rightarrow \bbl_{i+1}~|~
    i\in [0,n-1]\} \!\cup\! \{\ell_n : act(C_n) \rightarrow L_0\}\\
    &\cup \{\ell_i : \neg act(C_i) \rightarrow L_{\textit{next}(i)}^c ~|~ i\in [0,n],\, cmpl(C_{i})\in P \} \\
     &\cup \{\bbl_i : \cdo{guard}\:E_{a_i} \rightarrow \ell_{i},\, 
     \bbl_i : \neg \cdo{guard}\:E_{a_i} \rightarrow L_{i} ~|~
    i\in [1,n]\}. \qed
   \end{aligned}
\end{equation*}
\end{definition}

The new command $L_{0} : \cdo{guard}\: E_{a_0} \rightarrow \ell_0$ is therefore the entry conditional
of the stitched hot path $\stitch_{P}(hp)$, while any
command $C\in \stitch_{P}(hp)$ such that 
$suc(C) \in \mathit{labels}(P) \cup \ok{\overline{\mathbb{L}}}$ is a potential exit (or bail out) command 
of $\stitch_{P}(hp)$.

\begin{lemma}
If $P$ is well-formed then, for any hot path $hp$, $\extr_{hp}(P)$ is well-formed. 
\end{lemma}
\begin{proof} Recall that a program is well-formed when for any its conditional command it also
includes a unique complement conditional. It turns out that 
$\extr_{hp}(P)$ is  well-formed because: (1)~$P$ is well-formed; (2)~for each conditional in $P_{\mathit{new}}
= \extr_{hp}(P)\smallsetminus P =
\stitch_{P}(hp)$ $\cup$
$\{ \overline{L_0} : act(C_0) \ra L_1\} \cup \{\overline{L_0} : \neg act(C_0) \ra L_1^c~|~cmpl(C_0)\in P\}$ 
we also have
a unique complement conditional in $P_{\mathit{new}}$. 
Moreover, observe that if $P$ is deterministic then $\extr_{hp}(P)$ still is deterministic.
\end{proof}

Let us remark that the stitch of the hot path $hp$ into $P$ is always 
a linear sequence of different commands, namely, $\stitch_{P}(hp)$ 
does not contain loops nor join points. Furthermore, this happens 
even if the hot path $hp$ does contain some inner sub-loop. Technically, 
this is achieved as a consequence of the fact that the above relabeling functions $\ell$ and 
$\bbl$ are required to be 
injective.
Hence, even if some command $C$ occurs more than once inside $hp$, e.g., $C_i = C = C_j$ for some $i,j\in [0,n-1]$ with $i\neq j$, then 
these multiple occurrences of $C$ in $hp$ are transformed  into
differently labeled commands in $\stitch_{P}(hp)$, e.g., because $\ell_i \neq \ell_j$ and $\bbl_{i+1} \neq \bbl_{j+1}$. 

Let us now illustrate the trace extraction transform on a first simple example.  

\begin{example}\label{ex-one}
\rm 
Let us consider the program $P$ in Example~\ref{example-init} 
and the hot path $hp = \tuple{\top,C_1,\top,C_2,\top, C_3^c}$  in Example~\ref{ex-zero} (denoted there by $hp_1$), where
stores are abstracted to the trivial one-point abstraction $\Store^\sharp = \{\top\}$.
Here,  for any $\rho \in \Store$, we have that $\bsem\grasse{\cdo{guard}\:E_\top}\rho = 
\textit{true}$.
The trace extraction transform of $P$ w.r.t.\ $hp$ is therefore as follows:
$$
\extr_{hp}(P) =
     P \smallsetminus \{ C_{1}, C_{1}^c \} 
    \cup \{\overline{L_1} : x\leq 20 \ra L_2,\, \overline{L_1} : \neg (x \leq 20) \ra L_5\}
    \cup \stitch_{P}(hp)
$$
where    
    \begin{equation*}
    \begin{aligned}
    \stitch_{P}(hp)= &\;\{ H_0\equiv L_{1} : \cdo{guard}\: E_\top \rightarrow
    \ell_0,\, H_0^c\equiv L_{1} : \neg \cdo{guard}\:E_\top \rightarrow
    \overline{L_1}\} \\
    &\cup \{ H_1\equiv\ell_0 : x\leq 20 \ra \bbl_1,\, H_1^c \equiv\ell_0 : \neg(x\leq 20) \ra L_5\}\\
    &\cup \{ H_2 \equiv\bbl_1: \cdo{guard}\: E_\top \ra \ell_1,\, 
    H_2^c \equiv\bbl_1: \neg \cdo{guard}\: E_\top \ra L_2\}\\
    &\cup \{ H_3 \equiv \ell_1: x:= x+1 \ra \bbl_2\}\\
    &\cup \{ H_4\equiv\bbl_2: \cdo{guard}\: E_\top \ra \ell_2,\, H_4^c\equiv\bbl_2: \neg \cdo{guard}\: E_\top \ra L_3\}\\
    &\cup \{H_5\equiv \ell_2: \neg(x\% 3 =0) \ra\! L_1, H_5^c \equiv\ell_2: (x\%3 =0) \ra\! L_4\}.
   \end{aligned}
\end{equation*}

\smallskip
\noindent
The flow graph of $\extr_{hp}(P)$  is depicted in Figure~\ref{fig-ex}, while 
a higher level representation 
using while-loops and gotos is as follows:

\medskip
\vbox{
{
$x:=0$;

$L_1\!:$ \While{$\guard~ E_\top$}{
\Indp\Indp \lIf{$\neg(x\leq 20)$}{\KwSty{goto} $L_5$}

\lIf{$\neg \guard~ E_\top$}{\KwSty{goto} $L_2$}

  $x:=x+1$;

 \lIf{$\neg \guard~ E_\top$}{\KwSty{goto} $L_3$}

 \lIf{$(x\% 3 = 0)$}{\KwSty{goto} $L_4$}
 }
 \lIf{$\neg(x\leq 20)$}{\KwSty{goto} $L_5$}
 
	$L_2\!: x:=x+1$;
 
 	$L_3\!:$ \lIf{$\neg (x\%3 =0)$}{\KwSty{goto} $L_1$}

	$L_4\!: x:=x+3$;

	\KwSty{goto} $L_1$;

    $L_5\!: \KwSty{skip}$; \qquad\qed
}
}

\end{example}

\begin{figure*}
\newdimen\largeur\settowidth{\largeur}{\quad}\advance\largeur by -\columnwidth \largeur-\largeur
\begin{center}
\begin{tikzpicture}[scale=0.41] 

\scriptsize

\draw[thick,|->] (18,30) -- (18,29);
\draw[color=black!100] (18.75,29.5) node {{${L_0}$}};
\filldraw[thick,draw=black!100,fill=red!25] (16,28) rectangle +(4,1);
\draw[color=black!100] (18,28.5) node{{$x:=0$}}; \draw[color=black!100] (15.25,28.5) node {{\color{red}\boldmath{$C_0$}}};
\draw[thick,->] (18,28) -- (18,26);
\draw[color=black!100] (18.75,26.5) node {{${L_1}$}};
\filldraw[thick,draw=black!100,fill=blue!50] (18,24) -- (16,25) -- (18,26) -- (20,25) -- (18,24);
\draw[color=black!100] (18,25) node{{${\guard ~E_{\top}}$}}; 
\draw[color=black!100] (15.5,25.5) node {{\color{blue}\boldmath{$H_0$}}}; 
\draw[color=black!100] (15.5,24.5) node {{\color{blue}\boldmath{$H_0^c$}}};

\draw[thick,-] (20,25) -- (26,25);
\draw[thick,->] (26,25) -- (26,24); \draw[color=black!100] (26.75,24.5) node {{${\overline{L_1}}$}};
\filldraw[thick,draw=black!100,fill=red!25] (26,22) -- (24,23) -- (26,24) -- (28,23) -- (26,22);
\draw[color=black!100] (26,23) node{{$x\leq 20$}}; \draw[color=black!100] (23.5,23.5) node {{\color{red}\boldmath{$C_1$}}}; 
\draw[color=black!100] (23.5,22.5) node {{\color{red}\boldmath{$C_1^c$}}};

\draw[thick,->] (26,22) -- (26,19); 
\draw[thick,-] (28,23) -- (34,23); \draw[thick,->] (34,23) -- (34,20); 
\draw[color=black!100] (34.75,20.5) node {{${L_5}$}};
\filldraw[thick,draw=black!100,fill=red!25] (32,19) rectangle +(4,1); \draw[color=black!100] (34,19.5) node{{$\cd{skip}$}}; 
\draw[color=black!100] (31.25,19.5) node {{\color{red}\boldmath{$C_5$}}};
\draw[thick,-|] (34,19) -- (34,18); \draw[color=black!100] (34.75,18.5) node {{$\L$}};

\draw[thick,->] (18,24) -- (18,22);
\draw[color=black!100] (18.75,22.5) node {{${\ell_0}$}};
\filldraw[thick,draw=black!100,fill=blue!50] (18,20) -- (16,21) -- (18,22) -- (20,21) -- (18,20);
\draw[color=black!100] (18,21) node{{$x\leq 20$}}; 
\draw[color=black!100] (15.5,21.5) node {{\color{blue}\boldmath{$H_1$}}}; 
\draw[color=black!100] (15.5,20.5) node {{\color{blue}\boldmath{$H_1^c$}}};
\draw[thick,->] (20,21) -- (34,21);
\draw[thick,->] (18,20) -- (18,18); \draw[color=black!100] (18.75,18.5) node {{${\bbl_1}$}};

\filldraw[thick,draw=black!100,fill=blue!50] (18,16) -- (16,17) -- (18,18) -- (20,17) -- (18,16);
\draw[color=black!100] (18,17) node{{${\guard ~E_{\top}}$}};
\draw[color=black!100] (15.5,17.5) node {{\color{blue}\boldmath{$H_2$}}}; 
\draw[color=black!100] (15.5,16.5) node {{\color{blue}\boldmath{$H_2^c$}}};
\draw[thick,-] (20,17) -- (22,17) -- (22,20) -- (26,20); 
\draw[thick,->] (22,20) -- (26,20);
\filldraw[thick,draw=black!100,fill=red!25] (24,18) rectangle +(4,1); \draw[color=black!100] (26.75,19.5) node {{${L_2}$}};
\draw[color=black!100] (26,18.5) node{{$x:=x+1$}}; \draw[color=black!100] (23.25,18.5) node {{\color{red}\boldmath{$C_2$}}};
\draw[thick,->] (26,18) -- (26,15);
\filldraw[thick,draw=black!100,fill=red!25] (26,13) -- (24,14) -- (26,15) -- (28,14) -- (26,13);
\draw[color=black!100] (26,14) node{{$x\%3 =0$}}; \draw[color=black!100] (26.75,15.5) node {{${L_3}$}};
\draw[color=black!100] (23.5,14.5) node {{\color{red}\boldmath{$C_3$}}}; 
\draw[color=black!100] (23.5,13.5) node {{\color{red}\boldmath{$C_3^c$}}};

\draw[thick,->] (26,13) -- (26,10);
\filldraw[thick,draw=black!100,fill=red!25] (24,10) rectangle +(4,1);
\draw[color=black!100] (26,10.5) node{{$x:=x+3$}}; \draw[color=black!100] (26.75,11.5) node {{${L_4}$}}; \draw[color=black!100] (23.5,10.5) node {{\color{red}\boldmath{$C_4$}}};
\draw[thick,-] (26,10) -- (26,9) -- (38,9) -- (38,27);
\draw[thick,->] (38,27) -- (18,27);

\draw[thick,->] (18,16) -- (18,14); \draw[color=black!100] (18.75,14.5) node {{${\ell_1}$}};
\filldraw[thick,draw=black!100,fill=blue!50] (16,13) rectangle +(4,1);
\draw[color=black!100] (18,13.5) node{{$x:=x+1$}}; \draw[color=black!100] (15.25,13.5) node {{\color{blue}\boldmath{$H_3$}}};
\draw[thick,->] (18,13) -- (18,11); \draw[color=black!100] (18.75,11.5) node {{${\bbl_2}$}};
\draw[thick,-] (20,10) -- (22,10) -- (22,16);  
\draw[thick,->] (22,16) -- (26,16);

\filldraw[thick,draw=black!100,fill=blue!50] (18,9) -- (16,10) -- (18,11) -- (20,10) -- (18,9);
\draw[color=black!100] (18,10) node{{${\guard ~E_{\top}}$}};
\draw[color=black!100] (15.5,10.5) node {{\color{blue}\boldmath{$H_4$}}}; 
\draw[color=black!100] (15.5,9.5) node {{\color{blue}\boldmath{$H_4^c$}}};

\draw[thick,-] (20,6) -- (23,6) -- (23,12); \draw[thick,->] (23,12) -- (26,12);

\draw[thick,->] (18,9) -- (18,7); \draw[color=black!100] (18.75,7.5) node {{${\ell_2}$}};
\filldraw[thick,draw=black!100,fill=blue!50] (18,5) -- (16,6) -- (18,7) -- (20,6) -- (18,5);
\draw[color=black!100] (18,6) node{{$\!\neg(x\%3 \!=\! 0)$}};
\draw[color=black!100] (15.5,6.5) node {{\color{blue}\boldmath{$H_5$}}}; 
\draw[color=black!100] (15.5,5.5) node {{\color{blue}\boldmath{$H_5^c$}}};

\draw[thick,-] (18,5) -- (18,4) -- (14,4) -- (14,27);
\draw[thick,->] (14,27) -- (18,27);

\end{tikzpicture}
\end{center}
\caption{The flow graph of the trace extraction transform $\extr_{hp}(P)$ in Example~\ref{ex-one}, 
where commands of $\stitch_P(hp)$
are in blue (in black/white: densely dotted) shapes, while 
commands of the source program $P$ are in pink (in black/white: loosely dotted) 
shapes.}\label{fig-ex}
\end{figure*}
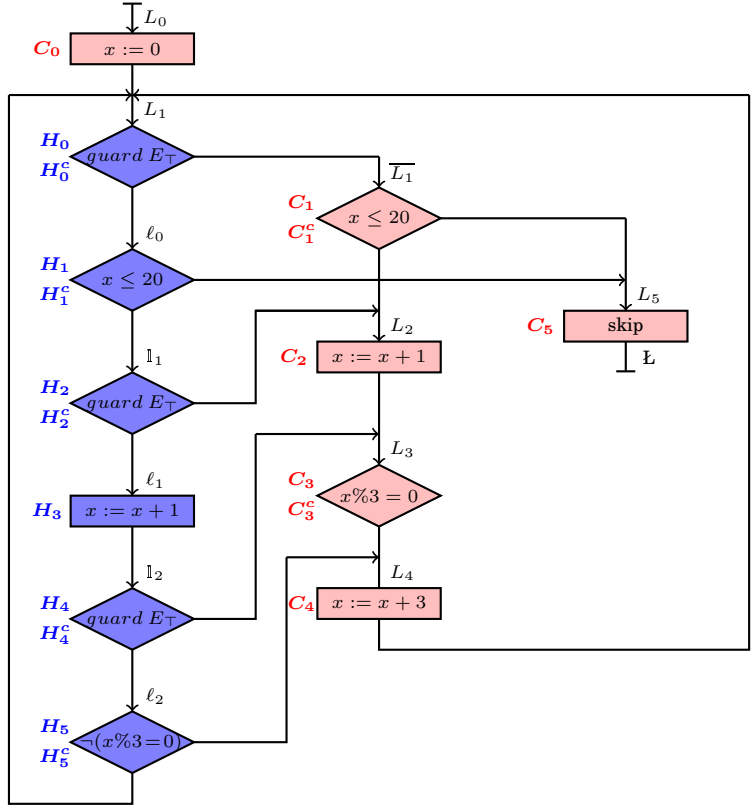

\section{Correctness}
\label{obs}

As advocated by
\citeN[par.~3.8]{cousot2002systematic},
correctness of dynamic program transformations and optimizations 
should be defined with respect to some observational
abstraction of program trace semantics: a dynamic program transform is correct when, at some
level of abstraction, the observation of the execution of the subject program
is equivalent to the observation of the execution of the transformed/optimized program. 

\subsubsection*{\textbf{Store Changes Abstraction}}
The approach by \citeN{palsberg} to tracing compilation basically relies on a
notion of correctness that requires the same \emph{store changes} to happen in both the
transformed/optimized program and the original program. This can be easily encoded 
by an observational
abstraction $\alpha_{\sch} : \wp(\Trace) \rightarrow \wp(\Store^{*})$ 
of partial traces that observes store changes in execution traces:  
\begin{equation*} \label{obs_abstraction}
  \begin{aligned}
    &\sch : \Trace \rightarrow \Store^{*} \\
    &\sch(\sigma) \ud\!
    \begin{cases}
      \varepsilon & \!\text{if } \sigma = \varepsilon\\
      \rho & \!\text{if } \sigma = \tuple{\rho,C}\\
      \sch (\tuple{\rho, C_1} \sigma') & 
      \!\text{if } \sigma = \tuple{\rho,C_0}\tuple{\rho, C_1} \sigma' \\
      \rho_0 \sch(\tuple{\rho_1, C_1} \sigma') & 
      \!\text{if }  \sigma =\! \tuple{\rho_0,C_0}\tuple{\rho_1, C_1} \sigma' ,
      \rho_0 \neq \rho_1
    \end{cases}    \\
    &\alpha_{\sch}(T) \ud \set{\sch(\sigma)}{\sigma \in T} 
  \end{aligned}
\end{equation*}

\noindent
Since the function $\alpha_{\sch}$ obviously preserves arbitrary set unions, as recalled in
Section~\ref{intro_abs_int}, it admits a 
right adjoint $\gamma_{\sch} :\wp(\Store^*) \ra \wp(\Trace)$ defined as  
$\gamma_{\sch}(S) \ud \cup \set{T \in
\wp(\Trace)}{\alpha_{\sch}(T) \subseteq S}$, that gives rise
to a GC $(\alpha_{\sch}, {\tuple{\wp(\Trace),\subseteq}},$ 
${\tuple{\wp(\Store^*),\subseteq}}, \gamma_{\sch})$. By a slight abuse of notation, 
$\alpha_{\sch}$ is also used as an abstraction of the partial trace
semantics of a given program $P$, that is,  $\alpha_{\sch}: \wp(\Trace_P) \rightarrow \wp(\Store_P^{*})$, 
which, clearly, gives rise to a corresponding GC $(\alpha_{\sch}, {\tuple{\wp(\Trace_P),\subseteq}},$ 
${\tuple{\wp(\Store_P^*),\subseteq}}, \gamma_{\sch})$.

\subsubsection*{\textbf{Output Abstraction}}
The store changes abstraction $\alpha_{\sch}$ may be too strong in practice.
This can be 
generalized to any observational abstraction of execution traces $\alpha_{o}:  \tuple{\wp(\Trace),\subseteq} \ra\tuple{A,\leq_A}$
(which gives rise to a GC).  
As a significant example, one may consider an output abstraction 
that demands to have the same 
stores (possibly restricted to some subset of program variables) 
only at some specific output points. For example, in
a language with no explicit output
primitives, as that considered by \citeN{palsberg}, 
one could be interested just in the final
store of the program (when it terminates), or in the entry and exit stores 
of any loop containing an extracted hot path. If we consider a language including a 
distinct primitive command ``$\cdo{put}~\cX$'' that ``outputs'' the value of
program variables ranging in some set $\cX$ 
then we may want to have the same stores for 
variables in $\cX$ at each
output point $\cdo{put}~\cX$. In this case, optimizations should preserve the same sequence of outputs, 
i.e.\ optimizations should not modify the
order of output commands. 
More formally, this can be achieved by adding
a further sort of actions: 
$\cdo{put}~\cX \in \mathbb{A}$, where $\cX\subseteq \Var$ is a set of
program variables. The semantics of $\cdo{put}~\cX$ 
obviously does not affect program stores, i.e.,
$\asem\grasse{\cdo{put}~\cX}\rho
\ud \rho$.  Correspondingly, if $\Store_{\cX}$ denotes stores on variables ranging in $\cX$ then 
the following output
abstraction $\alpha_{\out} : \wp(\Trace) \rightarrow \wp(\Store_{\cX}^{*})$
of partial traces observes program stores at output program points only:

\begin{equation*} 
  \begin{aligned}
    &\out : \Trace \rightarrow \Store_{\cX}^{*} \\
    &\out(\sigma) \ud
    \begin{cases}
      \varepsilon & \text{if } \sigma = \varepsilon\\
      \out (\sigma') & \text{if } \sigma = s\sigma' \wedge
      act(s) \neq \cdo{put}~\cX \\
      \rho_{|\cX}\out(\sigma') & \text{if } \sigma =
      \tuple{\rho, L : \cdo{put}~\cX \rightarrow L'}\sigma' 
    \end{cases}    \\
    &\alpha_{\out}(T) \ud \set{\out(\sigma)}{\sigma \in T} 
  \end{aligned}
\end{equation*}

\noindent
where $\rho_{|\cX}$ denotes the restriction of the store $\rho$ to variables 
in $\cX$. 
Similarly to $\alpha_{\sch}$, here again we have a GC 
$(\alpha_{o}, {\tuple{\wp(\Trace),\subseteq}},$ 
${\tuple{\wp(\Store_{\cX}^*),\subseteq}}, \gamma_{o})$.

\begin{example}[\textbf{Dead store elimination}]\rm
This approach based on a generic 
observational abstraction enables to prove the correctness of program optimizations
that are unsound in \citeN{palsberg}'s framework based on the store changes abstraction, such as
dead store elimination. For example, 
in a program fragment such as

\medskip
\noindent
{
\While{$(x\leq 0)$}{
   \noindent
   $z:=0$;\\
   $x:=x+1$;\\
   $z:=1$;
}
}

\medskip
\noindent
one can extract the hot path 
$hp=\tuple{x\leq 0, z:=0, x:=x+1, z:=1}$ (here we ignore store abstractions) 
and perform dead store elimination of the command $z:=0$
by optimizing $hp$ to $hp'=\tuple{x\leq 0,  x:=x+1, z:=1}$. As observed by 
\citeN[Section~4.3]{palsberg}, this is clearly unsound in bisimulation-based
correctness because this hot path optimization does not output bisimilar code. 
By contrast, this optimization can be made sound by choosing and then formalizing an observational
abstraction of program traces which requires to have the same stores 
at the beginning and at the exit of loops containing an
extracted hot path, while outside of hot paths one could still consider the store changes abstraction.  
\qed 
\end{example}

\subsubsection*{\textbf{Observational Abstraction}}
One can generalize the store changes abstraction $\alpha_{\sch}$ by considering any observational abstraction
$\alpha_{o}:  \tuple{\wp(\Trace),\subseteq} \ra\tuple{A,\leq_A}$ which is less precise (i.e., more approximate) than 
$\alpha_{\sch}$: this means that  for any  $T_1,T_2\in \wp(\Trace)$, if 
$\alpha_{\sch}(T_1)=\alpha_{\sch}(T_2)$ then $\alpha_{o}(T_1)=\alpha_{o}(T_2)$, or, equivalently, for any $T\in \wp(\Trace)$,  
$\gamma_{\sch}(\alpha_{\sch}(T))
\subseteq \gamma_o(\alpha_o (T))$. Informally, this means that $\alpha_{o}$ abstracts more information than $\alpha_{\sch}$. 
As an example, when considering programs with output actions, 
the following abstraction  $\alpha_{\osch} : \wp(\Trace) \rightarrow \wp(\Store_{\cX}^{*})$  
observes store changes at output program points only:

\begin{equation*} \label{obs_abstraction}
  \begin{aligned}
    &\!\osch : \Trace \rightarrow \Store_{\cX}^{*} \\
    &\!\osch(\sigma) \ud\!
    \begin{cases}
      \varepsilon & \!\!\text{if } \sigma = \varepsilon \text{~or~}\sigma = \tuple{\rho,C},\, act(C) \neq \cdo{put}~\cX\\
      \rho_{|\cX} & \!\!\text{if } \sigma = \!\tuple{\rho,C},\, act(C) =  \cdo{put}~\cX, \\
      \osch (\tuple{\rho, L_1: \cdo{put}~\cX \ra L_1'} \sigma') & 
      \!\!\text{if } \sigma = \!\tuple{\rho,C_0}\tuple{\rho, L_1: A_1 \ra L_1'} \sigma',\, act(C_0) = \cdo{put}~\cX \\
       \osch (\tuple{\rho, L_1: A_1 \ra  L_1'} \sigma') & 
      \!\!\text{if } \sigma = \!\tuple{\rho,C_0}\tuple{\rho, L_1: A_1 \ra L_1'} \sigma',\, act(C_0) \neq \cdo{put}~\cX \\
      \rho_{0|\cX} \osch(\tuple{\rho_1, C_1} \sigma') & 
      \!\!\text{if }  \sigma =\! \tuple{\rho_0,C_0}\tuple{\rho_1, C_1} \sigma',\,\rho_0 \neq \rho_1,\, act(C_0)= \cdo{put}~\cX\\
       \osch(\tuple{\rho_1, C_1} \sigma') & 
      \!\!\text{if }  \sigma =\! \tuple{\rho_0,C_0}\tuple{\rho_1, C_1} \sigma',\, \rho_0 \neq \rho_1,\, act(C_0)\neq \cdo{put}~\cX\\
    \end{cases}    \\
    &\alpha_{\osch}(T) \ud \set{\osch(\sigma)}{\sigma \in T} 
  \end{aligned}
\end{equation*}
Clearly, it turns out that $\alpha_{\osch}$ is more approximate than $\alpha_{\sch}$ since $\osch(\sigma)$ records 
a store change $\rho_0\rho_1$ only when the two contiguous subsequences of commands whose common stores are $\rho_0$ and $\rho_1$ contain
among them at least a $\cdo{put}$ command.

\subsection{Correctness of Trace Extraction}
It turns out that the observational correctness 
of the hot path extraction transform in Definition~\ref{tet-def} can be proved w.r.t.\ the 
observational abstraction $\alpha_{\sch}$ of store changes. 
    
\begin{theorem}[\textbf{Correctness of trace extraction}]\label{corr-th}
For any $P\in \Program$ and $hp\in \alpha^N_{\mathit{hot}}(\Trace_P)$, we have that
$\alpha_{\sch} (\tsem\grasse{\extr_{hp}(P)})=\alpha_{\sch} (\tsem\grasse{P})$.
\end{theorem}    

This is the crucial result concerning the correctness of our hot path extraction transform. 
We will show in Section~\ref{GPte-sec} (see Theorem~\ref{GP-corr-th}) 
that the correctness  of the hot path extraction strategy defined in \cite{palsberg}
 can be proved by a simple
adaptation of the proof technique that we will use here.   

In order to prove Theorem~\ref{corr-th}, we need to define some
suitable ``dynamic'' transformations of execution traces.  
Let us fix a hot path $hp=\tuple{a_0,C_0,...,a_n,C_n} 
\in \alpha^N_{\mathit{hot}}(\Trace_P)$ (w.r.t.\ some store abstraction)
and
let $\Php \ud \extr_{hp}(P)$ denote  the corresponding transform of $P$ given by Definition~\ref{tet-def}. 
We first define a mapping $\trf$ of the
execution traces of the program $P$ into execution traces of the transformed program $\Php$ 
that unfolds 
the hot path $hp$ (or any prefix of it) according to the hot path extraction strategy 
given by Definition~\ref{tet-def}: a function application $\trf(\tau)$ should replace any occurrence of the 
hot path $hp$ 
in the execution trace $\tau \in \Trace_P$ with its corresponding guarded and suitably relabeled  
path obtained through Definition~\ref{tet-def}. 
More precisely, Fig.~\ref{tr-def-figure} provides the definitions for the following two functions:
$$\trf:\Trace_P \ra \Trace_{\Php}\qquad
\trt:\Trace_P \ra (\State_P \cup \State_{\Php})^*$$

\noindent
Let us first describe how the trace transform $\trf$ works. 
A function application $\trf(s\sigma)$ on a trace $s\sigma$ of $P$---the superscript 
$\mathit{out}$ hints that the first state $s$ of the trace
$s\sigma$ is still \emph{outside} of the hot path $hp$ so that $\trf(s\sigma)$ could either enter into the transform of $hp$ or
remain outside of $hp$---triggers the unfolding   
of the hot path $hp$  in $\Php$ when the first state $s$ is such that: 
\begin{itemize}
\item[{\rm (i)}] $s=\tuple{\rho,C_0}$, where $C_0$ is the first command of $hp$; 
\item[{\rm (ii)}] the entry conditional $\guard\: E_{a_0}$ of $\stitch_P(hp)$ is satisfied in the store $\rho$ of the state 
$s=\tuple{\rho,C_0}$, that is, $\astore(\{\rho\}) \leq a_0$.
\end{itemize}
If the unfolding for the trace $\tuple{\rho,C_0}\sigma$ is actually started  by  applying 
$\trf (\tuple{\rho,C_0}\sigma)$ then: 
\begin{itemize}
\item[{\rm (iii)}] 
the first state $\tuple{\rho,C_0}$ is unfolded into the following sequence of two states of $\Php$:
$\tuple{\rho,L_0: \guard ~ E_{a_0} \ra \ell_0} \tuple{\rho,\ell_0: act(C_0) \ra \bbl_1}$; 
\item[{\rm (iv)}] in turn, 
 the unfolding of the residual trace $\sigma$ is carried on by applying $\trt(\sigma)$.
\end{itemize}

Let us now focus on the function $\trt$.
A function application $\trt(s\sigma)$---here the superscript 
$\mathit{in}$ suggests that we are currently \emph{inside} the hot path $hp$ so that $\trt(s\sigma)$
could either exit from the unfolding of $hp$ or advance with the unfolding of $hp$---carries on the unfolding of $hp$ as a trace in $\Php$ when 
the current state $s$ is such that: 
\begin{itemize}
\item[{\rm (i)}] $s=\tuple{\rho,C_i}$, where $i\in [1,n-1]$, meaning that the 
command $C_i$ is strictly inside $hp$, i.e., $C_i$ is different from the first command $C_0$ and the last
command $C_n$ of $hp$; 
\item[{\rm (ii)}] the guarded conditional 
$\guard\: E_{a_i}$ is satisfied in the store $\rho$ of the state $s=\tuple{\rho,C_i}$, that is, $\astore(\{\rho\}) \leq a_i$.
\end{itemize}
If one of these two conditions does not hold then  the trace transformation $\trt(\tuple{\rho,C_i}\sigma)$, after a suitable 
unfolding step for $\tuple{\rho,C_i}$, 
jumps back to the ``outside of $hp$'' modality  by progressing with $\trf(\sigma)$.

\begin{figure*}
\begin{mdframed} 
    \begin{equation*} 
  \begin{aligned}    
  &hp=\tuple{a_0,C_0,...,a_n,C_n} \text{~is a given hot path}\\[5pt]
  &\trf(\epsilon) \ud \epsilon  \\
     &\trf(s\sigma) \ud 
      \begin{cases}
        \tuple{\rho,L_0: \guard ~ E_{a_0} \ra \ell_0} \tuple{\rho,\ell_0: act(C_0) \ra \bbl_1}\, \trt(\sigma)&\\ 
      &\hspace{-11em} \text{if } s = \tuple{\rho, C_0},\, \astore(\{\rho\}) \leq a_0
      \\[5pt]
       \tuple{\rho,L_0: \neg \guard ~ E_{a_0} \ra \overline{L_0}} \tuple{\rho,\overline{L_0}: act(C_0) \ra L_1} \,\trf(\sigma) & \\
      &\hspace{-11em}\text{if } s = \tuple{\rho, C_0},\, \astore(\{\rho\}) \not\leq a_0
      \\[5pt]
      \tuple{\rho,L_0: \guard ~ E_{a_0} \ra \ell_0} \tuple{\rho,\ell_0: \neg act(C_0) \ra L_1^c} \,\trf(\sigma) & \\
      &\hspace{-11em} \text{if } s = \tuple{\rho, cmpl(C_0)},\,\astore(\{\rho\}) \leq a_0
      \\[5pt]
       \tuple{\rho,L_0: \neg \guard ~ E_{a_0} \ra \overline{L_0}} \tuple{\rho,\overline{L_0}: \neg act(C_0) \ra L_1^c} \,\trf(\sigma) & \\
      &\hspace{-11em} \text{if } s = \tuple{\rho, cmpl(C_0)},\, \astore(\{\rho\}) \not\leq a_0
      \\[5pt]
           s\cdot \trf(\sigma)  
      & \hspace{-11em}\text{otherwise }     
      \end{cases}   
      \\[5pt]
     &\trt(\epsilon) \ud \epsilon\\
     &\trt(s\sigma) \ud 
     \begin{cases}
        \tuple{\rho,\bbl_i: \guard ~ E_{a_i} \ra \ell_i} \tuple{\rho,\ell_i: act(C_i) \ra \bbl_{i+1}}\,\trt(\sigma) & \\
      &\hspace{-14.5em}\text{if } s = \tuple{\rho, C_i},\, i\in [1,n-1],\, \astore(\{\rho\}) \leq a_i
      \\[5pt]
       \tuple{\rho,\bbl_n: \guard ~ E_{a_n} \ra \ell_n} \tuple{\rho,\ell_n: act(C_n) \ra L_0}\,\trf(\sigma) &\\ 
      &\hspace{-14.5em}\text{if } s = \tuple{\rho, C_n},\, \astore(\{\rho\}) \leq a_n
      \\[5pt]
       \tuple{\rho,\bbl_i: \neg \guard ~ E_{a_i} \ra L_i} \tuple{\rho,C_i} \,\trf(\sigma) & \\
          &\hspace{-14.5em}\text{if } s = \tuple{\rho, C_i},\, i\in [1,n],\, \astore(\{\rho\}) \not\leq a_i
          \\[5pt]
        \tuple{\rho,\bbl_i: \guard ~ E_{a_i} \ra \ell_i} \tuple{\rho,\ell_i: \neg act(C_i) \ra L_{\textit{next}(i)}^c}\,\trf(\sigma) & \\
      &\hspace{-14.5em}\text{if } s = \tuple{\rho, cmpl(C_i)},\,i\in [1,n],\, \astore(\{\rho\}) \leq a_i
      \\[5pt]
       \tuple{\rho,\bbl_i: \neg \guard ~ E_{a_i} \ra L_i} \tuple{\rho,cmpl(C_i)} \,\trf(\sigma) & \\
      &\hspace{-14.5em}\text{if } s = \tuple{\rho, cmpl(C_i)},\, i\in [1,n],\, \astore(\{\rho\}) \not \leq a_i
      \\[5pt]
           s\cdot \trf(\sigma)  
      & \hspace{-14.5em}\text{otherwise }     
      \end{cases}    
    \end{aligned}
 \end{equation*}
 \end{mdframed} 
\caption{Definitions of $\trf$ and $\trt$.}
\label{tr-def-figure}
\end{figure*}

\begin{example}\label{ex-tr}
\rm
Consider the transform $\Php$ of Example~\ref{ex-one} for the program $P$ in Example~\ref{example-init} 
w.r.t.\ the hot path $hp = \tuple{\top,C_1,\top,C_2,\top, C_3^c}$. In particular, we refer to
the notation $H_i,H_i^c$ used to denote the
commands in the stitch of $hp$ into $P$. 
Consider the following trace fragment $\tau \in \Trace_P$: 
\begin{multline*}
\tau = \tuple{[x/3],C_0}\tuple{[x/0],C_1}\tuple{[x/0],C_2}\tuple{[x/1],C_3^c}\tuple{[x/1],C_1}\tuple{[x/1],C_2}\tuple{[x/2],C_3^c}\\
\tuple{[x/2],C_1}\tuple{[x/2],C_2}\tuple{[x/3],C_3}\tuple{[x/3],C_4}
\end{multline*}
Then, we have that the dynamic transformation $\trf(\tau)$ acts as follows:
\begin{align*}
\trf(\tau) &= \tuple{[x/3],C_0}\trf(\tau_{\suff{1}}) = \tuple{[x/3],C_0} \tuple{[x/0],H_0} \tuple{[x/0],H_1}\trt(\tau_{\suff{2}})\\
\trt(\tau_{\suff{2}})&= \tuple{[x/0],H_2}\tuple{[x/0],H_3} \trt(\tau_{\suff{3}}) \\
\trt(\tau_{\suff{3}}) &=\tuple{[x/1],H_4}\tuple{[x/1],H_5} \trt(\tau_{\suff{4}})\\
&\cdots\\
\trt(\tau_{\suff{9}}) &= \trt(\tuple{[x/3],C_3}\tuple{[x/3],C_4}) = \tuple{[x/3],H_4}\tuple{[x/3],H_5^c}\trf(\tuple{[x/3],C_4})\\
&= \tuple{[x/3],H_4}\tuple{[x/3],H_5^c}\tuple{[x/3],C_4}\trf(\epsilon)\\
&= \tuple{[x/3],H_4}\tuple{[x/3],H_5^c}\tuple{[x/3],C_4}
\end{align*}
Summing up, using the colors in the flow graph of $\Php$ in Fig.~\ref{fig-ex} and representing traces 
as sequences of commands only, we have that: 
\[
\tau \equiv {\color{red}\boxed{\boldmath{C_0}}} \ra  {\boxed{\color{red}\boldmath{C_1}}} \ra {\boxed{\color{red}\boldmath{C_2}}} \ra {\boxed{\color{red}\boldmath{C_3^c}}} \ra  
\boxed{\color{red}\boldmath{C_1}} \ra \boxed{\color{red}\boldmath{C_2}} \ra \boxed{\color{red}\boldmath{C_3^c}} \ra \boxed{\color{red}\boldmath{C_1}} \ra \boxed{\color{red}\boldmath{C_2}} \ra \boxed{\color{red}\boldmath{C_3}} \ra {\color{red}\boxed{\boldmath{C_4}}} 
\]
\begin{multline*}
\!\!\!\!\!\trf(\tau) \equiv {\color{red}\boxed{\boldmath{C_0}}} \ra  \boxed{\color{blue}\boldmath{H_0}}  \ra  \boxed{\color{blue}\boldmath{H_1}}  \ra  \boxed{\color{blue}\boldmath{H_2}}  \ra 
\boxed{\color{blue}\boldmath{H_3}}  \ra  \boxed{\color{blue}\boldmath{H_4}}  \ra  \boxed{\color{blue}\boldmath{H_5}}  \ra \boxed{\color{blue}\boldmath{H_0}}  \ra 
\boxed{\color{blue}\boldmath{H_1}} \ra \boxed{\color{blue}\boldmath{H_2}} \ra \\ 
\ra \boxed{\color{blue}\boldmath{H_3}} 
\ra  \boxed{\color{blue}\boldmath{H_4}} \ra \boxed{\color{blue}\boldmath{H_5}}  \ra
\boxed{\color{blue}\boldmath{H_0}}  \ra  \boxed{\color{blue}\boldmath{H_1}} \ra  \boxed{\color{blue}\boldmath{H_2}} \ra 
\boxed{\color{blue}\boldmath{H_3}}  \ra  \boxed{\color{blue}\boldmath{H_4}}  \ra  \boxed{\color{blue}\boldmath{H_5^c}}  \ra   {\color{red}\boxed{\boldmath{C_4}}\;}            
\end{multline*}
\noindent
where red boxes denote commands of $\tau$ and $\trf(\tau)$ outside of the hot path $hp$, black boxes with red commands denote commands of $\tau$ inside $hp$, 
while black boxes with blue commands denote commands of $\trf(\tau)$ in $\stitch_P(hp)$.
Hence, $\trf(\tau)$ carries out the unfolding of the hot path $hp$ for the execution trace $\tau$ of $P$, and therefore provides 
an execution trace of the transformed program $\Php$. 
\qed
\end{example}

It turns out that $\trf$ maps traces of $P$ into traces of $\Php$ and
does not alter store change sequences.

\begin{lemma}\label{lem-tr}
$\trf$ is well-defined and 
for any $\sigma \in \Trace_P$, $\sch (\trf(\sigma)) = \sch(\sigma)$. 
\end{lemma}
\begin{proof}
We first show that: (1)~$\trf$ is well-defined, i.e., 
for any $\sigma \in \Trace_P$, $\trf(\sigma) \in \Trace_{\Php}$, and 
(2)~for any $\sigma \in \Trace_{P}$, if $\mathit{cmd}(\sigma_0) \not\in \{C_0,cmpl(C_0)\}$ 
then $\trt(\sigma) \in \Trace_{\Php}$.
In order to prove these two points, it is enough an easy induction on the length
of the execution trace $\sigma$ and to observe that: 
\begin{enumerate}
\item[(i)] for the  first four clauses that
define $\trf(s\sigma)$ in Fig.~\ref{tr-def-figure} we have that   $\trf(s\sigma)=s's''\trf(\sigma)$
or $\trf(s\sigma)=s's''\trt(\sigma)$, where $s'$ is a guard command of $\Php$ and $s's''$ is in turn 
a legal sub-execution trace of  $\Php$; 
\item[(ii)] for the last clause that
defines $\trf(s\sigma)$ in Fig.~\ref{tr-def-figure} we have that $\mathit{cmd}(s) 
\not \in \{C_0,cmpl(C_0)\}$, hence 
$s$ is a legal state in $\Php$ and, in turn, $\trf(s\sigma)= s\cdot \trf(\sigma)$ 
is a trace of $\Php$; 
\item[(iii)] for the  clauses  1, 2 and 4 that
define $\trt(s\sigma)$ in Fig.~\ref{tr-def-figure} we have that   $\trt(s\sigma)=s's''\trt(\sigma)$
or $\trt(s\sigma)=s's''\trf(\sigma)$, where $s'$ is a guard command and $s''$ is an action command 
such that $s's''$ is 
a legal sub-execution trace of  $\Php$; 
\item[(iv)] for the  clauses  3 and 5 that
define $\trt(s\sigma)$ in Fig.~\ref{tr-def-figure} we have that   $\trt(s\sigma)=s's\:\trt(\sigma)$
where $s'$ is a guard command and $s's$ turns out to be 
a legal sub-execution trace of  $\Php$; 
\item[(v)] for the last clause that
defines $\trt(s\sigma)$ in Fig.~\ref{tr-def-figure} we have that $\mathit{cmd}(s) 
\not \in \{C_i,cmpl(C_i)~|~i\in [1,n]\}$; by hypothesis, 
$\mathit{cmd}(s) \not\in \{C_0,cmpl(C_0)\}$, so that $\mathit{cmd}(s)
\not \in \{C_i,cmpl(C_i)~|~i\in [0,n]\}$, hence $s$  is a legal state in $\Php$
and in turn $\trt(s\sigma)= s\cdot \trf(\sigma)$ 
is a trace of $\Php$;
\item[(vi)]
$\trt(s\sigma)$ is never recursively called by a function application
$\trf(s_0s\sigma)$ when $\mathit{cmd}(s) \in \{C_0,cmpl(C_0)\}$.    
\end{enumerate}
Then, it is immediate to check from the definitions in Fig.~\ref{tr-def-figure} that 
if $\trf(s\sigma) = s's'' \tau$ then 
$\mathit{store}(s) = \mathit{store}(s') = \mathit{store}(s'')$. Therefore,  
for any $\sigma \in \Trace_P$, we obtain that $\sch (\trf(\sigma)) = \sch(\sigma)$. 
\end{proof}

\begin{figure*}
\begin{mdframed} 
    \begin{equation*} 
  \begin{aligned}   
    &\hspace{-1em}hp=\tuple{a_0,C_0,...,a_n,C_n} \text{~is a given hot path}\\[5pt] 
  &\hspace{-1em}\rtr(\epsilon) \ud \epsilon  \\
   & \hspace{-1em}\rtr(s\sigma) \ud 
      \begin{cases} 
      \tuple{\mathit{store}(s),C_i}  & \text{if } \sigma=\epsilon,\,  act(s) \in \{\guard~E_{a_i}, \neg \guard~E_{a_i}\},\, i\in [1,n] 
      \\[5pt]
        \rtr(\sigma) 
      & \text{if } \sigma \neq \epsilon,\, act(s) \in \{\guard~E_{a_i}, \neg \guard~E_{a_i}\},\, i\in [1,n] 
      \\[5pt]
       \tuple{\rho,C_0} \rtr(\sigma) 
      &\text{if } s = \tuple{\rho, \overline{L_0}: act(C_0) \ra L_1}
      \\[5pt]
       \tuple{\rho,C_0^c} \rtr(\sigma) 
      &\text{if } s = \tuple{\rho, \overline{L_0}: \neg act(C_0) \ra L_1^c}
      \\[5pt]
       \tuple{\rho,C_i} \rtr(\sigma) 
      &\text{if } s = \tuple{\rho, \ell_i : act(C_i) \ra \bbl_{i+1}},\, i\in [1,n-1]
      \\[5pt]
       \tuple{\rho,C_i^c} \rtr(\sigma) 
      &\text{if } s = \tuple{\rho, \ell_i : \neg act(C_i) \ra  L_{\textit{next}(i)}^c},\,  i\in [1,n]
      \\[5pt]
       \tuple{\rho,C_n} \rtr(\sigma) 
      &\text{if } s = \tuple{\rho, \ell_n : act(C_n) \ra L_0}
      \\[5pt]
           s\cdot \rtr(\sigma)  
      &\text{otherwise }     
      \end{cases}   
     \end{aligned}
 \end{equation*}
 \end{mdframed} 
\caption{Definition of $\rtr$.}
\label{rtr-def-figure}
\end{figure*}

Vice versa, it is a simpler task to define a
reverse transformation  function $\rtr$ 
that ``decompiles'' an execution trace $\sigma$ of $\Php$ into an execution trace of $P$ by removing guarded commands
in $\sigma$, as generated by the hot path $hp$, and by mapping the relabeled commands of $hp$ in $\sigma$ 
back to their corresponding source 
commands of $hp$.  
This function $\rtr: \Trace_{\Php} \ra  \Trace_P $ is correctly defined 
by the clauses in Fig.~\ref{rtr-def-figure} and it preserves  store change sequences.
\begin{lemma}\label{lem-rtr}
$\rtr$ is well-defined and 
for any $\sigma \in \Trace_{\Php}$, $\sch (\rtr(\sigma)) = \rtr(\sigma)$. 
\end{lemma}
\begin{proof}
We show that $\rtr$ is well-defined, i.e., 
for any $\sigma \in \Trace_{\Php}$, $\rtr(\sigma) \in \Trace_{P}$. This follows by 
an easy induction on the length
of the execution trace $\sigma$ by observing that: 
\begin{enumerate}
\item[(i)] the  first clause that
defines $\rtr(s\sigma)$ in Fig.~\ref{rtr-def-figure} is an extremal base case where 
$s\sigma=s$ and the command action of $s$ is a guard command 
$\guard~E_{a_i}$ (or its complement); in this case, we simply retain the store of $s$ and 
pick the command $C_i$ of $P$. 

\item[(ii)] the  clause 2 of
$\rtr(s\sigma)$ in Fig.~\ref{rtr-def-figure} simply removes the states whose commands are some 
$\guard~E_{a_i}$; since  $\guard~E_{a_i}$ does not alter stores, this removal preserves the sequence
of store changes. 

\item[(iii)] the  clauses  3-7 of 
$\rtr(s\sigma)$ in Fig.~\ref{rtr-def-figure} map a state $s$ of $\Php$ whose command $H_i$ 
is a relabeled action $act(C_i)$ or $\neg act(C_i)$ of the hot path $hp$ to a corresponding state of $P$ that has
the same $\mathit{store}(s)$ and whose command is: $C_i$ for $act(C_i)$ and $C_i^c$ for $\neg act(C_i)$; 
here, we observe that since guards in $\sigma$ are removed, by induction, these definitions allow us to obtain that 
$s\sigma$ is mapped to a legal trace of $P$ that does not alter the sequence of store changes. 
\item[(iv)]  the  clause 8  
of  $\rtr(s\sigma)$ in Fig.~\ref{rtr-def-figure} states that if $s$ is already a state of $P$ then it is left
unchanged.  
\end{enumerate}
Hence, the above points also show that 
the sequence of store changes is not affected by $\rtr$, i.e.,   
for any $\sigma \in \Trace_{\Php}$, $\sch (\rtr(\sigma)) = \sch(\sigma)$. 
\end{proof}

\begin{example}\label{ex-tr2}
\rm
We carry on Example~\ref{ex-tr} by considering the following trace fragment $\sigma \in \Trace_{\Php}$, 
where the transformed program $\Php$ is in Example~\ref{ex-one}:
\begin{multline*}
\sigma = \tuple{[x/2],H_4}\tuple{[x/2],H_5}\tuple{[x/2],H_0}\tuple{[x/2],H_1}\tuple{[x/2],H_2}\tuple{[x/2],H_3}\tuple{[x/3],H_4}\\
\tuple{[x/3],H_5^c}\tuple{[x/3],C_4}\tuple{[x/6],C_1}
\end{multline*}
Here, the decompilation of $\sigma$ back into an execution trace of $P$ through $\rtr$ yields: 
\begin{align*}
\rtr(\sigma) &= \rtr (\sigma_{\suff{1}})=\tuple{[x/2],C_3^c}\rtr (\sigma_{\suff{2}}) =  \tuple{[x/2],C_3^c}\rtr (\sigma_{\suff{3}})\\
&= \tuple{[x/2],C_3^c}\tuple{[x/2],C_1}\rtr (\sigma_{\suff{4}})=\tuple{[x/2],C_3^c}\tuple{[x/2],C_1}\rtr (\sigma_{\suff{5}})\\
&=\tuple{[x/2],C_3^c}\tuple{[x/2],C_1} \tuple{[x/2],C_2} \rtr (\sigma_{\suff{6}})\\
&= \tuple{[x/2],C_3^c}\tuple{[x/2],C_1} \tuple{[x/2],C_2} \rtr (\sigma_{\suff{7}})\\
&=\tuple{[x/2],C_3^c}\tuple{[x/2],C_1} \tuple{[x/2],C_2}\tuple{[x/3],C_3}\rtr (\sigma_{\suff{8}}) \\
&= \tuple{[x/2],C_3^c}\tuple{[x/2],C_1} \tuple{[x/2],C_2}\tuple{[x/3],C_3}\tuple{[x/3],C_4}\rtr (\sigma_{\suff{9}})\\
&= \tuple{[x/2],C_3^c}\tuple{[x/2],C_1} \tuple{[x/2],C_2}\tuple{[x/3],C_3}\rtr (\sigma_{\suff{8}}) \\
&= \tuple{[x/2],C_3^c}\tuple{[x/2],C_1} \tuple{[x/2],C_2}\tuple{[x/3],C_3}\tuple{[x/3],C_4}\tuple{[x/6],C_1}
\end{align*}
Indeed, $\tuple{[x/2],C_3^c}\tuple{[x/2],C_1} \tuple{[x/2],C_2}\tuple{[x/3],C_3}\tuple{[x/3],C_4}\tuple{[x/6],C_1}$
is a well-defined execution trace of $P$. 
\qed
\end{example}

We are now in the position to prove Theorem~\ref{corr-th}.

\medskip
\textsc{Proof of Theorem~\ref{corr-th}.}
With an abuse of notation for $\rtr$, let us define two functions $\trh:\wp(\Trace_P)\ra \wp(\Trace_{\Php})$ 
and $\rtr: \wp(\Trace_{\Php}) \ra \wp(\Trace_P)$ which 
are the collecting
versions of $\trf$ and $\rtr$, that is, $\trh(T) \ud \{ \trf(\sigma)~|~ \sigma\in T\}$ and
$\rtr(T) \ud \{ \rtr(\sigma)~|~ \sigma\in T\}$. As consequences of the above
lemmata, we have the following properties.  
\begin{enumerate}
\item[(A)] $\alpha_{\sch}\circ \trh=\alpha_{\sch}$: by Lemma~\ref{lem-tr}.
\item[(B)] $\trh(\tsem\grasse{P}) \subseteq \tsem\grasse{\Php}$: because, 
by Lemma~\ref{lem-tr}, $\trf$ is well-defined.  
\item[(C)] $\alpha_{\sch}\circ \rtr=\alpha_{\sch}$: by Lemma~\ref{lem-rtr}.
\item[(D)] $\rtr(\tsem\grasse{\Php}) \subseteq \tsem\grasse{P}$: 
because, 
by Lemma~\ref{lem-rtr}, $\rtr$ is well-defined.  
\end{enumerate}
We therefore obtain: 
\begin{align*}
\alpha_{\sch}(\tsem\grasse{P})&=\text{\quad[By point~(A)]}\\
\alpha_{\sch}(\trh(\tsem\grasse{P})) &\subseteq\text{\quad[By point~(B)]}\\
\alpha_{\sch}(\tsem\grasse{\Php}) &= \text{\quad[By point~(C)]}\\
\alpha_{\sch}(\rtr(\tsem\grasse{\Php}))  &\subseteq\text{\quad[By point~(D)]}\\
\alpha_{\sch}(\tsem\grasse{P})  &
\end{align*}
and this closes the proof.
\qed

\subsection{Correctness of Hot Path Optimizations}\label{chpo-sec}

Guarded hot paths are a key feature of our tracing compilation 
model and are meant to be dynamically recorded by a hot path monitor. 
An abstract guard for
a command $C$ of some stitched hot path $\ok{\stitch_P(hp)}$ encodes a
property of program stores 
which is represented as an element of an abstract domain $\ok{\Store^\sharp}$ and is guaranteed to
hold at the entry of $C$. This information on program stores, as encapsulated by the abstract guards in $\ok{\stitch_P(hp)}$, 
can then be used in hot path optimizations, namely, to optimize the commands in $hp$. 

We follow a modular approach for proving the correctness of hot path
optimizations. 
A hot path optimization $O$ should optimize $P$ along some hot path $hp$ of $P$, by relying on 
the abstract store information recorded in $hp$, while leaving unchanged the commands outside of $hp$. 
Hence, in our framework, fixed  $P\in \Program$, an optimization 
$O$ is defined to be a program transform of the commands in $\stitch_{P}(hp)$, that is, 
$$O: \{ \stitch_{P}(hp)~|~ hp \in \ahotn(\Trace_P)\} \ra \Program$$

\noindent
where $\Program$ may allow new optimized expressions and/or actions introduced by $O$, as it will be 
the case
of type-specific additions $+_{\mathrm{Type}}$ in the type specialization optimization described in Section~\ref{type_specialization}. 
Let $P_{\neg hp} \ud \extr_{hp}(P) \smallsetminus \stitch_{P}(hp)$ denote the commands outside of the stitched hot path. 
Then, the corresponding full optimization $O_{\mathit{full}}$ 
of the whole program
$P$ w.r.t.\ the hot path $hp$ should extract and simultaneously optimize $hp$, namely, this
is defined by
$$O_{\mathit{full}}(P,hp) \ud P_{\neg hp} \cup O(\stitch_{P}(hp))$$
where $O_{\mathit{full}}(P,hp)$ is required to be a well-formed program, i.e., 
$O_{\mathit{full}}(P,hp)\in \Program$.
This full optimization $O_{\mathit{full}}(P,hp)$ 
has to be proved correct w.r.t.\ some observational abstraction $\alpha_{o}:\wp(\Trace_P) \ra A$ of program traces, 
which is assumed to 
be more abstract than the store changes abstraction $\alpha_{\sch}$ (cf.\ Section~\ref{obs}). Then, 
this full optimization is correct for $\alpha_{o}$ when:
$$\alpha_{o} (\tsem\grasse{O_{\mathit{full}}(P,hp)}) = \alpha_{o} (\tsem\grasse{P}).$$

\noindent
Since Theorem~\ref{corr-th} ensures that the unoptimized trace extraction transform
is already correct for 
the store changes abstraction $\alpha_{\sch}$, which is more precise than $\alpha_{o}$, the intuition is that
in order to prove the correctness of $O_{\full}$ w.r.t.\  $\alpha_{o}$, 
it is enough  to focus on the 
correctness of the optimization $O$ along the stitched hot path $\stitch_P(hp)$. 
This therefore leads to the following definition of correctness for a hot path optimization. 

\begin{definition}[\textbf{Correctness of hot path optimization}]\label{chpo-def}
\rm 
$O$ is \emph{correct} for the observational abstraction $\alpha_o$ if for any $P\in \Program$ and for 
any $hp \in \ahotn(\Trace_P)$, 
$\alpha_{o} (\tsem\grasse{O (\stitch_P(hp))}) =
\alpha_{o} (\tsem\grasse{\stitch_P(hp)})$.
\qed
\end{definition}

In order to prove that this correctness of a
hot path optimization implies the correctness of the corresponding full optimization, we define two functions
$$\tod: \Trace_{\stitch_P(hp)} \ra \Trace_{O (\stitch_P(hp))} \quad 
\tdo: \Trace_{O (\stitch_P(hp))} \ra \Trace_{\stitch_P(hp)}$$
which must be well-defined, i.e.\ they have to map well-formed traces into well-formed traces, and,
intuitively, encode the effect of optimizing (function $\tod$) and de-optimizing (function $\tdo$) execution
traces along a stitched hot path. 
Since $\Trace_{\stitch_P(hp)}\subseteq \Trace_{\extr_{hp}(P)}$ and $\Trace_{O (\stitch_P(hp))}
\subseteq \Trace_{O_{\full} (P,hp)}$, we then extend $\tod$ and $\tdo$ to 
two functions 
$$\tod_{\full}: \Trace_{\extr_{hp}(P)} \ra \Trace_{O_{\full} (P,hp)} \quad 
\tdo_{\full}: \Trace_{O_{\full} (P,hp)} \ra \Trace_{\extr_{hp}(P)}$$ which simply apply 
$\tod$ and $\tdo$ to maximal subtraces, respectively, in $\Trace_{\stitch_P(hp)}$ and   
$\Trace_{O (\stitch_P(hp))}$, while leaving unchanged the remaining states. Let us formalize this
idea. 
 If $\sigma \in \Trace_{\extr_{hp}(P)}$ is nonempty and, for some $k\in [0,|\sigma|)$, 
$\mathit{cmd}(\sigma_k) \in
\stitch_P(hp)$ then $\sigma_{[k,n_{st}]}$ denotes the maximal subtrace 
of $\sigma$ beginning at $\sigma_k$ which belongs to $\Trace_{\stitch_P(hp)}$, that is, the index $n_{st}\geq k$ is such
that: (1)~$\mathit{cmd}(\sigma_{n_{st}}) \in
\stitch_P(hp)$, (2)~if ${n_{st}} < |\sigma|-1$ then $\mathit{cmd}(\sigma_{{n_{st}}+1}) \not\in
\stitch_P(hp)$, (3)~for any $j\in [k,{n_{st}}]$,  $\mathit{cmd}(\sigma_j) \in
\stitch_P(hp)$. Analogously, if 
$\tau \in \Trace_{O_{\full} (P,hp)}$ is nonempty and $\mathit{cmd}(\tau_k) \in
O (\stitch_P(hp))$ then $\tau_{[k,n_{st}]}$ denotes the maximal subtrace 
of $\tau$ beginning at $\tau_k$ which belongs to $\Trace_{O (\stitch_P(hp))}$.
Then, the formal definition of $\tod_{\full}$ goes as follows:
\begin{equation*} 
  \begin{aligned}
     &\tod_{\full}(\sigma) \ud
    \begin{cases}
    \epsilon &\text{if }\sigma =\epsilon\\
    \sigma_0\tod_{\full}(\sigma_{1^{^{\!\shortrightarrow}}}) & \text{if }\sigma \neq \epsilon,\, \mathit{cmd}(\sigma_0) \not\in \stitch_P(hp)\\
    \tod(\sigma_{[0,n_{st}]}) 
    \tod_{\full}(\sigma_{{(n_{st}+1)}^{^{\!\shortrightarrow}}}) & \text{if }\sigma \neq \epsilon,\, \mathit{cmd}(\sigma_0) \in \stitch_P(hp)    
    \end{cases}  
  \end{aligned}
\end{equation*}
and analogously for $\tdo_{\full}$. Since $\tod$ and $\tdo$ are supposed to be well-defined,  it turns out that 
$\tod_{\full}$ and
$\tdo_{\full}$ are well-defined once we make the weak and reasonable assumption
that $\tod$ and $\tdo$ do not modify the entry (which is always $L_0$) and exit 
labels  of the stitched
hot path. This assumption, e.g., for $\tod$ can be formalized as follows: 
if $\sigma \in \Trace_{\stitch_P(hp)}$ and $\tod(\sigma)=\tau$ then
(i) if $lbl(\sigma_0)=L_0$ then $lbl(\tau_0)=L_0$; (ii) if $suc(\sigma_{|\sigma|-1}) =L' \not \in \mathit{labels}(P)$ 
then $suc(\tau_{|\tau|-1}) =L'$. 
In the following, 
$\tod_{\full}$ and $\tdo_{\full}$ are also used to denote their corresponding collecting functions defined on sets of traces. 

\begin{lemma}\label{lemma-corr}
Assume that 
$\alpha_{o} \circ \tod_{\full} = \alpha_{o} = \alpha_{o} \circ \tdo_{\full}$. If $O$ is correct for $\alpha_o$ then 
$O_{\mathit{full}}$ is correct for $\alpha_o$. 
\end{lemma}
\begin{proof}
We have that:  
\begin{align*}
\alpha_{o}(\tsem\grasse{O_{\full} (P,hp)})&=\text{\quad[By $\alpha_{o} \circ \tdo_{\full} = \alpha_{o}$]}\\
\alpha_{o}(\tdo_{\full}(\tsem\grasse{O_{\full} (P,hp)})) &\subseteq\text{\quad[Since $\tdo_{\full}$ is well-defined]}\\
\alpha_{o}(\tsem\grasse{\extr_{hp}(P)}) &= \text{\quad[By $\alpha_{o} \circ \tod_{\full}=\alpha_o$]}\\
 \alpha_{o}(\tod_{\full}(\tsem\grasse{\extr_{hp}(P)})) &\subseteq\text{\quad[Since $\tod_{\full}$ is well-defined]}\\
\alpha_{o}(\tsem\grasse{O_{\full} (P,hp)})  &
\end{align*}
Thus, $\alpha_{o}(\tsem\grasse{O_{\full} (P,hp)}) = \alpha_{o}(\tsem\grasse{\extr_{hp}(P)})$.
By Theorem~\ref{corr-th}, $\alpha_{\sch}(\tsem\grasse{\extr_{hp}(P)}) = \alpha_{\sch}(\tsem\grasse{P})$, 
so that, since  $\alpha_{\sch}$ is more precise than $\alpha_{o}$, $\alpha_{o}(\tsem\grasse{\extr_{hp}(P)}) = \alpha_{o}(\tsem\grasse{P})$, 
and, in turn, $\alpha_{o}(\tsem\grasse{O_{\full} (P,hp)})= \alpha_{o}(\tsem\grasse{P})$.
\end{proof}

We will see in Sections~\ref{type_specialization} and~\ref{cf-sec} two significant examples 
of hot path optimizations, namely, type specialization and constant folding.

\section{Type Specialization}
\label{type_specialization}

One key optimization for dynamic
languages like JavaScript and PHP is type specialization, that is, the use of type-specific
primitives in place of generic untyped operations whose runtime execution
can be costly. As a paradigmatic 
example, a generic
addition operation could be defined on more than one type, so that 
the runtime environment must
check the type of its operands and execute a different operation
depending on these types: this is the case of the addition operation
in JavaScript (see its runtime semantics in the ECMA-262 standard 
\cite[Section~12.7.3.1]{JS-ecma}) and of
the semantics of $+$ in our language as given  in Section~\ref{sem-sec}. 
Of course, type specialization avoids the overhead of dynamic type checking and dispatch of generic untyped operations. 
When a type is associated to each variable before the execution of a command in some hot path,  this type environment can be used to replace generic operations with type-specific primitives. 
In this section, we show that type specialization can be viewed
as a particular hot path optimization which can be proved correct according to our definition in Section~\ref{chpo-sec}.

\subsection{Type Abstraction}
Let us recall that the set of type names is 
$\Types = \{\topt, \Int,  \String, \Undef, \bott\}$, which
can  be viewed as the following finite lattice ${\tuple{\Types,\leqt}}$:

\begin{center}
    \begin{tikzpicture}[scale=0.9]
\small
     \tikzstyle{arrow}=[shorten >=-5pt, shorten <=-5pt]
      \draw (0,0) node[name=1] {{$\bott$}};
      \draw (-1.5,1) node[name=3] {{$\Int$}};
      \draw (0,1) node[name=5] {{{$\String$}}};
      \draw (1.5,1) node[name=6] {{$\Undef$}};
      \draw (0,2) node[name=7] {{$\topt$}};

      \draw[semithick,shorten >=-2.5pt, shorten <=-2.5pt] (1) -- (3);
      \draw[semithick,shorten >=-2.5pt, shorten <=-2.5pt] (1) -- (5);
      \draw[semithick,shorten >=-2.5pt, shorten <=-2.5pt] (1) -- (6);
      \draw[semithick,shorten >=-2.5pt, shorten <=-2.5pt] (3) -- (7);
      \draw[semithick,shorten >=-2.5pt, shorten <=-2.5pt] (5) -- (7);
      \draw[semithick,shorten >=-2.5pt, shorten <=-2.5pt] (6) -- (7);
         
\end{tikzpicture}
\end{center}

The
abstraction  $\atype: \wp(\Valueu) \ra \Types$ and concretization $\gtype: \Types \ra \wp(\Valueu)$ functions
are defined as follows: 
\[ 
\begin{array}{cc}
\atype(S) \ud 
\begin{cases}
\bott & \text{if } S=\varnothing\\
\Int & \text{if } \varnothing \neq S\subseteq \mathbb{Z}\\
\String & \text{if } \varnothing \neq S\subseteq \Char^*\\
\Undef & \text{if } \varnothing \neq S=\{\undeff\}\\
\topt & \text{otherwise}
\end{cases}
&\qquad
\gtype(T) \ud 
\begin{cases}
\varnothing & \text{if } T=\bott\\
\mathbb{Z} & \text{if } T=\Int\\
\Char^* & \text{if } T=\String\\
\{\undeff\} & \text{if } T=\Undef\\
\Valueu & \text{if } T=\topt
\end{cases}
\end{array}
\]
\noindent
Thus, $\atype(S)$ provides the smallest type in ${\tuple{\Types,\leqt}}$ for a set $S$ of values. In particular, 
given $v\in \Valueu$, $\atype(\{v\})$ coincides with $\type(v)$.
Following the
approach described in Section~\ref{abstract-lifting}, we then consider a simple nonrelational 
store abstraction for types
$$\Storet \ud \tuple{\Var \ra \Types,\dot{\leqt}}$$ 
where $\dot{\leqt}$ is the standard pointwise lifting of $\leqt$, 
so that $\lambda x.\bott$ and $\lambda x.\topt$ are, respectively, the bottom
and top abstract stores in $\Storet$. 
The abstraction and concretization maps $\astore: \wp(\Store) \ra \Storet$ and 
${\gstore: \Storet \ra \wp(\Store)}$ are defined as a straight instantiation of
the definitions in
Section~\ref{abstract-lifting}.

The abstract type semantics $\esemt: \Exp\ra \Storet \ra \Types$ of expressions  is defined as the best correct approximation 
of the  concrete collecting semantics $\esem:\Exp\ra \wp(\Store) \ra \wp(\Value)$  on the type abstractions $\Storet$ and $\Types$, i.e.,  
$$\esemt\grasse{E}\rho^t \ud \atype(\esem\grasse{E}\gstore(\rho^t)).$$
Hence, this definition leads to the following equalities:
\begin{align*}
&\esemt\grasse{v}\rho^t = \type(v) \\
&\esemt\grasse{x}\rho^t =\rho^t(x)\\
&\esemt\grasse{E_1 + E_2}\rho^t =
\begin{cases}    
      \bott & \text{if } \exists i.\, \esemt\grasse{E_i}\rho^t = \bott\\
      \esemt\grasse{E_1}\rho^t 
      & \text{else if } \esemt\grasse{E_1}\rho^t = \esemt\grasse{E_2}\rho^t \in \{\Int,\String\}\\
     \Undef & \text{else if } \forall i.\, \esemt\grasse{E_i}\rho^t < \topt\\
     \topt & \text{otherwise }
     \end{cases}
\end{align*}

\noindent
For instance, we have that:
\begin{multline*}
\esemt\grasse{x+y}[x/\String, y/\bott]=\atype(\esem\grasse{x+y}\varnothing)=\atype(\varnothing)=\bott\\
\end{multline*}
\vspace*{-30pt}
\begin{multline*}
\esemt\grasse{x + y}[x/\String, y/\String] = \atype(\esem\grasse{x+y}\{\rho~|~\rho(x),\rho(y)\in \Char^*\})=\\
=\atype(\Char^*)=\String,
\end{multline*}
\vspace*{-20pt}
\begin{multline*}
\esemt\grasse{x + y}[x/\Int, y/\String] = \atype(\esem\grasse{x+y}\{\rho~|~\rho(x)\in \mathbb{Z}, \rho(y)\in \Char^*\})=\\
\atype(\{\undeff\})= \Undef,
\end{multline*}
\vspace*{-20pt}
\begin{multline*}
\esemt\grasse{x+y}[x/\Int, y/\topt] = \atype(\esem\grasse{x+y}\{\rho~|~\rho(x)\in \mathbb{Z}, \rho(y)\in \Valueu\}) =\\ 
\atype(\mathbb{Z}\cup \{\undeff\})= \topt
\end{multline*} 
Being defined as best correct approximation, it turns out that 
the abstract type semantics $\esemt$ of expressions is correct by definition. 
\begin{corollary}\label{coro-correct}
If $\rho \in \gstore(\rho^t)$ then $\esem\grasse{E}\rho \in \esemt\grasse{E}\rho^t$. 
\end{corollary}

According to Section~\ref{trace-ext-sec}, 
for any abstract type store (that we also call type environment) $[x_i/ T_i~|~ x_i\in \Var]\in \Storet$ we consider a
corresponding Boolean action guard denoted by
$$\text{guard}\:x_0:T_0,\ldots, x_n: T_n\in \BExp$$ whose corresponding action
semantics is automatically 
induced, as defined in Section~\ref{trace-ext-sec}, 
by the Galois connection $(\astore, \wp(\Store),\Storet,\gstore)$:
for any $\rho\in \Store$, 
\begin{align*}
\asem\grasse{\text{guard}\:x_0:T_0,..., x_n: T_n}\rho &\ud
\begin{cases}
\rho & \!\!\text{if } 
\rho \in \gstore([x_i/ T_i~|~ x_i\in \Var])\\
\bot & \!\!\text{otherwise}
\end{cases}
\\
&=
\begin{cases}
\rho & \!\!\text{if } 
\forall i.\, \rho(x_i) \in \gtype(T_i)\\
\bot & \!\!\exists i.\, \rho(x_i) \not\in \gtype(T_i)
\end{cases}
\end{align*}

\noindent
For example, we have that:
\begin{align*}
&\asem\grasse{\text{guard}\:x:\String, y:\String} [x/\mathtt{foo}, y/\mathtt{bar}] = [x/\mathtt{foo}, y/\mathtt{bar}],\\
&\asem\grasse{\text{guard}\:x:\String, y:\topt} [x/\mathtt{foo}, y/3] = [x/\mathtt{foo}, y/3],\\
&\asem\grasse{\text{guard}\:x:\String, y:\topt} [x/1, y/3] = \bot,\\
&\asem\grasse{\text{guard}\:x:\String, y:\Undef} [x/\mathtt{foo}] = [x/\mathtt{foo}].
\end{align*}

\subsection{Type Specialization of Hot Paths}
Let us consider some hot path $hp=\tuple{\rho^t_0,C_0,\ldots,\rho^t_n,C_n}
\in$ $\ahotn(\Trace_P)$ on the type abstraction $\tuple{\Store_P^t,\dot{\leqt}}$,
where each $\rho^t_i$ is therefore a type environment for $P$.
Thus, in the transformed program $\extr_{hp}(P)$, 
the stitched hot path $\stitch_P(hp)$ contains $n+1$ typed guards, that, for any $i\in [0,n]$,  
we simply denote as $\guard~\rho^t_i$. Typed guards 
allow us to perform
type specialization of commands in the stitched hot path. In order to keep the notation simple,
we only focus on type specialization of addition operations occurring in assignments, while one could also
consider an analogous type specialization of Boolean comparisons in conditional commands. 
This is defined as a program transform 
that instantiates most type-specific addition operations in place of
generic untyped additions by exploiting the type information dynamically recorded 
by typed guards in $\stitch_P(hp)$. Note that if $C\in \stitch_P(hp)$ and 
$act(C)\equiv x:= E_1 + E_2$ then $C\equiv \ell_i:x:= E_1 + E_2 \ra L'$, for some
$i\in [0,n]$, where $L' \in \{ \bbl_{i+1}, L_0\}$.  
Let $\mathbb{C}^t$ denote the extended set of commands which includes 
type specific additions $+_{\Int}$ and $+_{\String}$ and,
in turn, let $\Program^t$ denote the possibly type-specialized programs with commands ranging in  $\mathbb{C}^t$.
The semantic function $\esem$ for expressions is then updated to 
type specific additions as follows:
\begin{equation*}
 \begin{aligned}
    \esem\grasse{E_{1} +_{\Int} E_{2}}\rho \ud &
     \begin{cases}
       \esem\grasse{E_1}\rho +_{\mathbb{Z}} \esem\grasse{E_2}\rho &
      \text{if } \type(\esem\grasse{E_i}\rho)=\Int\\
       \undeff & \text{otherwise} 
         \end{cases} \\
    \esem\grasse{E_{1} +_{\String} E_{2}}\rho \ud&
     \begin{cases}
       \esem\grasse{E_1}\rho \cdot \esem\grasse{E_2}\rho &
      \text{if } \type(\esem\grasse{E_i}\rho)=\String\\
       \undeff & \text{otherwise} 
         \end{cases}
\end{aligned}
\end{equation*}
Given a hot path $hp=\tuple{\rho^t_0,C_0,\ldots,\rho^t_n,C_n}$,
the type specialization function ${\ts_{hp}: \stitch_{P}(hp) \ra  \mathbb{C}^t}$ 
is defined  as follows:
\begin{equation*}
\begin{aligned}
      &\ts_{hp}(\ell_i :x:=E_1 + E_2 \!\ra\! L') \ud 
      \begin{cases}
      \ell_i:x:= E_1 +_{\Int} E_2 \!\ra\! L' & \!\!\text{if } \esemt\grasse{E_1 + E_2}\rho_i^t = \Int\\
      \ell_i:x:= E_1 +_{\String} E_2 \!\ra\! L' & \!\!\text{if } \esemt\grasse{E_1 + E_2}\rho_i^t = \String\\
      \ell_i:x:= E_1 + E_2 \ra L' & \!\!\text{otherwise } 
      \end{cases}  \\
       &\ts_{hp}(C) \ud  C \hspace*{17.5ex}\text{if }C \not\equiv \ell_i :x:=E_1 + E_2 \!\ra\! L' 
\end{aligned}
\end{equation*}      
Hence, if a typed guard  $\guard~\rho^t_i$ preceding a command $\ell_i :x:=E_1 + E_2 \!\ra\! L'$
allows us to derive abstractly on $\Storet$ that $E_1$ and $E_2$ have the same type ($\Int$ or $\String$) then 
the addition $E_1 + E_2$ is accordingly type specialized. This function allows us to define 
the hot path type specialization optimization
$$O^{\ts}: \{ \stitch_{P}(hp)~|~ hp \in \ahotn(\Trace_P)\} \ra \Program^t$$ simply by
$$O^{\ts}(\stitch_{P}(hp)) \ud \{\ts_{hp}(C)~|~ C\in \stitch_{P}(hp)\}.$$
In turn, as described in Section~\ref{chpo-sec}, this induces the full type specialization optimization 
$$O^{\ts}_{\mathit{full}}(P,hp) \ud \extr_{hp}(P) \smallsetminus \stitch_{P}(hp) \cup O^{\ts}(\stitch_{P}(hp)).$$
$O^{\ts}_{\mathit{full}}(P,hp)$ is also called \emph{typed trace extraction}
since it extracts and simultaneously 
type specializes a typed hot path $hp$ in a program $P$. 
The correctness of this program optimization can be proved for the store changes observational abstraction
by relying on Lemma~\ref{lemma-corr}. 

\begin{theorem}[\textbf{Correctness of typed trace extraction}]\label{corr-tte}
For any typed hot path $hp \in \ahotn(\Trace_P)$, we have that
$\alpha_{\sch} (\tsem \grasse{O^{\ts}_{\mathit{full}}(P,hp)})=\alpha_{\sch} (\tsem \grasse{P})$.
\end{theorem}
\begin{proof}
Let $td:\Trace_{O^{\ts}(\stitch_{P}(hp))} \ra \Trace_{\stitch_{P}(hp)}$ be the following type
de-specialization function, 
where $\mathrm{Type}$ is either $\Int$ or $\String$:
\begin{equation*} 
  \begin{aligned}
  td(\epsilon) &\ud \epsilon\\
     td(s\sigma) &\ud
    \begin{cases}
    \tuple{\rho,\ell_i :x:=E_1 + E_2\ra L'} 
     &\hspace*{0ex} \text{if } s=\tuple{\rho, \ell_i :x:=E_1 +_{\mathrm{Type}} E_2\ra L'},\\ 
      & \hspace*{2ex}  \type(\esem\grasse{E_1 + E_2}\rho) \neq  \mathrm{Type}\\
      \tuple{\rho,\ell_i :x:=E_1 + E_2\ra L'}\cdot td(\sigma) 
     &\hspace*{0ex} \text{if } s=\tuple{\rho, \ell_i :x:=E_1 +_{\mathrm{Type}} E_2\ra L'},\\
     &\hspace*{2ex} \type(\esem\grasse{E_1 + E_2}\rho) = \mathrm{Type}
     \\
      s\cdot td(\sigma)  &\hspace*{0ex}\text{otherwise}
    \end{cases}  
  \end{aligned}
\end{equation*}
Let us explain the first defining clause of $td(s\sigma)$, i.e.,
$s=\tuple{\rho, \ell_i :x:=E_1 +_{\mathrm{Type}} E_2\ra L'}$ and 
$\type(\esem\grasse{E_1 + E_2}\rho) \neq  \mathrm{Type}$.
These conditions can never hold in an inductive call of the function $td$: in fact, when $td(s\sigma)$ is recursively
called by $td(s's\sigma)$, we necessarily have that 
$s'=\tuple{\rho,\bbl_i: \guard~\rho^t_i\ra \ell_i}$, so that 
$\rho \in \gstore (\rho^t_i)$, and, in turn, by Corollary~\ref{coro-correct}, 
$\esem\grasse{E_1 + E_2}\rho \in \esemt\grasse{E_1 + E_2}\rho^t$, which implies $\type(\esem\grasse{E_1 + E_2}\rho) =  \mathrm{Type}$, which
is a contradiction.
Thus, the first defining clause of $td(s\sigma)$ 
only applies to type specialized traces in $\Trace_{O^{\ts}(\stitch_{P}(hp))}$
whose first state 
is $s=\tuple{\rho, \ell_i :x:=E_1 +_{\mathrm{Type}} E_2\ra L'}$: in this case, we necessarily have that $\sigma = \epsilon$, because $\asem\grasse{E_1 +_{\mathrm{Type}} E_2}\rho = \undeff$ so that $\ssem s = \varnothing$.
This clarifies the definition of $td$ in this particular case. Also, observe that in this case,   $\sch( td(s)) = \sch(s)$
trivially holds. 
In all the remaining cases, it is clear that $td$ maps type specialized traces into legal unspecialized traces of $\stitch_{P}(hp)$
since labels are left unchanged. Moreover, $\sch \circ\, td = \sch$ holds, in particular because in the
second defining clause of $td(s\sigma)$, the condition $\type(\esem\grasse{E_1 + E_2}\rho) =  \mathrm{Type}$ guarantees
that $\esem\grasse{E_1 + E_2}\rho = \esem\grasse{E_1 +_{\mathrm{Type}} E_2}\rho$. 

\noindent
On the other hand, we define a trace specialization function $sp: \Trace_{\stitch_{P}(hp)} \ra \Trace_{O^{\ts}(\stitch_{P}(hp))}$ as follows:
\begin{equation*} 
  \begin{aligned}
   &sp(\epsilon) \ud \epsilon\\
     &sp(\tuple{\mu_0,H_0}\!\cdots\! \tuple{\mu_k,H_k}) \!\ud\!
    \begin{cases}
    \tuple{\mu_0,\ts_{hp}(H_0)} 
     &\hspace*{-10.5ex} \text{if } \ts_{hp}(H_0) \equiv \ell_i : x:=E_1 \!+_{\mathrm{Type}}\! E_2 \ra L',\\ 
     	&\hfill\hspace*{-8ex}\mu_0 \not\in \gstore (\rho^t_i)\\[5pt]
      \tuple{\mu_0,\ts_{hp}(H_0)} \!\cdots\!  \tuple{\mu_k,\ts_{hp}(H_k)}
     &\hfill\text{otherwise } 
         \end{cases}  
  \end{aligned}
\end{equation*}

\noindent
Let us comment on this definition. If $\sigma \in \Trace_{\stitch_{P}(hp)}$ and $\sigma\neq \epsilon$ then it may happen
that the first state $\tuple{\mu_0,H_0}$ of $\sigma$ is such that the command $H_0$ is $\ell_i : x:=E_1 \!+\! E_2 \ra L'$
and, since  $\esemt\grasse{E_1 + E_2}\rho_i^t =\mathrm{Type}$ ($\Int$ or $\String$), $H_0$ is type specialized to
$\ts_{hp}(H_0) \equiv \ell_i : x:=E_1 \!+_{\mathrm{Type}}\! E_2 \ra L'$, while the store
$\mu_0$ is not approximated by the abstract store $\rho^t_i$, i.e., $\mu_0 \not\in \gstore (\rho^t_i)$. 
Thus, in this case, the trace in  $O^{\ts}(\stitch_{P}(hp))$  beginning at $\tuple{\mu_0,\ts_{hp}(H_0)}$ is stuck,  
because the concrete semantics of addition is  
$\esem\grasse{E_{1} +_{{\Type}} E_{2}}\mu_0 = \undeff$, and in turn 
$\asem\grasse{x:=E_{1} +_{\Type} E_{2}}\mu_0 = \bot$,
so that we necessarily have to define
$sp(\sigma)=\tuple{\mu_0,\ts_{hp}(H_0)}$. 
 Otherwise, 
$sp(\sigma)$ simply type specializes through $\ts_{hp}$ all the commands (actually,
addition expressions) occurring in $\sigma$. Here,  
it turns out that $sp$ is well-defined, i.e.\ $sp(\sigma)$ is
a legal trace of $O^{\ts}(\stitch_{P}(hp))$, because 
any state 
$\tuple{\rho,\ell_i :x:=E_1 + E_2\ra L'}$ of $\sigma$ is always preceded by 
the state $\tuple{\rho,\bbl_i: \guard~\rho^t_i\ra \ell_i}$ and 
$\rho \in  \gstore (\rho^t_i)$ must hold. Thus, by Corollary~\ref{coro-correct}, 
$\esem\grasse{E_1 + E_2}\rho \in \esemt\grasse{E_1 + E_2}\rho^t = \mathrm{Type}$, so that 
$\asem\grasse{x:= E_1 +_{\mathrm{Type}} E_2}\rho = \asem\grasse{x:= E_1 + E_2}\rho$ holds. Consequently, 
the trace fragment 
\begin{multline*}
sp(\tuple{\rho,\bbl_i: \guard~\rho^t_i\ra \ell_i} \tuple{\rho,\ell_i :x:=E_1 + E_2\ra L'})=\\ 
\tuple{\rho,\bbl_i: \guard~\rho^t_i\ra \ell_i} \tuple{\rho,\ell_i :x:=E_1 +_{\mathrm{Type}} E_2\ra L'}
\end{multline*}
 is legal in $O^{\ts}(\stitch_{P}(hp))$. 
Furthermore, let us also observe that $\sch \circ\, ts = \sch$ trivially holds. 

\noindent
Thus, following the scheme in Section~\ref{chpo-sec}, 
these two functions $td$ and $ts$ allow us to define $td_{\full}: \Trace_{O^{\ts}_{\mathit{full}}(P,hp)} 
\ra \Trace_{\extr_{hp}(P)}$ and 
$ts_{\full}: \Trace_{\extr_{hp}(P)} \ra \Trace_{O^{\ts}_{\mathit{full}}(P,hp)}$ such that 
$\alpha_{\sch} \circ td_{\full} = \alpha_{\sch} = \alpha_{\sch} \circ ts_{\full}$, so that 
the thesis follows by Lemma~\ref{lemma-corr}. 
\end{proof}

\begin{example}\label{sieve-ex}\rm
Let us consider the following sieve of Eratosthenes in a Javascript-like language---this is taken 
from  the running example in \cite{gal2009}---where $\mathit{primes}$ is an array 
initialized with 100 $\textit{true}$ values:

\medskip
\noindent
{
\For{{\rm $(\KwSty{var}~ i = 2;\; i < 100;\; i = i+1)$}}{

\noindent
\lIf{{\rm (!$\mathit{primes}[i]$)}}{\KwSty{continue}}

\noindent
\lFor{{\rm $(\KwSty{var}~ k = i+i;\; k < 100;\; k = k+i)$}}{
  $\mathit{primes}[k] = \textit{false}$}
}
}

\medskip
\noindent
With a slight abuse, we assume that our language is extended with arrays and Boolean values
ranging in the type $\Bool$. 
The semantics of read and store for arrays is standard: 
first, the index expression is checked to be in bounds, then the value is read or stored into the array.
If the index is out of bounds then the corresponding action command gives
$\bot$, that is, we assume that the program generates an error (e.g., it is aborted).
The above program is encoded in our language as follows:
\[
\begin{aligned}
    P= \big\{
    & C_0 \equiv L_0: i:=2 \ra L_1,\,
      C_1 \equiv L_1: i<100 \ra L_2,\, C_1^c\equiv L_1: \neg(i<100) \ra L_8,\\ 
    & C_2 \equiv L_2: \mathit{primes}[i] = \cd{tt} \ra L_3,\, C_2^c\equiv L_2: \neg(\mathit{primes}[i] = \cd{ff} ) \ra L_7,\\ 
    & C_3 \equiv L_3: k:=i+i \ra L_4,\, 
      C_4 \equiv L_4: k<100 \ra L_5,\, C_4^c \equiv L_4: \neg(k<100) \ra L_7,\\
    & C_5 \equiv L_5: \mathit{primes}[k] := \cd{ff} \ra L_6,\, C_6 \equiv L_6: k:= k+i \ra L_4,\\
    & C_7 \equiv L_7: i:= i+1 \ra L_1,\,
    C_8 \equiv L_8: \cd{skip} \ra \L\big
    \}.
\end{aligned}
\]

\noindent
Let us consider the following type environment
$$\rho^t \ud \{ \mathit{primes}[n]/\Bool, i/\Int, k/\Int \}\in \Store^t$$
where $\mathit{primes}[n]/\Bool$ is a shorthand for $\mathit{primes}[0]/\Bool,\ldots,$ $\mathit{primes}[99]/\Bool$.
Then  the first traced 2-hot path on the type abstraction $\Storet$ is
$hp_1 \ud \tuple{\rho^t,C_4,\rho^t, C_5,\rho^t, C_6}$.
As a consequence, the typed trace extraction of $hp_1$ yields:
\begin{align*}
P_1 &\ud O^{\ts}_{\mathit{full}}(P,hp_1)\\
    &=
    P \smallsetminus \lbrace C_{4}, C_{4}^c \rbrace  \cup
    \{\overline{L_4} : k<100 \ra L_5,\, \ok{\overline{L_4}} : \neg(k<100) \ra L_7\}
    \cup O^{\ts}(\stitch_{P}(hp_1))
\end{align*}
where:
  \begin{align*}
    O^{\ts}(\stitch_{P}(hp_1))=\big\{& H_0 \equiv L_{4} : \cdo{guard}\: (\mathit{primes}[n]:\Bool, i:\Int, k:\Int) \rightarrow
    \ell_0,\\ 
    & H_0^c \equiv L_{4} : \neg \cdo{guard}\: (\mathit{primes}[n]:\Bool, i:\Int, k:\Int) \rightarrow
    \ok{\overline{L_4}},\\
    & H_1 \equiv \ell_0 : k<100 \ra \bbl_1,\,H_1^c \equiv \ell_0: \neg(k<100) \ra L_7,\\
    & H_2 \equiv \bbl_1 : \cdo{guard}\: (\mathit{primes}[n]:\Bool, i:\Int, k:\Int) \rightarrow
    \ell_1,\\ 
    & H_2^c \equiv \bbl_1 : \neg \cdo{guard}\: (\mathit{primes}[n]:\Bool, i:\Int, k:\Int) \rightarrow
    L_5,\\
    &H_3 \equiv \ell_1 : \mathit{primes}[k]:=\cd{ff} \ra \bbl_2,\\
    & H_4 \equiv \bbl_2 : \cdo{guard}\: (\mathit{primes}[n]:\Bool, i:\Int, k:\Int) \rightarrow
    \ell_2,\\ 
    & H_4^c \equiv \bbl_2 : \neg \cdo{guard}\: (\mathit{primes}[n]:\Bool, i:\Int, k:\Int) \rightarrow L_6,\\   
    &H_5 \equiv \ell_2: k:=k+_{\Int} i \ra L_4 \big\}.\qed
   \end{align*}
\end{example}

\section{Constant Variable Folding} \label{cf-sec}
Constant variable folding, a.k.a.\ constant propagation  \cite{WZ91}, 
is a standard and well-known program optimization, whose goal is to detect which program 
variables at some program point are constant on all possible executions and then to propagate these constant
values as far forward through the program as possible. \citeN{palsberg} show how to define this optimization along
hot paths and then prove its correctness. As a significant example, we show here how to specify and prove the correctness 
w.r.t.\ the store changes abstraction $\alpha_{sc}$ 
of this simple 
hot path optimization according to the approach defined in Section~\ref{chpo-sec}.

The constant propagation store abstraction $\CPst$ and its corresponding 
GI $(\alpha_{\CP},\wp(\Store),\CPst,\gamma_{\CP})$ have been defined 
in Example~\ref{cp-ex}. 
Following Section~\ref{trace-ext-sec}, 
any abstract store  $[x_i/ a_i~|~ x_i\in \Var]\in \CPst$, where, as usual, 
the bindings $x_i/\undeff$ are omitted, defines a
corresponding
$\text{guard}\:x_0:a_0,\ldots, x_n: a_n\in \BExp$ whose
semantics is 
induced by the GI $(\alpha_{\CP},\wp(\Store),\CPst,\gamma_{\CP})$,
as defined in Section~\ref{trace-ext-sec}:
for any $\rho\in \Store$, 
\begin{align*}
\asem\grasse{\text{guard}\:x_0:a_0,..., x_n: a_n}\rho &\ud
\begin{cases}
\rho & \!\!\text{if } 
\rho \in \gamma_{\CP}([x_i/ a_i~|~ x_i\in \Var])\\
\bot & \!\!\text{otherwise}
\end{cases}
\\
&=
\begin{cases}
\rho & \!\!\text{if } 
\forall i.\, \rho(x_i) \in \gamma_{cp}(a_i)\\
\bot & \!\!\exists i.\, \rho(x_i) \not\in \gamma_{cp}(a_i)
\end{cases}
\end{align*}

\noindent
Therefore, we have that:
\begin{align*}
&\asem\grasse{\text{guard}\:x:2, y:\mathtt{foo}} [x/2, y/3] = \bot,\\
&\asem\grasse{\text{guard}\:x:2, y:\mathtt{foo}} [x/2, y/\mathtt{foo}, z/4] = \bot,\\
&\asem\grasse{\text{guard}\:x:2, y:\mathtt{foo}} [x/2] = \bot,\\
&\asem\grasse{\text{guard}\:x:2, y:\mathtt{foo}} [x/2, y/\mathtt{foo}] = [x/2, y/\mathtt{foo}],\\
& \asem\grasse{\text{guard}\:x:2, y:\top} [x/2, y/\mathtt{foo}] = [x/2, y/\mathtt{foo}],\\
& \asem\grasse{\text{guard}\:x:2, y:\top} [x/2] = [x/2].
\end{align*}

Let us consider some hot path $hp=\tuple{\rho^c_0,C_0,\ldots,\rho^c_n,C_n}
\in$ $\ahotn(\Trace_P)$ on the constant propagation abstraction $\CPst$,
where each $\rho^c_i$ is therefore an abstract store in $\CPst$, whose corresponding guard in $\stitch_P(hp)$
will be denoted by $\guard~\rho^c_i$. The constant value information encoded in these guards 
is used to define the variable folding in the stitched hot path. 
Following \citeN[Section 2.4]{palsberg}, let $\FV: \wp(\mathbb{C}) \ra \wp(\Var)$ denote the function that returns
the ``free'' variables occurring in some set of commands (in particular, a well-defined program), i.e.,
$\FV(P)$ is the set of variables occurring in $P$ which are never-assigned-to in some command of $P$. 
As in \citeN{palsberg}, constant variable folding is restricted to expressions $E$ of some assignment $x:=E$ and is
defined as a program transform which exploits the constant information recorded 
by abstract guards in $\stitch_P(hp)$. The constant folding function 
${\cf_{hp}: \stitch_{P}(hp) \ra  \mathbb{C}}$ 
is defined  as follows:
\begin{multline*}
      \cf_{hp}(\ell_i :x:=E \!\ra\! L') \ud \\
      \begin{cases}
      \ell_i:x:= E[y_1/v_{y_1},...,y_k/v_{y_k}] \!\ra\! L' & \!\!\text{if } 
      \{y_1,...,y_k\}=\{y\in \vars(E)\cap \FV(\stitch_{P}(hp))~|\\
      &  \hfill \rho_i^c (y)=v_y\in \Value\}\neq\varnothing   \\
      \ell_i:x:= E \ra L' & \!\!\text{otherwise } 
      \end{cases}  
\end{multline*}

       $\cf_{hp}(C) \ud  C \hspace*{28.5ex}\text{if }C \not\equiv \ell_i :x:=E \!\ra\! L'$

\medskip
\noindent
where $E[y_1/v_{y_1},...,y_k/v_{y_k}]$ denotes the standard synctatic substitution of variables $y_j\in \vars(E)$ with constant
values $\rho_i^c (y_j)=v_{y_j}\in \Value$. 
Hence, when the abstract guard $\guard~\rho^c_i$ which precedes an assignment $\ell_i :x:=E \!\ra\! L'$
tells us that a free variable $y$ occuring in the expression $E$ is definitely a constant value $v_{y}\in \Value$ 
then $\cf_{hp}$ performs the corresponding variable folding in $E$. 
Thus, the hot path constant folding optimization is defined by
$$O^{\cf}(\stitch_{P}(hp)) \ud \{\cf_{hp}(C)~|~ C\in \stitch_{P}(hp)\}$$
and, in turn, this induces the full constant folding optimization 
$O^{\cf}_{\mathit{full}}(P,hp)$. The correctness of this constant folding 
optimization can be proved for the store changes observational abstraction 
\begin{theorem}[\textbf{Correctness of constant folding optimization}]\ 
For any hot path $hp \in \ahotn(\Trace_P)$ w.r.t.\ the constant propagation store abstraction $\CPst$,
$\alpha_{\sch} (\tsem \grasse{O^{\cf}_{\mathit{full}}(P,hp)})=\alpha_{\sch} (\tsem \grasse{P})$.
\end{theorem}
This proof is 
omitted, since it follows the same pattern of Theorem~\ref{corr-tte} for 
the correctness of typed trace extraction, in particular it relies on Lemma~\ref{lemma-corr}. 

\begin{example}\rm
Let us consider the following program written in a while-language:

\medskip
{ 
$x:=0$; $a:=2$;

\While{~$(x\leq 15)$}{
   \Indp
   \noindent
   \lIf{$(x\leq 5)$}{$x:=x+a$}
   
   \noindent
   \textbf{else}~$\{ a:=a+1; ~x:=x+a;\}$
}
}

\medskip
\noindent
whose translation as $P\in \Program$ goes as follows: 
\begin{equation*}
\begin{aligned}
    P= \big\{ &C_0 \equiv L_0: x:=0 \ra L_1,\, C_1\equiv L_1: a:=2 \ra L_2\\
    & C_2 \equiv L_2: x\leq 15 \ra L_3,\, C_2^c \equiv L_2: \neg(x\leq 15) \ra L_7,\\
    & C_3 \equiv L_3: x\leq 5 \ra L_4,\, C_3^c \equiv L_3: \neg(x\leq 5) \ra L_5,\\
    & C_4 \equiv L_4: x:=x+a \ra L_2,\, C_5 \equiv L_5: a:=a+1 \ra L_6\\
    & C_6 \equiv L_6: x:=x+a \ra L_2,\, C_7 \equiv L_7: \cd{skip} \ra \L\big\}
\end{aligned}
\end{equation*}
The first traced 2-hot path for the abstraction $\CPst$ is:
$$hp=\tuple{[x/\top,a/2],C_2, [x/\top,a/2], C_3, [x/\top,a/2], C_4}.$$
In fact, the initial prefix of the complete trace of $P$ which corresponds to
the terminating run of $P$ is as follows:
\begin{multline*}
\tuple{[\,],C_0}\tuple{[x/0],C_1}\tuple{[x/0,a/2],C_2}\tuple{[x/0,a/2],C_3}\tuple{[x/0,a/2],C_4}
\tuple{[x/2,a/2],C_2}\\
\tuple{[x/2,a/2],C_3}\tuple{[x/2,a/2],C_4}\tuple{[x/4,a/2],C_2}\tuple{[x/4,a/2],C_3}\tuple{[x/4,a/2],C_4}
\end{multline*}
so that $hp \in \alpha_{\mathit{hot}}^2(\Trace_P)$. 
Hence, the constant folding optimization $O^{\cf}$ along $hp$ provides:
\[
O^{\cf}_{\mathit{full}}(P,hp) =
    P \!\smallsetminus\! \lbrace C_{2}, C_{2}^c \rbrace  \cup
    \{\overline{L_2} : x\leq 15 \ra L_3,\, \overline{L_2} : \neg(x\leq 15) \ra L_7\}
    \cup O^{\cf}(\stitch_{P}(hp))
\]
where:
  \begin{align*}
    O^{\cf}(\stitch_{P}(hp))=\big\{& H_0 \equiv L_{2} : \cdo{guard}\: [x\!:\!\top, a\!:\!2] \rightarrow
    \ell_0,\,
    H_0^c \equiv L_{2} : \neg \cdo{guard}\: [x\!:\!\top, a\!:\!2] \rightarrow
    \overline{L_2},\\
    & H_1 \equiv \ell_0 :x\leq 15 \ra \bbl_1,\,H_1^c \equiv \ell_0: \neg(x\leq 15) \ra L_7,\\
    & H_2 \equiv \bbl_1 : \cdo{guard}\:  [x\!:\!\top, a\!:\!2] \rightarrow
    \ell_1,\, H_2^c \equiv \bbl_1 : \neg \cdo{guard}\: [x\!:\!\top, a\!:\!2] \rightarrow
    L_3,\\
    &H_3 \equiv \ell_1 :  x\leq 5 \ra \bbl_2,\, H_3^c \equiv \ell_1 :  \neg(x\leq 5) \ra L_5,\\
    & H_4 \equiv \bbl_2 : \cdo{guard}\: [x\!:\!\top, a\!:\!2] \rightarrow
    \ell_2,\,
    H_4^c \equiv \bbl_2 : \neg \cdo{guard}\: [x\!:\!\top, a\!:\!2] \rightarrow L_4,\\   
    &H_5 \equiv \ell_2: x:=x+2 \ra L_2 \big\}.
   \end{align*}
Therefore,  this hot path optimization allows us to fold the constant value $2$ for the variable $a$, in the hot path
command $H_5\equiv \ell_2: x:=x+2 \ra L_2$.
\qed
\end{example}

\section{Nested Hot Paths}\label{mte-sec} 
Once a first hot path $hp_1$ has been extracted by transforming $P$ to $P_1\ud 
\extr_{hp_1}(P)$, it may well happen that
a new hot path $hp_2$ in $P_1$ contains $hp_1$ as a nested sub-path.  
Following TraceMonkey's trace recording strategy~\cite{gal2009}, we attempt to nest
an inner hot path inside the current trace: during trace recording, an inner hot path is
called as a kind of ``subroutine'', this executes a loop to a successful completion and then
returns to the trace recorder that may therefore register 
the inner hot path as part of a new hot path. 
    
In order to handle nested hot paths, we need a more general definition 
of hot path which takes into account previously extracted hot paths and a corresponding program transform
for extracting nested hot paths.    
Let $P$ be the original program and let $P'$ be a hot path transform of $P$ so that
$P'\smallsetminus P$ contains all the commands (guards included) in the hot path.  
We define a function $\hotcut:\Trace_{P'} \ra (\State_{P'})^*$ that cuts from 
an execution trace $\sigma$ of $P'$ all the states whose commands appear in some previous hot path $hp$ except 
for the entry and exit states of $hp$:
\begin{align*}
    &\hotcut(\sigma) \ud
    \begin{cases}
    \epsilon & \text{if } \sigma=\epsilon\\
     \hotcut(\tuple{\rho_1,C_1}\tuple{\rho_3,C_3}\sigma')        & \text{if }   \sigma=\tuple{\rho_1,C_1}\tuple{\rho_2,C_2}
     \tuple{\rho_3,C_3}\sigma'\,\&\, C_1,C_2,C_3\not\in P\\
     \sigma_0 \hotcut(\sigma_{1^{^{\!\shortrightarrow}}})        & \text{otherwise } \\            
    \end{cases}
\end{align*}

\noindent
In turn, we define $\outerhot^N:\Trace_{P'} \ra 
\wp((\State_{P'}^\sharp)^*)$ as follows:
\begin{multline*}
    \outerhot^N(\sigma) \ud \{\tuple{a_i,C_i} \cdots \tuple{a_j,C_j}\in (\State_{P'}^\sharp)^* 
    ~|~ 
    \exists \tuple{\rho_i,C_i}\cdots\tuple{\rho_j,C_j}\in \sloop(\hotcut(\sigma))\\
    \qquad\qquad\qquad \text{~such that~} i\leq j,\, 
    \astore(\tuple{\rho_i,C_i}\cdots\tuple{\rho_j,C_j}) = \tuple{a_i,C_i} \cdots \tuple{a_j,C_j},
    \\
    \scount(\astore(\hotcut(\sigma)),\tuple{a_i,C_i} \cdots \tuple{a_j,C_j} )\geq N\}.
\end{multline*}
Clearly, when $P'=P$ it turns out that $\hotcut = \lambda \sigma.\sigma$ so that $\outerhot^N=\hot^N$. 
We define the usual collecting version of $\outerhot^N$ on $\wp(\Trace_{P'})$ as the abstraction map
$\alpha_{\mathit{outerhot}}^N \ud \lambda T. \cup_{\sigma\in T} \outerhot^N(\sigma)$. Then, 
$\alpha_{\mathit{outerhot}}^N(\tsem\grasse{P'})$ provides the set of $N$-hot paths in $P'$. 

\begin{example}\label{ex-two}\rm
Let us consider again Example~\ref{ex-one}, where
$\Store^\sharp$ is the trivial one-point store abstraction $\{\top\}$. 
In Example~\ref{ex-one}, we first extracted $hp_1 = \tuple{\top,C_1, \top,C_2, \top, 
C_3^c}$ by transforming $P$ to
$P_1 \ud \extr_{hp}(P)$.
We then consider the following trace in $\tsem\grasse{P_1}$:
\begin{align*}
\sigma = &\tuple{[\, ],C_0}\tuple{[x/0],H_0}\tuple{[x/0],H_1}\tuple{[x/0],H_2}\tuple{[x/0],H_3}
\tuple{[x/1],H_4}\tuple{[x/1],H_5}\cdots \tuple{[x/2],H_3}\\ 
&\tuple{[x/3],H_4}\tuple{[x/3],H_5^c}
\tuple{[x/3],C_4}\tuple{[x/6],H_0}\cdots \tuple{[x/9],H_5^c}\tuple{[x/9],C_4}
\tuple{[x/12],H_0}\cdots
\end{align*}

\noindent
Thus, here we have that 
\[
\hotcut(\sigma) = \tuple{[\, ],C_0}\tuple{[x/0],H_0}\tuple{[x/3],H_5^c}
\tuple{[x/3],C_4}
\tuple{[x/6],H_0} 
\tuple{[x/9],H_5^c}\tuple{[x/9],C_4}\cdots
\]
so that $hp_2 = \tuple{\top, H_0,\top, H_5^c, \top, C_4}\in 
\alpha_{\mathit{outerhot}}^2 (\tsem\grasse{P_1})$.
Hence, $hp_2$ contains a nested hot path, which is called at the beginning
of $hp_2$ and whose entry and exit 
commands are, respectively, $H_0$ and $H_5^c$.  
\qed
\end{example}

Let $hp = \tuple{a_0,C_0, \ldots,a_n,C_n}\in \alpha_{\mathit{outerhot}}^N(\tsem\grasse{P'})$ be a
$N$-hot path in $P'$, where, for all $i\in [0,n]$, we assume that
$C_i\equiv L_i: A_i \ra L_{\textit{next}(i)}$. 
Let us note that:
\begin{itemize}
\item[--] If for all $i\in [0,n]$, $C_i\in P$ then $hp$ actually is a hot path in $P$, i.e.,
$hp\in \alpha_{\mathit{hot}}^N (\tsem\grasse{P})$. 
\item[--] Otherwise, there exists some $C_k\not\in P$. 
If $C_i\in P$ and $C_{i+1}\not\in P$ then $C_{i+1}$ is the entry command of some inner hot path; 
on the other hand, if $C_{i}\not\in P$ and $C_{i+1}\in P$ then $C_i$ is the exit command of some inner hot path. 
\end{itemize}

The transform of $P'$ for extracting  $hp$ is then given as the following generalization of 
Definition~\ref{tet-def}.  

\begin{definition}[\textbf{Nested trace extraction transform}]\label{tet-def2}
\rm
The \emph{nested trace extraction transform} of $P'$ for the hot path $hp=\tuple{a_0,C_0, \ldots,a_n,C_n}$ is:
\begin{equation*}
  \begin{aligned}
   \extr_{hp}(P') &\ud 
    P \\
   (1)\quad\qquad &\mkern-30mu \smallsetminus (\{ C_{0}~|~ C_0 \in P\} \cup\{cmpl(C_{0}) ~|~ cmpl(C_0) \in P\})\\
    (2)\quad\qquad &\mkern-30mu \cup \{\overline{H_0} : act(C_0) \ra L_1~|~C_0\in P\} \cup\{\overline{H_0} : \neg act(C_0) \ra L_1^c~|~cmpl(C_0)\in P\} \\
    (3)\quad\qquad &\mkern-30mu\cup \{ L_{0} : \cdo{guard}\: E_{a_0} \rightarrow \hbar_0,\, 
    L_{0} : \neg \cdo{guard}\:E_a \rightarrow \overline{H_0}~|~C_0\in P\}\\ 
   (4)\quad\qquad & \mkern-30mu\cup \{\hbar_i : act(C_{i}) \rightarrow \bbh_{i+1}~|~
    i\in [0,n-1], C_i, C_{i+1} \in P\} \cup \{\hbar_n: act(C_n) \ra L_0~|~C_n\in P\}\\
    (5)\quad\qquad &\mkern-30mu\cup \{\hbar_i : \neg act(C_{i}) \rightarrow L_{\textit{next}(i)}^c~|~
    i\in [0,n], C_i,cmpl(C_i) \in P\}\\
    (6)\quad\qquad &\mkern-30mu\cup \{\bbh_i: \guard~ E_{a_i} \ra \hbar_i, \bbh_i: \neg \guard~ E_{a_i} \ra L_i  ~|~ i\in [1,n], C_i \in P\}\\
    (7)\quad\qquad &\mkern-30mu\cup \{\hbar_i : act(C_{i}) \rightarrow L_{i+1}~|~
    i\in [0,n-1], C_i\in P, C_{i+1} \not\in P\} \\
    (8)\quad\qquad &\mkern-30mu\smallsetminus \{C_i ~|~ i\in [0,n-1], C_i \not\in P, C_{i+1}\in P\}\\
   (9)\quad\qquad &\mkern-30mu\cup \{L_i : act(C_{i}) \rightarrow \bbh_{i+1}~|~
    i\in [0,n-1], C_i\not\in P, C_{i+1} \in P\} 
       \end{aligned}
\end{equation*}
where we define $\stitch_{P'}(hp)\ud (3) \cup (4) \cup (5) \cup (6) \cup (7) \cup (9)$.
\qed
\end{definition}
\noindent
Let us observe that:
\begin{itemize}
\item[--] Clauses (1)--(6) are the same clauses of the trace extraction transform of 
Definition~\ref{tet-def}, with
the additional constraint that all the commands $C_i$ of $hp$ are required to belong to the original program $P$. This is
equivalent to ask that 
any $C_i$ is not the entry or exit command of a nested hot path inside $hp$, i.e.,  $C_i \not \in P'\smallsetminus P$.  
In  Definition~\ref{tet-def}, where no previous hot path extraction is assumed, 
any command $C_i$ of $hp$ belongs to $P$, so that this constraint is trivially satisfied.  
\item[--] Clause (7) where $C_i\in P$ and $C_{i+1}\not \in P$, namely 
$next(C_i)$ is the call program 
point of a nested hot path $nhp$ and $C_{i+1}$ is the entry command 
of $nhp$, performs a relabeling that allows to neatly nest $nhp$ in $hp$. 
 
\item[--] Clauses (8)--(9) where $C_i\not\in P$ and $C_{i+1}\in P$, i.e., 
$C_i$ is the exit command of a nested hot path $nhp$ that returns to the
program point $lbl(C_{i+1})$, performs the relabeling of $suc(C_i)$ in $C_i$ in order
to return from $nhp$ to $hp$;  
\item[--] $\overline{H_0}$, $\hbar_i$ and $\bbh_i$ are meant to be fresh labels, i.e., they
have not been already used in $P'$.
\end{itemize}

\begin{example}\label{ex-three}
\rm 
Let us go on with Example~\ref{ex-two}. 
The second traced hot path in $\alpha_{\mathit{outerhot}}^2 (\tsem\grasse{P_1})$ is:  
\begin{multline*}
hp_2 = 
\tuple{\top,  H_0\equiv L_1: \guard~E_\top \ra \ell_0,\\ 
\top, H_5^c\equiv \ell_2: (x\% 3 = 0) \ra L_4,
\top, C_4\equiv 
L_4: x:=x+3 \ra L_1}.
\end{multline*}
According to Definition~\ref{tet-def2}, trace extraction of $hp_2$ in $P_1$ yields 
the following transform: 
\begin{equation*}
  \begin{aligned}
    \extr_{hp_2}(P_1) \ud &\\
 \text{[by clause~(8)]}\qquad   & P_1 \smallsetminus \{H_5^c\}\\
   \text{[by clause~(9)]}\qquad &\cup \{ \ell_2:(x\%3 = 0) \ra \bbh_2\}\\
    \text{[by clause~(6)]}\qquad&\cup \{ \bbh_2: \guard~E_\top \ra \hbar_2, \bbh_2: \neg\guard~E_\top \ra L_4\}\\
    \text{[by clause~(4)]}\qquad&\cup \{ \hbar_2: x:=x+3 \ra L_1\} 
  \end{aligned} 
\end{equation*}
where we used the additional fresh labels $\bbh_2$ and $\hbar_2$.
\qed
\end{example}

\begin{example}\rm
Let us consider again Example~\ref{sieve-ex}. After the trace extraction of $hp_1$ that transforms $P$ to $P_1$, 
a
second traced 2-hot path is the following:
$$hp_2 \ud \tuple{\rho^t,
C_1,\rho^t, C_2,\rho^t, C_3 ,\rho^t, H_0,\rho^t, H_1^c,\rho^t, C_7 }$$
where $\rho^t = \{ \mathit{primes}[n]/\Bool, i/\Int, k/\Int \}\in \Store^t$. 
Thus, $hp_2$ contains a nested hot path which 
is called at $suc(C_3)=L_4$ and whose entry and exit commands are, respectively, 
$H_0$ and $H_1^c$.  
Here, typed  trace extraction according to Definition~\ref{tet-def2}
provides 
the following transform of $P_1$:
\begin{equation*}
  \begin{aligned}
    P_2 \ud O^{\ts}_{\mathit{full}}(P_1,hp_2) =
    P_1 &\smallsetminus  \{C_{1}, C_{1}^c\} \cup \big\{\\
    & \overline{H_0} : i<100 \ra L_2,\,  \overline{H_0} : \neg(i<100) \ra L_8,\\
    & H_6\equiv  L_{1} : \cdo{guard}\: (\mathit{primes}[n]:\Bool, i:\Int, k:\Int) \rightarrow
    \hbar_0,\\ 
    & H_6^c\equiv L_{1} : \neg \cdo{guard}\: (\mathit{primes}[n]:\Bool, i:\Int, k:\Int) \rightarrow
    \overline{H_0},\\
    & H_7\equiv \hbar_0 : i<100 \ra \bbh_1,\,H_7^c\equiv \hbar_0: \neg(i<100) \ra L_8,\\
    & H_{8}\equiv \bbh_1: \cdo{guard}\: (\mathit{primes}[n]:\Bool, i:\Int, k:\Int) 
    \ra \hbar_1,\\ 
    & H_{8}^c \equiv \bbh_1: \neg 
    \cdo{guard}\: (\mathit{primes}[n]:\Bool, i:\Int, k:\Int) 
    \ra L_2,\\
    &H_{9}\equiv \hbar_1 : \mathit{primes}[i]=\cdo{tt} \ra \bbh_1,\,
    H_{9}^c\equiv \hbar_1 : \neg(\mathit{primes}[i]=\cdo{tt}) \ra L_7,\\
    & H_{10}\equiv \bbh_2: \cdo{guard}\: (\mathit{primes}[n]:\Bool, i:\Int, k:\Int) 
    \ra \hbar_2,\\ 
    &H_{10}^c \equiv \bbh_2: \neg 
    \cdo{guard}\: (\mathit{primes}[n]:\Bool, i:\Int, k:\Int) 
    \ra L_3,\\
    & H_{11}\equiv \hbar_2: k:=i+_{\Int} i \ra L_4 \big \}\\
    & \smallsetminus \{H_1^c\} \cup \big\{(H_1^c)' \equiv \ell_0: \neg(k<100) \ra \bbh_3,\\
     & H_{12}\equiv \bbh_3: \cdo{guard}\: (\mathit{primes}[n]:\Bool, i:\Int, k:\Int) 
    \ra \hbar_3,\\ 
    &H_{12}^c \equiv \bbh_3: \neg 
    \cdo{guard}\: (\mathit{primes}[n]:\Bool, i:\Int, k:\Int) 
    \ra L_7,\\
    &H_{13}\equiv \hbar_3: i:=i+_{\Int} 1 \ra L_1\big\}.
   \end{aligned}
\end{equation*}
\noindent
Finally, a third traced 2-hot path in $P_2$ is 
$hp_3 \ud \tuple{\rho^t, H_6, \rho^t, H_{9}^c, \rho^t, C_7 }$ 
which contains a nested hot path which is called at the beginning
of $hp_3$ and whose entry and exit commands are,
respectively, $H_6$ and $H_{9}^c$. Here, 
typed trace extraction of $hp_3$ yields:
\begin{equation*}
  \begin{aligned}
    P_3 \ud O^{\ts}_{\mathit{full}}(P_2,hp_3) =
    P_2 \smallsetminus \{ H_{9}^c \} \cup \big\{
    &(H_{9}^c)'\equiv \hbar_1 : \neg (\mathit{primes}[i]= \cdo{tt}) \ra \bbj_2,\,\\
    & \bbj_2: \cdo{guard}\: (\mathit{primes}[n]:\Bool, i:\Int, k:\Int) 
    \ra \jmath_2,\\ 
    & \bbj_2: \neg 
    \cdo{guard}\: (\mathit{primes}[n]:\Bool, i:\Int, k:\Int) 
    \ra L_7,\\
    & \jmath_2: i:=i+_{\Int} 1 \ra L_1
  \big\}.
   \end{aligned}
\end{equation*}
We have thus obtained the same three trace extraction steps described by 
\citeN[Section~2]{gal2009}. In particular, in $P_1$ we specialized the typed
addition operation ${k+_{\Int} i}$, in $P_2$ we specialized $i+_{\Int} i$ and
$i+_{\Int} 1$, while in $P_3$ we specialized once again $i+_{\Int} 1$ in a different
hot path. Thus, in 
$P_3$ all the addition operations occurring in assignments have been type specialized. 
\qed
\end{example}

\section{Comparison with Guo and Palsberg's Framework}\label{GP-sec}
A formal model for tracing JIT compilation has been put forward at POPL~2011 symposium by 
\citeN{palsberg}. Its main distinctive feature is the use of a
bisimulation relation \cite{milner95} 
to model the
operational equivalence between source and optimized programs. 
In this section, we show
how this model can be
expressed within our framework.

\subsection{Language and Semantics}
\citeN{palsberg} rely on a simple imperative language (without jumps and) with while loops and 
a so-called bail construct. Its syntax is as follows:
\begin{align*}
&E::= v ~|~ x ~|~ E_1 + E_2\\ 
& B::= \cd{tt}~|~\cd{ff}~|~ E_1 \leq E_2 ~|~ \neg B ~|~ B_1 \wedge B_2 \\
 &\Cmd\ni c ::= \KwSty{skip}; ~|~ x:=E; ~|~ \KwSty{if}~B~\KwSty{then}~S ~|~
\KwSty{while}~B~\KwSty{do}~S~|~
\KwSty{bail}~B~\KwSty{to}~S \\
&\Stm\ni S ::= \epsilon ~|~ c S
\end{align*}
where $\epsilon$ stands for the empty string. Thus, any statement $S\in \Stm$ is a (possibly empty) sequence
of commands $c^n$, with $n\geq 0$.  
We follow \citeN{palsberg} in making an abuse in program syntax by assuming that 
if $S_1,S_2\in \Stm$ then $S_1 S_2\in \Stm$, where $S_1 S_2$ denotes a simple string concatenation of $S_1$ and $S_2$.  
We denote by $\State_{\GP} \ud \Store \times \Stm$ the set of states for this language.  
The baseline small-step operational semantics $\ra_B \: \subseteq \State_{\GP} \times \State_{\GP}$
is standard and is given in continuation-style (where $K\in \Stm$): 
\begin{equation*}
\begin{array}{ll}
\tuple{\rho, \epsilon } \not \ra_B  &  \\[5pt]
\tuple{\rho, \KwSty{skip}; K} \ra_B \tuple{\rho,K} &  \\[5pt]
\tuple{\rho, x:=E; K} \ra_B \tuple{\rho[x/\esem\grasse{E}\rho],K} &  \\[5pt]
\tuple{\rho, (\KwSty{if}~B~\KwSty{then}~S) K} \ra_B 
\begin{cases}
\tuple{\rho,K} & 
\text{if}~\bsem\grasse{B}\rho = \textit{false}\\
\tuple{\rho,S K} & 
\text{if}~\bsem\grasse{B}\rho = \textit{true}\\
\end{cases}\\[15pt]
\tuple{\rho, (\KwSty{while}~B~\KwSty{do}~S) K} \ra_B \tuple{\rho,
(\KwSty{if}~B~\KwSty{then}~(S\, \KwSty{while}~B~\KwSty{do}~S))\, K} &\\[5pt]
\tuple{\rho, (\KwSty{bail}~B~\KwSty{to}~S) K} \ra_B 
\begin{cases}
\tuple{\rho,K} & 
\text{if}~\bsem\grasse{B}\rho = \textit{false}\\
\tuple{\rho,S} & 
\text{if}~\bsem\grasse{B}\rho = \textit{true}
\end{cases}
\end{array}	 
\end{equation*}

\noindent
The relation $\ra_B$ is clearly deterministic and we denote by 
$$\Trace^{\GP}\ud \{\sigma\in \State_{\GP}^+~|~\forall i\in [0,|\sigma|-1). \,
\sigma_{i} \ra_B \sigma_{i+1}\}$$ the set of 
generic program traces for Guo and Palsberg's language. Then,
given a program $S\in \text{Stm}$, so that $\Store_S\ud \mathit{vars}(S) \ra \Valueu$ denotes
the set of stores for $S$, its partial trace semantics is 
$$
\tsem_{\GP}\grasse{S} = \Trace^{\GP}_S \ud \{ 
\sigma\in \Trace^{\GP} ~|~\sigma_0 = \tuple{\rho,S},\, \rho \in \Store_S\}.
$$
Notice that, differently from our trace semantics, a partial trace of the program $S$ 
always starts from an initial state, i.e., $\tuple{\rho,S}$.

\subsection{Language Compilation}

Programs in $\Stm$ can be compiled into $\Program$ by resorting 
to an \emph{injective} labeling function $\elb:\Stm \ra \mathbb{L}$ that assigns different
labels to different statements. 

\begin{definition}[\textbf{Language compilation}] \label{compilation}
\rm
The ``first command'' compilation function $\comC:\Stm \ra \wp(\mathbb{C})$ is defined as follows: 
\begin{equation*}
\begin{aligned}
\comC(\epsilon) \ud & \:\{\elb(\epsilon) : \cd{skip}\ra \L\}\\
\comC\big(S'\equiv (\KwSty{skip}; K)\big) \ud &\: \{\elb(S'): \cd{skip} \ra \elb(K) \}\\
\comC\big(S'\equiv (x:=E; K)\big) \ud & \:\{\elb(S'): x:=E \ra \elb(K) \}\\
\comC\big(S'\equiv ((\KwSty{if}~B~\KwSty{then}~S)K)\big) \ud &
\: \{\elb(S'): B \ra \elb(SK),\, \elb(S'): \neg B \ra \elb(K) \}\\
\comC\big(S'\equiv ((\KwSty{while}~B~\KwSty{do}~S) K)\big) \ud &
\: \{\elb(S') : \cd{skip} \ra \elb((\KwSty{if}~B~\KwSty{then}~(S\,\KwSty{while}~B~\KwSty{do}~S)) K)\}\\
\comC\big(S'\equiv((\KwSty{bail}~B~\KwSty{to}~S) K)\big) \ud &
\: \{\elb(S'): B \ra \elb(S),\, \elb(S'): \neg B \ra \elb(K) \} 
\end{aligned}
\end{equation*}

\noindent 
Then, the
full compilation function $\cC: \Stm \ra \wp(\mathbb{C})$ is recursively defined
by the following clauses:
\begin{equation*}
\begin{aligned}
\cC(\epsilon) \ud & \: \comC(\epsilon)\\
\cC(\KwSty{skip}; K) \ud &\: \comC(\KwSty{skip}; K) \cup \cC(K)\\
\cC(x:=E; K) \ud & \:\comC(x:=E; K) \cup \cC(K) \\
\cC((\KwSty{if}~B~\KwSty{then}~S)K) \ud &\: \comC((\KwSty{if}~B~\KwSty{then}~S)K)
  \cup \cC(SK) \cup \cC(K)\\
\cC((\KwSty{while}~B~\KwSty{do}~S) K) \ud &\: \comC((\KwSty{while}~B~\KwSty{do}~S) K) \cup \cC((\KwSty{if}~B~\KwSty{then}~(S\, \KwSty{while}~B~\KwSty{do}~S)) K) \\
\cC((\KwSty{bail}~B~\KwSty{to}~S) K) \ud &
\: \comC((\KwSty{bail}~B~\KwSty{to}~S) K) 
\cup \cC(S) \cup \cC(K) 
\end{aligned}
\end{equation*}  
Given $S\in \Stm$, $\elb(S)$ is the initial label of $\cC(S)$, while $\L$ is, as usual, the undefined label
where the execution becomes stuck. 
\hfill\ensuremath{\qed}
\end{definition}

It turns out that the recursive function $\cC$ is well-defined---the easy proof is standard and is omitted, let us just
observe that $\cC((\KwSty{while}~B~\KwSty{do}~S) K)$ is a base case---so that, 
for any $S\in \Stm$, $\cC(S)$ is a finite set
of commands. 
Let us observe that, by Definition~\ref{compilation}, if 
$\tuple{\rho,S}\ra_B \tuple{\rho',S'}$ then 
$\cC(S')\subseteq \cC(S)$ (this can be proved through an easy structural induction on $S$). Consequently, 
if  $\tuple{\rho,S}\ra_B^* \tuple{\rho',S'}$ then 
$\cC(S')\subseteq \cC(S)$. 

\begin{example}\label{ex-compilation}
\rm
Consider the following program $S\in\Stm$ in Guo and Palsberg's syntax:

\medskip
{%
$x:=0$;

\KwSty{while}~$B_1$~\KwSty{do}~$x:=1$;

$x:=2$;

\KwSty{bail}~$B_2$~\KwSty{to}~$x:=3$;

$x:=4$;

}

\medskip
\noindent
$S$ is then compiled in our language by $\cC$ in Definition~\ref{compilation}  as follows:
\begin{equation*}
  \begin{aligned}
    \cC(S) = \big\{ &\elb (S) : x:=0 \ra \elb_{\KwSty{while}},\,
    \elb_{\KwSty{while}}: \cd{skip} \ra \elb_{\KwSty{ifwhile}},\\
    &\elb_{\KwSty{ifwhile}}: B_1 \ra \elb_1,\, \elb_{\KwSty{ifwhile}}: \neg B_1 \ra \elb_2,\,
    \elb_1: x:=1 \ra \elb_{\KwSty{while}},\\ 
    & \elb_2: x:=2 \ra \elb_{\KwSty{bail}},\,
      \elb_{\KwSty{bail}}: B_2 \ra \elb_3,\; \elb_{\KwSty{bail}}: \neg B_2 \ra \elb_4, \\
     &\elb_3: x:=3 \ra \elb_\epsilon,\, \elb_4: x:=4\ra \elb_\epsilon,\, \elb_\epsilon : \cd{skip}\ra \L\big\}.
   \end{aligned}
\end{equation*}
Notice that in the command $\elb_{\KwSty{bail}}: B_2 \ra \elb_3$, the label $\elb_3$ stands for 
$\elb (x:=3;)$ so that $\cC(x:=3;) \equiv \elb_3: x:=3 \ra \elb_\epsilon$, i.e.,
after the execution of $x:=3$ the program terminates. 
\qed
\end{example}

Correctness for the above compilation function $\cC$ means that for
any $S\in \Stm$: (i)~$\cC(S)\in \Program$ and (ii)~program 
traces of $S$ and $\cC(S)$ have the same store
sequences. In the proof we will make use of a ``state compile'' function $\stateC:\State_{\GP} \ra \State$
as defined in Figure~\ref{scfun}. In turn, $\stateC$ allows us to define a ``trace compile'' 
function $\traceC: \tsem_{\GP}\grasse{S} \ra \tsem^\iota\grasse{\cC(S)}$ which applies
state-by-state the function $\stateC$ to traces as follows:
$$\traceC(\epsilon) \ud \epsilon;\quad \traceC(s\tau)\ud \stateC(s) \traceC(\tau).$$ 

\begin{figure}
\begin{align*}
\stateC(\tuple{\rho,\epsilon}) & \ud \tuple{\rho, \elb(\epsilon) : \cd{skip} \ra \L}\\
\stateC(\tuple{\rho, S\equiv (\KwSty{skip}; K)}) & \ud \tuple{\rho, \elb(S): \cd{skip} \ra \elb(K)}\\
\stateC(\tuple{\rho, S\equiv (x:=E; K)}) & \ud \tuple{\rho, \elb(S): x:=E \ra \elb(K)}\\
\stateC(\tuple{\rho, S\equiv ((\KwSty{if}~B~\KwSty{then}~S')K)}) & \ud 
\begin{cases}
\tuple{\rho, \elb(S): B \ra \elb(S' K)} & \text{if~} \bsem\grasse{B}\rho = \textit{true}\\
\tuple{\rho, \elb(S): \neg B \ra \elb(K)} & \text{if~} \bsem\grasse{B}\rho = \textit{false}
\end{cases}\\
\stateC(\tuple{\rho, S\equiv ((\KwSty{while}~B~\KwSty{do}~S')K)}) & \ud 
\tuple{\rho, \elb(S): \cd{skip} \ra \elb((\KwSty{if}~B~\KwSty{then}~(S'\, \KwSty{while}~B~\KwSty{do}~S'))K)}\\
\stateC(\tuple{\rho, S\equiv ((\KwSty{bail}~B~\KwSty{to}~S')K)}) & \ud 
\begin{cases}
\tuple{\rho, \elb(S): B \ra \elb(S')} & \text{if~} \bsem\grasse{B}\rho = \textit{true}\\
\tuple{\rho, \elb(S): \neg B \ra \elb(K)} & \text{if~} \bsem\grasse{B}\rho = \textit{false}
\end{cases}
\end{align*}
\caption{Definition of the state compile function $\stateC:\State_{\GP} \ra \State$.}
\label{scfun}
\end{figure}

\begin{lemma}\label{sc-tc-lemma}
\begin{enumerate}
\item[{\rm (1)}] $\tuple{\rho,S}  \ra_B \tuple{\rho',S'} \; \Lra\; \stateC(\tuple{\rho',S'}) \in \ssem (\stateC(\tuple{\rho,S}))$
\item[{\rm (2)}] $\traceC$ is well-defined.
\end{enumerate}
\end{lemma}
\begin{proof}
We show the equivalence (1) by structural induction on $S\in \Stm$.

\medskip
$[S\equiv \epsilon]$: Trivially true, since $\tuple{\rho,S}  \not\ra_B$ and $\ssem \tuple{\rho,\elb(\epsilon) : \cd{skip} \ra \L} = \varnothing$.

\medskip
$[S\equiv x:=E; K]$ $(\Ra)$: If $\tuple{\rho,x:=E; K} \ra_B \tuple{\rho[x/\esem\grasse{E}\rho],K}$, 
$\stateC(\tuple{\rho, x:=E; K})=\tuple{\rho, \elb(S):x:=E \ra \elb(K)}$ 
and $\stateC(\tuple{\rho[x/\esem\grasse{E}\rho], K}) = \tuple{\rho[x/\esem\grasse{E}\rho], \elb(K): A \ra \elb(S')}$ for some action $A$ and statement $S'$, then, by definition of the transition semantics $\ssem$,   
$\tuple{\rho[x/\esem\grasse{E}\rho], \elb(K): A \ra \elb(S')} \in 
\ssem\tuple{\rho, \elb(S):x:=E \ra \elb(K)}$. 

\noindent
$(\La)$: If $\tuple{\rho'',C}=\stateC(\tuple{\rho', S'}) \in 
\ssem\tuple{\rho, \elb(S):x:=E \ra \elb(K)}$ then: (1)~$\esem\grasse{E}\rho \neq \mathit{undef}$, 
(2)~$\rho''=\rho[x/\esem\grasse{E}\rho]$, and therefore $\rho'=\rho[x/\esem\grasse{E}\rho]$; 
(3)~$lbl(C)=\elb(K)$, and therefore $S'=K$. Hence, $\tuple{\rho,x:=E; K}  \ra_B \tuple{\rho[x/\esem\grasse{E}\rho],K}=\tuple{\rho',S'}$.

\medskip
$[S\equiv \KwSty{skip}; K]$  Analogous to $S\equiv x:=E; K$.

\medskip
$[S\equiv (\KwSty{if}~B~\KwSty{then}~T)K]$ $(\Ra)$:
Assume that $\bsem\grasse{B}\rho =\textit{false}$, so that 
$\tuple{\rho,(\KwSty{if}~B~\KwSty{then}~T)K} \ra_B \tuple{\rho,K}$, 
$\stateC(\tuple{\rho, (\KwSty{if}~B~\KwSty{then}~T)K})=\tuple{\rho, \elb(S): \neg B \ra \elb(K)}$ and  
$\stateC(\tuple{\rho, K}) = \tuple{\rho, \elb(K): A \ra \elb(T')}$ for some $A$ and 
$T'\in\Stm$. Hence, by definition of $\ssem$,   
$\tuple{\rho, \elb(K): A \ra \elb(T')} \in 
\ssem\tuple{\rho, \elb(S): \neg B \ra \elb(K)}$. 
On the other hand, if $\bsem\grasse{B}\rho =\textit{true}$ then 
$\tuple{\rho,(\KwSty{if}~B~\KwSty{then}~T)K} \ra_B \tuple{\rho,T K}$, 
$\stateC(\tuple{\rho, (\KwSty{if}~B~\KwSty{then}~T)K})=\tuple{\rho, \elb(S): B \ra \elb(T K)}$ and  
$\stateC(\tuple{\rho, T K}) = \tuple{\rho, \elb(T K): A \ra \elb(T')}$ for some $A$ and 
$T'$. Hence,   
$\tuple{\rho, \elb(T K): A \ra \elb(T')} \in 
\ssem\tuple{\rho, \elb(S): B \ra \elb(T K)}$. 

\noindent
$(\La)$: Assume that $\bsem\grasse{B}\rho =\textit{false}$, so that $\stateC(\tuple{\rho, (\KwSty{if}~B~\KwSty{then}~T)K})=\tuple{\rho, \elb(S): \neg B \ra \elb(K)}$, and  
$\tuple{\rho'',C}=\stateC(\tuple{\rho', S'}) \in 
\ssem\tuple{\rho, \elb(S): \neg B \ra \elb(K)}$. 
Hence: 
(1)~$\rho''=\rho$ and therefore $\rho'=\rho$; 
(2)~$lbl(C)=\elb(K)$, and therefore $S'=K$. Hence, 
$\tuple{\rho, (\KwSty{if}~B~\KwSty{then}~T)K}  \ra_B \tuple{\rho,K}=\tuple{\rho',S'}$.
On the other hand, if 
$\bsem\grasse{B}\rho =\textit{true}$ then 
$\stateC(\tuple{\rho, (\KwSty{if}~B~\KwSty{then}~T)K})=\tuple{\rho, \elb(S): B \ra \elb(T K)}$ and  
$\tuple{\rho'',C}=\stateC(\tuple{\rho', S'}) \in 
\ssem\tuple{\rho, \elb(S): B \ra \elb(T K)}$. 
We thus have that:
(1)~$\rho''=\rho$ and therefore $\rho'=\rho$; 
(2)~$lbl(C)=\elb(TK)$, and therefore $S'=TK$. Hence, 
$\tuple{\rho, (\KwSty{if}~B~\KwSty{then}~T)K}  \ra_B \tuple{\rho,TK}=\tuple{\rho',S'}$.

\medskip
$[S\equiv (\KwSty{while}~B~\KwSty{do}~T)K ]$ 
$(\Ra)$: We have that $\tuple{\rho,(\KwSty{while}~B~\KwSty{do}~T)K} \ra_B 
\tuple{\rho,(\KwSty{if}~B~\KwSty{then}~(T\, \KwSty{while}~B~\KwSty{do}~T))\, K}$ and
$\stateC(\tuple{\rho, (\KwSty{while}~B~\KwSty{do}~T)K})=\tuple{\rho, \elb(S):\cd{skip} \ra 
\elb((\KwSty{if}~B~\KwSty{then}~(T\, \KwSty{while}~B~\KwSty{do}~T))\, K)}$. 
If  
$\bsem\grasse{B}\rho =\textit{true}$ then
$\stateC(\tuple{\rho, (\KwSty{if}~B~\KwSty{then}~(T\, \KwSty{while}~B~\KwSty{do}~T))\, K}) = 
\tuple{\rho, \elb((\KwSty{if}~B~\KwSty{then}~(T\, \KwSty{while}~B~\KwSty{do}~T))\, K): B \ra 
\elb(T\, (\KwSty{while}~B~\KwSty{do}~T) K)}$; on the other hand, 
if $\bsem\grasse{B}\rho =\textit{false}$ then
$\stateC(\tuple{\rho, (\KwSty{if}~B~\KwSty{then}~(T\, \KwSty{while}~B~\KwSty{do}~T))\, K}) = 
\tuple{\rho, \elb((\KwSty{if}~B~\KwSty{then}~(T\, \KwSty{while}~B~\KwSty{do}~T))\, K): \neg B \ra 
\elb(K)}$. In both cases, we have that:
\begin{align*}
&\tuple{\rho, \elb((\KwSty{if}~B~\KwSty{then}~(T\, \KwSty{while}~B~\KwSty{do}~T))\, K): B \ra 
\elb(T\, (\KwSty{while}~B~\KwSty{do}~T) K)},\\
&\tuple{\rho, \elb((\KwSty{if}~B~\KwSty{then}~(T\, \KwSty{while}~B~\KwSty{do}~T))\, K): B \ra 
\elb(K)}\\
 &\qquad\qquad\qquad\qquad \in 
\ssem\tuple{\rho, \elb(S):\cd{skip} \ra 
\elb((\KwSty{if}~B~\KwSty{then}~(T\, \KwSty{while}~B~\KwSty{do}~T))\, K)}.
\end{align*}

\noindent
$(\La)$: If $\tuple{\rho'',C}=\stateC(\tuple{\rho', S'}) \in 
\ssem\tuple{\rho, \elb(S):\cd{skip} \ra \elb((\KwSty{if}~B~\KwSty{then}~(T\, \KwSty{while}~B~\KwSty{do}~T))\, K)}$ 
then:
(1)~$\rho''=\rho$, and therefore $\rho'=\rho$; 
(2)~$lbl(C)= \elb((\KwSty{if}~B~\KwSty{then}~(T\, \KwSty{while}~B~\KwSty{do}~T))\, K)$, and therefore $S'=(\KwSty{if}~B~\KwSty{then}~(T\, \KwSty{while}~B~\KwSty{do}~T))\, K$. 
Hence, $\tuple{\rho, (\KwSty{while}~B~\KwSty{do}~T)K }  \ra_B 
\tuple{\rho, (\KwSty{if}~B~\KwSty{then}~(T\, \KwSty{while}~B~\KwSty{do}~T))\, K}=\tuple{\rho',S'}$.

\medskip
$[S\equiv (\KwSty{bail}~B~\KwSty{to}~T)K]$ Analogous to $S\equiv (\KwSty{if}~B~\KwSty{then}~T)K$.

\medskip
Let us now turn to point~(2). By the $\Ra$ implication of the 
equivalence $(1)$, we have that if $\tau \in \tsem_{\GP}\grasse{S}$ then $\traceC(\tau) \in \tsem\grasse{\cC(S)}$:
this can be shown by an easy induction on the length of $\tau$ and by using the fact that 
if $\traceC(\tau)= \tuple{\rho_0,C_0}\tuple{\rho_1,C_1}\cdots \tuple{\rho_n,C_n}$ then, for any $i$, $C_i\in \cC(S)$. 
Moreover, since $\elb(S)$ is the initial label of the compiled program $\cC(S)$ and $lbl(C_0)= \elb(S)$, 
we also notice that $\traceC(\tau) \in \tsem^\iota\grasse{\cC(S)}$. Therefore, $\traceC$ is a well-defined function. 
\end{proof}

Let $\st:\Trace^{\GP}\cup \Trace \ra \Store^*$ be the function that returns the store sequence
of any trace, that is: 
$$\st (\epsilon)\ud \epsilon \quad \text{and} \quad
\st (\tuple{\rho,S}\sigma)\ud \rho\cdot \st(\sigma).$$ 
Also, given a set $X$ of traces, let $\alpha_{\st}(X) \ud \{\st(\sigma)~|~\sigma \in X\}$. 
Then,  
correctness of the compilation function $\cC$ goes as follows:

\begin{theorem}[\textbf{Correctness of language compilation}]
If $S\in \Stm$ then $\cC(S)\in \Program$ and 
$\alpha_{\st}(\mathbf{T}_{\GP}\grasse{S}) = \alpha_{\st}(\mathbf{T}^\iota \grasse{\cC(S)})$. 
\end{theorem}
\begin{proof}
We define a ``trace de-compile'' function $\traceD: \tsem^\iota\grasse{\cC(S)} \ra 
\tsem_{\GP}\grasse{S}$ as follows. Consider a trace $\sigma = \tuple{\rho_0,C_0}\cdots \tuple{\rho_n,C_n}\in 
\tsem^\iota\grasse{\cC(S)}$, so that $lbl(C_0)=\elb(S)$, for any $i\in [0,n]$, $C_i \in \cC(S)$ and for any $i\in [0,n)$, 
$\tuple{\rho_{i+1},C_{i+1}}\in \ssem\grasse{\cC(S)}\tuple{\rho_i,C_i}$. Since $lbl(C_0)=\elb(S)$, by definition
of $\stateC$, we have that 
$\tuple{\rho_0,C_0} = \stateC(\tuple{\rho_0,S})$. 
Then, since $\tuple{\rho_1,C_1}\in \ssem\grasse{\cC(S)}(\stateC(\tuple{\rho_0,S}))$, there exists $S_1\in \Stm$ such
that $lbl(C_1) =\elb(S_1)$, so that, $\tuple{\rho_1,C_1}=\stateC(\tuple{\rho_1,S_1})$. 
Hence, from $\stateC(\tuple{\rho_1,S_1})\in \ssem\grasse{\cC(S)}(\stateC(\tuple{\rho_0,S}))$, by the 
 implication $\La$ of Lemma~\ref{sc-tc-lemma}~(1), 
we obtain that $\tuple{\rho_0,S}\ra_B \tuple{\rho_1,S_1}$. 
Thus, an easy induction allows us to show that for any $i\in [1,n]$ there exists $S_i\in \Stm$
such that 
$$\tuple{\rho_0,S}\ra_B \tuple{\rho_1,S_1} \ra_B \cdots \ra_B \tuple{\rho_n,S_n}$$
and $\stateC(\tuple{\rho_i,S_i})=\tuple{\rho_i,C_i}$. 
We therefore define 
$\traceD(\sigma) \ud \tuple{\rho_0,S}\tuple{\rho_1,S_1}\cdots \tuple{\rho_n,S_n}\in  \tsem_{\GP}\grasse{S}$. 
Moreover, we notice that $st(\traceD(\sigma)) = st(\sigma)$. 
Let us also observe that $st\circ \traceC = st$, since $\traceC$ does 
not affect stores.  

\noindent
Summing up, we obtain:
\begin{align*}
\alpha_{\st}(\tsem_{\GP}\grasse{S})&=\text{\quad[since $st\circ \traceC = st$]}\\
\alpha_{\st}(\traceC(\tsem_{\GP}\grasse{S}))  &\subseteq\text{\quad[by Lemma~\ref{sc-tc-lemma}~(2), $\traceC$ is well-defined]}\\
\alpha_{\st}(\tsem^\iota\grasse{\cC(S)}) &= \text{\quad[since $st\circ \traceD = st$]}\\
\alpha_{\st}(\traceD(\tsem^\iota\grasse{\cC(S)}))  &\subseteq\text{\quad[since $\traceD$ is well-defined]}\\
\alpha_{\st}(\tsem_{\GP}\grasse{S})  &
\end{align*}
and this closes the proof.
\end{proof}

\subsection{Bisimulation} 

Correctness of trace extraction in \cite{palsberg} relies on a notion
of bisimulation relation, parameterized by program stores. 
Let us recall this definition. 
If $\tuple{\rho,S} \ra_B \tuple{\rho,S'}$ then this ``silent'' transition that does not change
the store is also denoted 
by $\tuple{\rho,S} \rat_B \tuple{\rho,S'}$. Moreover, for the assignment transition 
$\tuple{\rho,x:=E;K} \ra_B \tuple{\rho[x/\esem\grasse{E}\rho],K}$, if 
$\delta = [x/\esem\grasse{E}\rho]$ denotes the corresponding store update of $\rho$ then this transition
is also
denoted by $\tuple{\rho,x:=E;K} \rad_B \tuple{\rho[x/\esem\grasse{E}\rho],K}$.
Let $\Act\ud \{ \delta~|~\delta$ is a store update$\}\cup \{\tau\}$. Then, for a nonempty sequence of actions 
$s=a_1\cdots a_n \in \Act^+$, we define:  
$$\tuple{\rho,S} \stackrel{s}{\Rightarrow}_B \tuple{\rho',S'}\quad\text{iff}\quad 
\tuple{\rho,S} \stackrel{\tau}{\ra}_{B}^{_{{\scriptstyle *}}} \circ \stackrel{a_1}{\ra}_{B} \circ 
\stackrel{\tau}{\ra}_{B}^{_{{\scriptstyle *}}} \cdots  \stackrel{\tau}{\ra}_{B}^{_{{\scriptstyle *}}} \circ \stackrel{a_n}{\ra}_{B} \circ 
\stackrel{\tau}{\ra}_{B}^{_{{\scriptstyle *}}}  \tuple{\rho',S'},$$
namely, there may be any number of silent transitions either in front of or following 
any $a_i$-transition $\stackrel{a_i}{\ra}_{B}$. Moreover, if $s\in \Act^+$ is a nonempty sequence 
of actions then $\hat{s}\in \Act^*$ denotes the possibly empty sequence of actions where all the occurrences
of $\tau$ are removed.

\begin{definition}[{\rm \textbf{\cite{palsberg}}}]
\rm
A relation $R\subseteq \Store\times \Stm \times \Stm$ is a \emph{bisimulation} when $R(\rho,S_1,S_2)$ implies:
\begin{enumerate}
\item[{\rm (1)}] if $\tuple{\rho,S_1} \raa_B \tuple{\rho',S_1'}$ then $\tuple{\rho,S_2} \Raac_B\tuple{\rho',S_2'}$,
for some $\tuple{\rho',S_2'}$ such that $R(\rho',S_1',S_2')$; 
\item[{\rm (2)}] if $\tuple{\rho,S_2} \raa_B \tuple{\rho',S_2'}$ then $\tuple{\rho,S_1}\Raac_B \tuple{\rho',S_1'}$,
for some $\tuple{\rho',S_1'}$ such that $R(\rho',S_1',S_2')$.
\end{enumerate}
$S_1$ is bisimilar to $S_2$ for a given $\rho\in \Store$, denoted by 
$S_1 \approx_\rho S_2$, if $R(\rho,S_1,S_2)$ for some bisimulation $R$. 
\qed 
\end{definition}

\noindent
Let us remark that if  $\tuple{\rho,S_1} \stackrel{\tau}{\rightarrow}\tuple{\rho',S_1'}$ then $\hat{\tau}=\epsilon$, 
so that  $\big(\tuple{\rho,S_2} \stackrel{\hat{\tau}}{\Rightarrow} \tuple{\rho,S_2}\big)\equiv \tuple{\rho,S_2}$ 
is allowed to be the matching (empty) transition sequence.

It turns out that bisimilarity can be characterized through an
abstraction of traces that observes store changes. By a negligible abuse of notation, 
the store changes function $\sch:\Trace \ra \Store^*$ defined in Section~\ref{obs} is applied
to GP traces, so that  $\sch:\Trace\cup \Trace^{\GP}  \ra \Store^*$.  
In turn, given $\rho\in \Store$, the function
$\alpha^{\rho}_{\sch}: \wp(\Trace^{\GP}) \ra \wp(\Store^*)$ is then defined as follows: 
$$\alpha^{\rho}_{\sch} (X)\ud \{\sch(\tau)\in \Store^*~|~\tau\in X,\, \exists S,\tau'\!.\: \tau = \tuple{\rho,S}\tau' \}.$$
It is worth remarking that $\alpha^{\rho}_{\sch}$ is a weaker abstraction than $\alpha_{\sch}$ defined in Section~\ref{obs}, 
that is, 
for any $X,Y\in \wp(\Trace^{\GP})$, $\alpha_{\sch}(X) = \alpha_{\sch}(Y) \Ra \alpha^{\rho}_{\sch}(X) = \alpha^{\rho}_{\sch}(Y)$
(while the converse does not hold in general). 

\begin{theorem}\label{bis-abs-th}
For any $S_1,S_2\in \Stm$, $\rho\in \Store$, we have that $S_1 \approx_\rho S_2$ iff
$\alpha^{\rho}_{\sch}(\mathbf{T}_{\GP}\grasse{S_1}) = \alpha^{\rho}_{\sch}(\mathbf{T}_{\GP}\grasse{S_2})$.
\end{theorem}
\begin{proof}
$(\Ra)$: We prove that if $R(\rho,S_1,S_2)$ holds for some bisimulation $R$ then 
$\alpha^{\rho}_{\sch}(\mathbf{T}_{\GP}\grasse{S_1}) \subseteq \alpha^{\rho}_{\sch}(\mathbf{T}_{\GP}\grasse{S_2})$ (the reverse containment is symmetric), that is, if $\sch(\tau)\in \Store^*$ for some $\tau\in \mathbf{T}_{\GP}\grasse{S_1}$ such that 
$\tau = \tuple{\rho,S_1}\tau'$ then
there exists some $\psi \in \mathbf{T}_{\GP}\grasse{S_2}$ such that $\psi = \tuple{\rho,S_2}\psi'$ and $\sch(\tau)=\sch(\psi)$. 
Let us then consider $\tau\in \mathbf{T}_{\GP}\grasse{S_1}$ such that 
$\tau = \tuple{\rho,S_1}\tau'$. If $\tau'=\epsilon$ then we pick 
$\tuple{\rho,S_2}\in \mathbf{T}_{\GP}\grasse{S_2}$ so
that $\sch(\tuple{\rho,S_1}) = \rho = \sch(\tuple{\rho,S_2})$. 
Otherwise, $\tau = \tuple{\rho,S_1}\tau'\in \mathbf{T}_{\GP}\grasse{S_1}$, with $\epsilon \neq \tau' = \tau'' \tuple{\mu,S}$.   
We prove by induction on $|\tau'|\geq 1$  
that there exists 
$\psi =\tuple{\rho,S_2}\psi'' \tuple{\mu,T} \in \mathbf{T}_{\GP}\grasse{S_2}$ such that 
$\sch(\tau) = \sch(\psi)$ and $R(\mu,S,T)$. 

\noindent
$(|\tau'|=1)$: In this case, $\tau = \tuple{\rho,S_1}\tuple{\mu,S}\in \mathbf{T}_{\GP}\grasse{S_1}$, so that 
$\tuple{\rho,S_1} \raa_B \tuple{\mu,S}$. Since, by hypothesis, $R(\rho,S_1,S_2)$ holds, we have that 
$\tuple{\rho,S_2} \Raac_B \tuple{\mu,T}$, for some $T$, and $R(\mu,S,T)$. Let $\psi\in  \mathbf{T}_{\GP}\grasse{S_2}$ be the
trace corresponding to the sequence of transitions $\tuple{\rho,S_2} \Raac_B \tuple{\mu,T}$. Then, by
definition of $\Raac_B$, we have that $\sch(\tau)= \sch(\psi)$, and, by definition of bisimulation, $R(\mu, S,T)$ holds. 

\noindent
$(|\tau'|>1)$: Here, $\tau'=\tau'' \tuple{\mu,S}$ and 
$\tau = \tuple{\rho,S_1}\tau'\in \mathbf{T}_{\GP}\grasse{S_1}$, with $|\tau''|=|\tau'|-1\geq 1$. 
Hence, $\tau'' = \tau''' \tuple{\eta,U}$. By inductive hypothesis, there exists $\psi = 
\tuple{\rho,S_2} \psi'' \tuple{\eta,V} \in \mathbf{T}_{\GP}\grasse{S_2}$ such that 
$\sch(\tuple{\rho,S_1} \tau''' \tuple{\eta,U}) = \sch(\tuple{\rho,S_2} \psi'' \tuple{\eta,V})$ and $R(\eta,U,V)$. 
Since $\tuple{\eta,U} \raa_B \tuple{\mu,S}$ and $R(\eta,U,V)$ holds, we obtain that 
$\tuple{\eta,V} \Raac_B \tuple{\mu,T}$, for some $T$, and $R(\mu,S,T)$ holds. 
Let $\tuple{\eta,V}\cdots \tuple{\mu,T}$ 
be the sequence of states corresponding to the sequence of transitions $\tuple{\eta,V} \Raac_B \tuple{\mu,T}$
so that  we pick 
$\tuple{\rho,S_2} \psi'' \tuple{\eta,V}\cdots \tuple{\mu,T}\in \mathbf{T}_{\GP}\grasse{S_2}$. The condition 
$R(\mu,S,T)$ already holds. 
Moreover, 
by definition of $\Raac_B$, we have that $\sch(\tuple{\eta,U} \tuple{\mu,S}) = \sch(\tuple{\eta,V}\cdots \tuple{\mu,T})$, 
and therefore we obtain
$\sch(\tau) = \sch(\tuple{\rho,S_1} \tau''' \tuple{\eta,U}\tuple{\mu,S})=
\sch(\tuple{\mu,S_2} \psi''\tuple{\eta,V}\cdots \tuple{\mu,T})$. 

\medskip
\noindent
$(\La)$: We first observe the following property $(*)$, which is a straight consequence of the fact that $\ra_B$ is a deterministic relation: 
If $S\in \Stm$ and $\sigma,\tau\in \tsem\grasse{S}$ are such that
$\sigma_0 = \tuple{\mu,S}= \tau_0$ and $|\tau|\leq |\sigma|$ then  
there exists some $\psi$ such that
$\sigma = \tau \psi$.

\noindent
Given $\rho \in \Store$, we 
assume that $\alpha^{\rho}_{\sch}(\mathbf{T}_{\GP}\grasse{S_1}) = \alpha^{\rho}_{\sch}(\mathbf{T}_{\GP}\grasse{S_2})$ and
we then define the following relation $R$:
\begin{align*}
R \ud \{(\rho,S_1,S_2)\} \cup \{ (\mu, T_1,T_2)~|~ & \tuple{\rho,S_1}\cdots \tuple{\mu,T_1} \in \tsem\grasse{S_1},\, 
\tuple{\rho,S_2}\cdots \tuple{\mu,T_1} \in \tsem\grasse{S_2},\\ 
& \sch(\tuple{\rho,S_1}\cdots \tuple{\mu,T_1}) = 
\sch(\tuple{\rho,S_2}\cdots \tuple{\mu,T_1})\}.
\end{align*}
We show that $R$ is a bisimulation, so that $R(\rho, S_1,S_2)$ follows. 

\noindent
(case A) Assume that $\tuple{\rho,S_1} \raa_B \tuple{\rho',S_1'}$. Then, since $\tuple{\rho,S_1}\tuple{\rho',S_1'} 
\in \tsem\grasse{S_1}$ and $\alpha^{\rho}_{\sch}(\mathbf{T}_{\GP}\grasse{S_1}) = \alpha^{\rho}_{\sch}(\mathbf{T}_{\GP}\grasse{S_2})$, we have that  there exists $\tau = \tuple{\rho,S_2}\cdots \in \tsem\grasse{S_2}$ such that 
$\sch(\tuple{\rho,S_1}\tuple{\rho',S_1'}) = \sch(\tau)$. Hence, $\tau$ necessarily has the following shape:
$$\tau = \tuple{\rho,S_2}\tuple{\rho,U_1}\cdots \tuple{\rho,U_n}\tuple{\rho',V_1}\cdots\tuple{\rho',V_m}$$
where $n\geq 0$ ($n=0$ means that $\tuple{\rho,U_1}\cdots \tuple{\rho,U_n}$ is indeed the empty sequence) and
$m\geq 1$. This therefore means that $\tuple{\rho,S_2} \Raac_B \tuple{\rho',V_m}$, so that, by
definition of $R$, $R(\rho',S_1',V_m)$ holds. 

\noindent
(case B) Assume now that $R(\mu,T_1,T_2)$ holds because $\delta=\tuple{\rho,S_1}\cdots \tuple{\mu,T_1} \in \tsem\grasse{S_1}$, 
$\sigma=\tuple{\rho,S_2}\cdots \tuple{\mu,T_2} \in \tsem\grasse{S_2}$ and $\sch(\delta) = 
\sch(\sigma)$. Hence, let us suppose that $\tuple{\mu,T_1} \raa_B \tuple{\mu',T_1'}$.
Then, since $\delta\tuple{\mu',T_1'}
\in \tsem\grasse{S_1}$ and $\alpha^{\rho}_{\sch}(\mathbf{T}_{\GP}\grasse{S_1}) = \alpha^{\rho}_{\sch}(\mathbf{T}_{\GP}\grasse{S_2})$, we have that  there exists $\tau = \tuple{\rho,S_2}\cdots \in \tsem\grasse{S_2}$ such that 
$\sch(\delta\tuple{\mu',T_1'}) = \sch(\tau)$. 
 
 \begin{description}
\item[{\rm (case B1)}] If $|\tau| \leq |\sigma|$ then, by the property $(*)$ above, $\sigma = \tau \psi$, for some $\psi$. 
Hence, $\sch(\tau)=\sch(\delta \tuple{\mu',T_1'})$ is a prefix of $\sch(\sigma)=\sch(\delta)$. Consequently, 
$\sch(\delta\tuple{\mu',T_1'})$ can be a prefix of $\sch(\delta)$ only if $\sch(\delta\tuple{\mu',T_1'}) = \sch(\delta)$, 
so that the action $a$ is $\tau$ and $\mu'=\mu$, that is, 
$\tuple{\mu,T_1}\stackrel{\tau}{\ra}_{B}\tuple{\mu,T_1'}$. We thus consider the empty transition sequence 
$\tuple{\mu,T_2} \stackrel{\hat{\tau}}{\Rightarrow} \tuple{\mu,T_2}$, so that from
$\sch(\delta\tuple{\mu,T_1'})=\sch(\sigma)$, by definition of $R$ we obtain that $R(\mu,T_1',T_2)$ holds.  

\item[{\rm (case B2)}] If $|\tau|> |\sigma|$ then, by $(*)$ above, $\tau = \sigma \psi$, for some $\psi$, i.e.,
$\tau =  \sigma \cdots \tuple{\mu'',T_2'}$, for some 
$\mu''$ and $T_2'$. Since $\sch(\tuple{\rho,S_2}\cdots \tuple{\mu,T_2}) = \sch(\tuple{\rho,S_1}\cdots \tuple{\mu,T_1})$ and $\sch(\tuple{\rho,S_2}\cdots \tuple{\mu,T_2}\cdots \tuple{\mu'',T_2'})=\sch(\tuple{\rho,S_1}\cdots \tuple{\mu,T_1}\tuple{\mu',T_1'})$, we derive
that $\mu''=\mu'$ and $\tuple{\mu,T_2} \stackrel{a}{\Rightarrow} \tuple{\mu''=\mu',T_2'}$. By definition of $R$, 
$R(\mu', T_1', T_2')$ holds.
\end{description}
This closes the proof.
\end{proof}

\subsection{Hot Paths}\label{hotpaths-sec}
Let us recall the set of rules that define the tracing transitions in Guo and Palsberg~\citeyear{palsberg} model. Let 
$\TState_{\GP} \ud \Store\times \Stm \times \Stm \times \Stm$ denote the set of states in trace recording mode, whose
components are, respectively, the current store, the entry point of the recorded trace (this is always a while statement), 
the current trace (i.e., a sequence of commands) 
and the current program to be evaluated. In turn, $\State_{\GP}^e \ud \State_{\GP} \cup \TState_{\GP}$ denotes 
the corresponding extended notion of state, which encompasses the trace recording mode. Then, the relation
$\ra_T\: \subseteq \State^e_{\GP} \times \State^e_{\GP}$ is defined 
by the clauses in Figure~\ref{trace-rel}, where $O:\Stm\times \Store \ra \Stm$ is a ``sound'' optimization function that depends 
on a given store. Correspondingly, the trace semantics $\tsem_{\GP}\grasse{S}\subseteq (\State_{\GP}^e)^+$ 
of a program $S\in \Stm$ is naturally extended to the relation $\ra_{B,T}\; \ud\; \ra_B \cup
\ra_T\; \subseteq \State_{\GP}^e\times \State_{\GP}^e$.

\begin{figure}
{\small
\begin{equation*}
\begin{array}{l}
(T_1)\quad \tuple{\rho, (\KwSty{if}~B~\KwSty{then}~(S\, \KwSty{while}~B~\KwSty{do}~S))\, K } \ra_T 
\tuple{\rho, (\KwSty{while}~B~\KwSty{do}~S) K, \epsilon, S (\KwSty{while}~B~\KwSty{do}~S) K}\\
\hspace*{80ex} \text{if~}\bsem\grasse{B}\rho=\textit{true}
\\
(T_2)\quad \tuple{\rho, K_w, t, \KwSty{skip}; K} \ra_T \tuple{\rho,K_w, t (\KwSty{skip};\!), K}   \\[5pt]
(T_3)\quad \tuple{\rho, K_w, t, x:=E; K} \ra_T \tuple{\rho[x/\esem\grasse{E}\rho],K_w, t (x:=E;\!), K}   \\[5pt]
(T_4)\quad \tuple{\rho, K_w, t, (\KwSty{if}~B~\KwSty{then}~S) K} \ra_T 
\begin{cases}
\tuple{\rho,K_w, t (\KwSty{bail}~B~\KwSty{to}~(SK)),K} 
& \text{if}~\bsem\grasse{B}\rho = \textit{false}\\
\tuple{\rho,K_w, t (\KwSty{bail}~\neg B~\KwSty{to}~K),SK} & 
\text{if}~\bsem\grasse{B}\rho = \textit{true}\\[5pt]
\end{cases} \\
(T_5)\quad \tuple{\rho, K_w, t, (\KwSty{while}~B~\KwSty{do}~S) K} \ra_T 
\begin{cases}
\tuple{\rho,K_w, t (\KwSty{skip};), (\KwSty{if}~B~\KwSty{then}~(S\, \KwSty{while}~B~\KwSty{do}~S))\, K} &\\
\qquad\qquad\qquad\qquad\qquad\qquad\quad \text{if}~K_w \not\equiv (\KwSty{while}~B~\KwSty{do}~S) K\\
\tuple{\rho,O(\KwSty{while}~B~\KwSty{do}~t,\rho) K} 
\hfill  \text{if}~K_w \equiv (\KwSty{while}~B~\KwSty{do}~S) K\\[5pt]
\end{cases} \\
(T_6)\quad \tuple{\rho, K_w, t, S} \ra_T \tuple{\rho',S'}\quad \text{if~} K_w \not\equiv S \text{~and~} \tuple{\rho,S}\ra_B \tuple{\rho',S'}\\[5pt]
\end{array}	 
\end{equation*}
}
\caption{Definition of the tracing relation $\ra_T$.}\label{trace-rel}
\end{figure}

Let us notice that in Guo and Palsberg's model of hot paths:

\begin{enumerate}
\item[{\rm (i)}] By clause $(T_1)$, trace recording is always triggered by 
an unfolded while loop, and the loop itself is 
not included in the hot path.
\item[{\rm (ii)}] By clause $(T_4)$, when we bail out of a hot path $t$ 
through a $\KwSty{bail}$ 
command, we cannot anymore re-enter into $t$. 
\item[{\rm (iii)}] By clause $(T_5)$---the second condition of this clause is called stitch rule 
in \cite{palsberg}---the store used to optimize a hot path $t$ is recorded at the end of the 
first loop iteration. This is a concrete store which is used by $O$ to optimize the stitched 
hot path $\KwSty{while}~B~\KwSty{do}~t$.  
\item[{\rm (iv)}] Hot paths actually are $1$-hot paths according
to our definition, since, by clause~$(T_1)$, once the first iteration of the traced while loop 
is terminated, trace recording necessarily discontinues.
\item[{\rm (v)}] There are no clauses for trace recording \KwSty{bail} commands. Hence, when 
trying to trace a loop that already contains a nested hot path, by clause $(T_6)$, trace recording is aborted
when a \KwSty{bail} command is encountered. In other terms, in contrast to our approach described in
Section~\ref{mte-sec}, nested hot paths are not allowed. 
\item[{\rm (vi)}] Observe that when tracing a loop $\KwSty{while}~B~\KwSty{do}~S$ whose body $S$ does not contain
branching commands, i.e.\ $\KwSty{if}$ or $\KwSty{while}$ statements, it turns out 
that the hot path $t$ coincides with the body $S$, so 
that  $\KwSty{while}~B~\KwSty{do}~t \equiv \KwSty{while}~B~\KwSty{do}~S$,  namely, in this case 
the hot path transform does not change the subject while loop. 
\end{enumerate}

In the following, we show how this hot path extraction model
can be formalized within our trace-based approach. To this aim,  
we do not consider optimizations of hot paths, which is an orthogonal issue here, so that
we assume that $O$ performs no optimization, that is, 
$O(\KwSty{while}~B~\KwSty{do}~t,\rho)= \KwSty{while}~B~\KwSty{do}~t$. 

A sequence of commands $t\in \Stm$ is defined to be a GP  hot path for a program $Q\in \Stm$ when 
we have the following transition sequence:  
$$
\tuple{\rho,Q} \ra_{B,T}^* 
\tuple{\rho', (\KwSty{while}~B~\KwSty{do}~S)K} 
\ra_{B,T}^*
\tuple{\rho'', (\KwSty{while}~B~\KwSty{do}~S)K, t, (\KwSty{while}~B~\KwSty{do}~S)K}.
$$
Since the operational semantics $\ra_{B,T}$ is given in continuation-style, without loss of generality, 
we assume that the program $Q$ begins with a while statement, that is $Q\equiv (\KwSty{while}~B~\KwSty{do}~S)K$. 
Guo and Palsberg's hot loops can be modeled in our framework 
by exploiting a revised loop selection map 
$\sloop_{\GP}:\Trace \ra \wp(\mathbb{C}^+)$ defined as follows:
\begin{multline*}
    \sloop_{\GP}(\ab{\rho_{0}}{C_{0}} \cdots
    \ab{\rho_{n}}{C_{n}}) \ud
    \big\{C_i C_{i+1}\cdots C_j~|~
    0\leq i \leq j < n,\,C_{i} \lessdot C_{j},\\
     suc(C_{j})=lbl(C_i),\,
       \forall k\in (i,j].\,
       C_{k} \not\in\{ C_{i},cmpl(C_{i})\}\big\}.
 \end{multline*}
       
\noindent
Thus, $\sloop_{\GP}(\tau)$ contains sequences of commands without store.     
The map $\alpha_{\mathit{hot}}^{\GP}: \wp(\Trace) \ra \wp(\mathbb{C}^+)$ 
then lifts $\sloop_{\GP}$ to sets of traces as usual: $\alpha_{\mathit{hot}}^{\GP}(T) \ud \cup_{\tau \in T} 
\sloop_{\GP}(\tau)$. 
Then, let us consider a GP hot path $t$ as recorded by 
a transition sequence $\tau$: 
\[
\begin{array}{ll}
\tau \ud & \tuple{\rho,S_0\equiv (\KwSty{while}~B~\KwSty{do}~S)K} \ra_B\\
&\tuple{\rho,S_{1}\equiv (\KwSty{if}~B~\KwSty{then}~(S\, \KwSty{while}~B~\KwSty{do}~S))\, K}\ra_T\\ 
&\tuple{\rho, (\KwSty{while}~B~\KwSty{do}~S) K, \epsilon, S_{2}\equiv S (\KwSty{while}~B~\KwSty{do}~S) K}\ra_T  \\
&\: \cdots \ra_T\\
&\tuple{\rho', (\KwSty{while}~B~\KwSty{do}~S)K, t', S_n}\ra_T \\
&\tuple{\rho'', (\KwSty{while}~B~\KwSty{do}~S)K, t, S_{n+1}\equiv (\KwSty{while}~B~\KwSty{do}~S)K}
\end{array}\eqno(\ddagger)
\]
where $\bsem\grasse{B}\rho = \textit{true}$. 
Hence, the $S_i$'s occurring in $\tau$ are the current statements to be evaluated. 
With a negligible abuse of notation, we assume that $\tau \in \tsem_{\GP}\grasse{(\KwSty{while}~B~\KwSty{do}~S)K}$, that is, 
the arrow symbols $\ra_B$ and $\ra_T$ are taken out of the sequence $\tau$. 
By Lemma~\ref{sc-tc-lemma}~(2), 
we therefore consider the corresponding execution trace $\traceC(\tau)$ of the compiled program $\cC((\KwSty{while}~B~\KwSty{do}~S)K)$, where 
the state compile function $\stateC$ in Figure~\ref{scfun}, 
when applied
to states in trace recording mode, is assumed to act on the current store and the 
program to be evaluated, that is, $\stateC(\tuple{\rho,K_w,t,S}) = \stateC(\tuple{\rho,S})$. We thus obtain:
\[
\begin{array}{ll}
\traceC(\tau) \ud &\tuple{\rho,C_0\equiv \elb((\KwSty{while}~B~\KwSty{do}~S)K): \cd{skip} 
\ra \elb((\KwSty{if}~B~\KwSty{then}~(S\, \KwSty{while}~B~\KwSty{do}~S))\, K)}\\ 
&\tuple{\rho,C_{1}\equiv \elb((\KwSty{if}~B~\KwSty{then}~(S\, \KwSty{while}~B~\KwSty{do}~S))\, K):B\ra \elb(S\, (\KwSty{while}~B~\KwSty{do}~S) K)}\\ 
&\tuple{\rho, C_{2}\equiv \elb(S (\KwSty{while}~B~\KwSty{do}~S) K):A_{2} \ra \elb(T)}\\
&\cdots\\    
&\tuple{\rho', C_{n}\equiv \elb(S_n): A_n
\ra \elb((\KwSty{while}~B~\KwSty{do}~S)\, K)}\\
&\tuple{\rho'', C_{n+1}\equiv \elb((\KwSty{while}~B~\KwSty{do}~S)K): \cd{skip} 
\ra \elb((\KwSty{if}~B~\KwSty{then}~(S\, \KwSty{while}~B~\KwSty{do}~S))\, K)}.
\end{array}
\]
We therefore obtain a hot path $hp_t=C_0C_1 \cdots C_n\in \sloop_{\GP}(\traceC(\tau))$, i.e.\ 
$hp_t\in \alpha_{\mathit{hot}}^{\GP}( \mathbf{T}^\iota \grasse{\cC((\KwSty{while}~B~\KwSty{do}~S)K)})$,
 where 
$lbl(C_0) = \elb ((\KwSty{while}~B~\KwSty{do}~S)K) = suc(C_n)$. This is a consequence of the fact that 
for all $k\in (0,n]$, $C_k$ cannot be the entry command 
$C_0$ or its complement command, because, by the stitch rule of clause $(T_5)$, 
$S_{n+1}$ is necessarily the first occurrence 
of $(\KwSty{while}~B~\KwSty{do}~S))\, K$ as current program to be evaluated in the trace $\tau$, so that, 
for any $k\in (0,n]$,  
$lbl(C_k) \neq \elb((\KwSty{while}~B~\KwSty{do}~S)K)$. 
We have thus shown that any GP hot path arising from a trace $\tau$ generates a corresponding
hot path extracted by our
selection map $\sloop_{\GP}$ on the compiled trace $\traceC(\tau)$:  

\begin{lemma}\label{lemma-GPhp}
Let $Q_w \equiv (\KwSty{while}~B~\KwSty{do}~S)K$. 
If $t$ is a GP hot path for $Q_w$ where $\tau \equiv \tuple{\rho,Q_w}
\ra_{B,T}^* \tuple{\rho', Q_w, t, Q_w}$ is the transition sequence $(\ddagger)$ that records $t$, 
then there exists a hot path 
$hp_t = C_0C_1\cdots C_n \in \alpha_{\mathit{hot}}^{\GP}( \mathbf{T}^\iota \grasse{\cC(Q_w)})$ such that, 
for any $i\in [0,n]$,  
{\rm $lbl(C_i) = \elb(S_i)$}, and, in particular,
{\rm $lbl(C_0) = \elb(Q_w) = suc(C_n)$}. 
\end{lemma}

\begin{example}\label{trace-ex1}
Let us consider the while statement $Q_w$ of the program in Example~\ref{example-init}: 
\[
Q_w\equiv \KwSty{while}~(x\leq 20)~\KwSty{do}~(x:=x+1;~ 
(\KwSty{if}~(x\%3=0)~\KwSty{then}~x:=x+3;))
\]

\noindent
This program is already written in 
Guo and Palsberg language, so that $Q_w$ is a well formed statement in $\Stm$. 
The tracing rules in Figure~\ref{trace-rel} yield the following trace $t$ for $Q_w$:
\[
t\equiv x:=x+1;~
\KwSty{bail}~(x\%3 =0)~\KwSty{to}~(x:=x+3;~Q_w).
\]
On the other hand, the compiled program $\cC(Q_w)\in \wp(\mathbb{C})$ is as follows: 
\begin{equation*}
  \begin{aligned}
    \cC(Q_w) = \big\{ &D_0\equiv \elb_{\KwSty{while}}: \cd{skip} \ra \elb_{\KwSty{ifwhile}},\, \\
    &D_1\equiv \elb_{\KwSty{ifwhile}}: (x\leq 20) \ra \elb_{1},\, D_1^c \equiv \elb_{\KwSty{ifwhile}}: \neg (x\leq 20) \ra \elb_\epsilon,\,\\
    & D_2 \equiv \elb_1: x:=x+1 \ra \elb_{\KwSty{if}},\\
    &D_3\equiv \elb_{\KwSty{if}}: (x\% 3 =0) \ra \elb_2,\, D_3^c\equiv \elb_{\KwSty{if}}: \neg (x\% 3 =0) \ra\elb_{\KwSty{while}},\\
    & D_4\equiv \elb_2: x:= x+3 \ra \elb_{\KwSty{while}},\,D_5\equiv \elb_\epsilon : \cd{skip}\ra \L\big\},
   \end{aligned}
\end{equation*}
where labels have 
the following meaning:
\begin{align*}
\elb_{\KwSty{while}} &\ud \elb(Q_w)\\
\elb_{\KwSty{ifwhile}} &\ud \elb(\KwSty{if}~(x\leq 20)~\KwSty{then}~(x:=x+1; (\KwSty{if}~(x\%3=0)~\KwSty{then}~x:=x+3;)\: Q_w))\\
\elb_{1} &\ud \elb(x:=x+1; (\KwSty{if}~(x\%3=0)~\KwSty{then}~x:=x+3;)\: Q_w))\\
\elb_{\KwSty{if}} &\ud \elb((\KwSty{if}~(x\%3=0)~\KwSty{then}~x:=x+3;)\: Q_w))\\
\elb_2 &\ud \elb(x:=x+3; Q_w).
\end{align*}

\noindent
Hence, in correspondence with the trace $t$,  we obtain the hot path 
$hp_t = D_0D_1 D_2 D_3^c \in \alpha_{\mathit{hot}}^{\GP}( \mathbf{T}^\iota \grasse{\cC(Q_w)})$. 
In turn, this hot path $hp_t$ 
corresponds to the 2-hot path $hp_1$ consisting of the analogous sequence of commands, which has been 
selected in Example~\ref{ex-zero}. \qed
\end{example}

\subsection{GP Trace Extraction}\label{GPte-sec}
In the following, we conform to the notation used in 
Section~\ref{trace-ext-sec} for our trace extraction transform. 
Let us consider a while program $Q_w\equiv (\KwSty{while}~B~\KwSty{do}~S)K \in \Stm$ and its
compilation $P_w\ud \cC(Q_w)\in \wp(\mathbb{C})$. Observe that, by Definition~\ref{compilation} of
compilation $\cC$, a hot path  $C_0\cdots C_n \in 
\alpha_{\mathit{hot}}^{\GP}( \mathbf{T}\grasse{P_w})$ for the compiled program $P_w$ always arises in correspondence 
with some while loop $\KwSty{while}~B'~\KwSty{do}~S'$ occurring in $Q_w$ and therefore has
necessarily the following shape:  
\begin{align*}
C_{0} &\equiv \elb((\KwSty{while}~B'~\KwSty{do}~S')J): \cd{skip} 
\ra \elb((\KwSty{if}~B'~\KwSty{then}~(S'\, \KwSty{while}~B'~\KwSty{do}~S'))\, J)\\ 
C_{1} & \equiv \elb((\KwSty{if}~B'~\KwSty{then}~(S'\, (\KwSty{while}~B'~\KwSty{do}~S')))\, J):B'\ra \elb(S'\, (\KwSty{while}~B'~\KwSty{do}~S')\, J)\\ 
C_{2} & \equiv \elb(S'\, (\KwSty{while}~B'~\KwSty{do}~S')\, J):A_{2} \ra \elb(T_3)\\
&\cdots\\    
C_{n} & \equiv \elb(T_n): A_n \ra \elb((\KwSty{while}~B'~\KwSty{do}~S')J)
\end{align*}
The GP hot path extraction scheme for $Q_w$ described by the rules in Figure~\ref{trace-rel}  
can be defined in our language by the following simple transform of $P_w$. 

\begin{definition}[\textbf{GP trace extraction transform}]\label{gp-tt-def}
\rm
The \emph{GP trace extraction transform} $\extr_{hp}^{\GP}(P_w)$ of $P_w$ for the hot path $hp= C_0 C_1\cdots C_n\in \alpha_{\mathit{hot}}^{\GP}( \mathbf{T}\grasse{P_w})$ is defined as follows:

\begin{enumerate}
\item[{\rm (1)}] If for any $i\in [2,n]$, $cmpl(C_i)\not\in P_w$ then  $\extr_{hp}^{\GP}(P_w) \ud P_w$;
\item[{\rm (2)}] Otherwise:
\belowdisplayskip=-10pt
\begin{equation*}
  \begin{aligned}
    \extr_{hp}^{\GP}(P_w) \ud 
    P_w  & \cup \{\ell_i : act(C_{i}) \rightarrow \ell_{\textit{next}(i)}~|~
    i\in [0,n]\}\\
    &\cup \{\ell_i : \neg act(C_i) \rightarrow L_{\textit{next}(i)}^c ~|~ i\in [0,n],\,  cmpl(C_{i})\in P_w\}. 
   \end{aligned}
\end{equation*}
\hfill\ensuremath{\qed}
\end{enumerate}
\end{definition}
    
Clearly, $\extr_{hp}^{\GP}(P)$ remains a well-formed program. Also observe
that the case~(1) of Definition~\ref{gp-tt-def} means that the traced hot path $hp$
does not contain conditional commands (except from the entry conditional $C_1$) 
and therefore corresponds to point~(vi) 
in Section~\ref{hotpaths-sec}.

\begin{example}
\rm
Let us consider the programs $Q_w$ and $\cC(Q_w)$ of Example~\ref{trace-ex1} 
and the hot path 
$hp_t = D_0 D_1 D_2 D_3^c \in 
\alpha_{\mathit{hot}}^{\GP}( \mathbf{T}\grasse{\cC(Q_w)})$ which corresponds to  
the trace $t\equiv x:=x+1;~
\KwSty{bail}~(x\%3 =0)~\KwSty{to}~(x:=x+3;~Q_w)$ of $Q_w$. 
Here, the GP trace extraction of $hp_t$, according to Definition~\ref{gp-tt-def},
provides the following
program transform: 
\begin{equation*}
  \begin{aligned}
    \extr_{hp_t}^{\GP}(\cC(Q_w)) 
    \ud \big\{ &
    D_0\equiv \elb_{\KwSty{while}}: \cd{skip} \ra \elb_{\KwSty{ifwhile}},\,
    D_1\equiv \elb_{\KwSty{ifwhile}}: (x\leq 20) \ra \elb_{1},\\ 
    & D_1^c \equiv \elb_{\KwSty{ifwhile}}: \neg (x\leq 20) \ra \elb_\epsilon,\,
    D_2 \equiv \elb_1: x:=x+1 \ra \elb_{\KwSty{if}},\, \\ 
    & D_3 \equiv \elb_{\KwSty{if}}: (x\% 3 =0) \ra \elb_2,\, 
    D_3^c \equiv \elb_{\KwSty{if}}: \neg (x\% 3 =0) \ra\elb_{\KwSty{while}},\\
    & D_4 \equiv \elb_2: x:= x+3 \ra \elb_{\KwSty{while}},\, D_5 \equiv \elb_\epsilon : \cd{skip}\ra \L\big\} \:\cup \\
    & \!\!\!\big\{ \ell_0: \cd{skip} \ra \ell_1,\, \ell_1 : x\leq 20 \ra \ell_2,\, \ell_1: \neg(x\leq 20) \ra \elb_\epsilon,\\ 
    &\ell_2 : x:=x+1 \ra \ell_3,\, \ell_3: \neg (x\%3 = 0) \ra \ell_0,\, \ell_3 : (x\%3 =0) \ra \elb_{2} \big\}.         
    \end{aligned} 
\end{equation*}
On the other hand, the stitch rule $(T_5)$ transforms $Q_w$ into the following program $Q_{t}$:   
\[
\begin{array}{ll}
\KwSty{while}&\!\!(x\leq 20)~\KwSty{do}\\
&x:=x+1;\\ 
&\KwSty{bail}~(x\%3 =0)~\KwSty{to}~ (x:=x+3;\: Q_w)
\end{array}
\]
whose compilation yields the following program: 
\begin{equation*}
  \begin{aligned}
    \cC(Q_t) = \big\{ &
     \elb_{\KwSty{while}_t}: \cd{skip} \ra \elb_{\KwSty{ifwhile}_t},\, \elb_{\KwSty{ifwhile}_t}: (x\leq 20) \ra \elb_{1t},\, \elb_{\KwSty{ifwhile}_t}: \neg (x\leq 20) \ra \elb_\epsilon,\, \\
    &\elb_{1t}: x:=x+1 \ra \elb_{\KwSty{bail}},\, \elb_{\KwSty{bail}}: (x\% 3 =0) \ra \elb_{\KwSty{bail}_{\textit{true}}},\, \elb_{\KwSty{bail}}: \neg (x\% 3 =0) \ra\elb_{\KwSty{while}_t},\\
    &      \elb_{\KwSty{while}}: \cd{skip} \ra \elb_{\KwSty{ifwhile}},\, 
    \elb_{\KwSty{ifwhile}}: (x\leq 20) \ra \elb_{1},\,  \elb_{\KwSty{ifwhile}}: \neg (x\leq 20) \ra \elb_\epsilon,\, \\
    & \elb_{1}: x:=x+1 \ra \elb_{\KwSty{if}},\, \elb_{\KwSty{if}}: (x\% 3 =0) \ra \elb_{2},\,  \elb_{\KwSty{if}}: \neg (x\% 3 =0) \ra\elb_{\KwSty{while}},\, \\
     & \elb_{\KwSty{bail}_{\textit{true}}}: x:= x+3 \ra \elb_{\KwSty{while}},\,  \elb_\epsilon : \cd{skip}\ra \L\big\}
   \end{aligned}
\end{equation*}
with the following new labels:
\begin{align*}
\elb_{\KwSty{while}_t} &\ud \elb(\KwSty{while}~(x\leq 20)~t)\\
\elb_{\KwSty{ifwhile}_t} &\ud \elb(\KwSty{if}~(x\leq 20)~\KwSty{then}~(t \;(\KwSty{while}~(x\leq 20) ~t)))\\
\elb_{1t} &\ud \elb(t \;(\KwSty{while}~(x\leq 20) ~t))\\
\elb_{\KwSty{bail}} &\ud \elb((\KwSty{bail}~(x\%3=0)~\KwSty{to}~(x:=x+3; Q_w))(\KwSty{while}~(x\leq 20) ~t) )
\end{align*}
while observe that $\elb_{\KwSty{bail}_{\textit{true}}} \ud \elb(x:=x+3; Q_w) = \elb_2$. 
It is then immediate to check that the programs $\cC(Q_t)$ and $\extr_{hp_t}^{\GP} (\cC(Q_w))$ are equal up to the 
following label renaming of $\extr_{hp_t}^{\GP}(\cC(Q_w))$: 
\belowdisplayskip=-10pt
\begin{align*}
\{\ell_0 \mapsto  \elb_{\KwSty{while}_t}, \ell_1 \mapsto \elb_{\KwSty{ifwhile}_t}, \ell_2 \mapsto \elb_{1t}, 
\ell_3\mapsto \elb_{\KwSty{bail}} \}. \qed
\end{align*}
\end{example}

The equivalence 
of this GP trace extraction with the stitch of hot paths by \citeN{palsberg} 
goes as follows.

\begin{theorem}[\textbf{Equivalence with GP trace extraction}]\label{GP-equiv-th}
Let $t$ be a GP trace 
such that $\tuple{\rho,(\KwSty{while}~B~\KwSty{do}~S)K}
\ra_{B,T}^* \tuple{\rho', (\KwSty{while}~B~\KwSty{do}~S)K, t, (\KwSty{while}~B~\KwSty{do}~S)K}$ and 
let $hp_t \in \alpha_{\mathit{hot}}^{\GP}( \mathbf{T}^\iota \grasse{\cC((\KwSty{while}~B~\KwSty{do}~S)K)})$ be the corresponding 
GP hot path as determined by Lemma~\ref{lemma-GPhp}.
Then, $\cC((\KwSty{while}~B~\KwSty{do}~t)K) \cong \extr_{hp_t}^{\GP}(\cC((\KwSty{while}~B~\KwSty{do}~S)K))$.
\end{theorem}
\begin{proof}
Let the GP hot path $t$ be recorded by 
the following transition sequence for $\tuple{\rho,(\KwSty{while}~B~\KwSty{do}~S)K}$: 
\[
\begin{array}{ll}
& \tuple{\rho,S_{-2}\equiv (\KwSty{while}~B~\KwSty{do}~S)K} \ra_B\\
&\tuple{\rho,S_{-1}\equiv (\KwSty{if}~B~\KwSty{then}~(S\, \KwSty{while}~B~\KwSty{do}~S))\, K}\ra_T \qquad\qquad [\text{with}~\bsem\grasse{B}\rho = \textit{true}] \\  
&\tuple{\rho_0 \ud \rho, (\KwSty{while}~B~\KwSty{do}~S) K, t_0\equiv \epsilon, S_{0}\equiv S (\KwSty{while}~B~\KwSty{do}~S) K}\ra_T  \\
&\tuple{\rho_1, (\KwSty{while}~B~\KwSty{do}~S) K, t_1\equiv c_1, S_{1}}\ra_T  \\
&\: \cdots \ra_T\\
&\tuple{\rho_n, (\KwSty{while}~B~\KwSty{do}~S)K, t_n \equiv t_{n-1} c_n, S_n}\ra_T \\
&\tuple{\rho_{n+1}\ud\rho', (\KwSty{while}~B~\KwSty{do}~S)K, t \equiv t_n c_{n+1}, S_{n+1}\equiv (\KwSty{while}~B~\KwSty{do}~S)K}
\end{array}
\]
where $n\geq 0$, so that 
the body $S$ is assumed to be nonempty, i.e., $S\neq \epsilon$ (there is no loss of generality since 
for $S=\epsilon$  the result trivially holds). 
Hence, $t=c_1...c_n c_{n+1}$, for some commands $c_i\in\Cmd$, and the corresponding hot path $hp_t \equiv H_{-2} H_{-1} H_0 ... H_n$ 
as determined by Lemma~\ref{lemma-GPhp} is as follows:
\begin{align*}
H_{-2} &\ud \elb(S_{-2}): \cd{skip} 
\ra \elb(S_{-1})\\ 
H_{-1} & \ud \elb(S_{-1}):B\ra \elb(S_0)\qquad\qquad [\text{because}~\bsem\grasse{B}\rho = \textit{true}]\\ 
H_{0} & \ud \elb(S_0):A_{0} \ra \elb(S_1)\\
H_{1} & \ud \elb(S_1):A_{1} \ra \elb(S_2)\\
&\cdots\\    
H_{n} & \ud \elb(S_n): A_n
\ra \elb(S_{-2})
\end{align*}
where the action $A_i$, with $i\in [0,n]$, and the command $c_{i+1}$ depend 
on the first command of the statement $S_i$ as follows (this range of cases will be later referred to as $(*)$):
\begin{itemize}
\item[{\rm (1)}] $S_i \equiv \KwSty{skip}; J \quad \Ra \quad A_i \equiv \cd{skip}\;\:\&\;\: c_{i+1}\equiv \KwSty{skip}; \;\:\&\;\:
S_{i+1}\equiv J$
\item[{\rm (2)}]
$S_i \equiv x:=E; J \quad \Ra \quad A_i \equiv x:=E\;\:\&\;\: c_{i+1}\equiv x:=E; \;\:\&\;\:
S_{i+1}\equiv J$ 

\item[{\rm (3)}]
$S_i \equiv (\KwSty{if}~B'~\KwSty{then}~S')J \;\:\&\;\: \bsem\grasse{B'}\rho_i = \textit{true} 
\quad \Ra$\\ 
\hspace*{38ex}
$A_i \equiv B'\;\:\&\;\: 
c_{i+1}\equiv \KwSty{bail}~\neg B'~\KwSty{to}~J  \;\:\&\;\: S_{i+1}\equiv S'J$ 

\item[{\rm (4)}]
$S_i \equiv (\KwSty{if}~B'~\KwSty{then}~S')J \;\:\&\;\: \bsem\grasse{B'}\rho_i = \textit{false} 
\quad \Ra$\\
\hspace*{37ex}
$A_i \equiv \neg B' \;\:\&\;\: 
c_{i+1}\equiv \KwSty{bail}~ B'~\KwSty{to}~(S'J) \;\:\&\;\:
S_{i+1}\equiv J$ 

\item[{\rm (5)}]
$S_i \equiv (\KwSty{while}~B'~\KwSty{do}~S')J \;\:\&\;\: (\KwSty{while}~B'~\KwSty{do}~S')J\neq (\KwSty{while}~B~\KwSty{do}~S)K
\quad \Ra$\\  
\hspace*{12ex}
$A_i \equiv \cd{skip}\;\:\&\;\: 
c_{i+1}\equiv \KwSty{skip}; \;\:\&\;\:
S_{i+1}\equiv (\KwSty{if}~B'~\KwSty{then}~(S' (\KwSty{while}~B'~\KwSty{do}~S')))J$ 
\end{itemize}
\noindent 
If, for any $i\in [0,n]$, $H_i$ is not a conditional command then, by case~(1) of Definition~\ref{gp-tt-def}, 
we have that $\extr_{hp_t}^{\GP}(\cC((\KwSty{while}~B~\KwSty{do}~S)K)) = \cC((\KwSty{while}~B~\KwSty{do}~S)K)$. 
Also, for any $i\in [0,n]$, $A_i$ is either a skip or an assignment, so that $c_{i+1} = A_i$, and, in turn, 
$t=S$. Hence, $(\KwSty{while}~B~\KwSty{do}~t)K \equiv (\KwSty{while}~B~\KwSty{do}~S)K$, so that the thesis  
follows trivially. 

\noindent
Thus, we assume that $H_k$, with $k\in [0,n]$, is the first conditional command occuring in the sequence $H_0...H_n$. 
Case~(2) of Definition~\ref{gp-tt-def} applies, so that:
\begin{align*}
\extr_{hp_t}^{\GP}(\cC((\KwSty{while}~B~\KwSty{do}~S)K)) = &\; \cC((\KwSty{while}~B~\KwSty{do}~S)K)\, \cup \\
&\{\ell_{-2}: \cd{skip} \ra \ell_{-1},\, \ell_{-1}: B \ra \ell_0,\, \ell_{-1}: \neg B \ra \elb(K),\,\\
&\; \ell_{0} : A_0 \ra \ell_{1}, ..., \ell_{n} : A_n \ra \ell_{-2}\} \, \cup \\
& \{ \ell_i : \neg A_i \ra \elb(S_{\textit{next}(i)})^c~|~ i\in [0,n],\, A_i \in \BExp\}.
\end{align*} 
Moreover, we have that:
\begin{align*}
    \cC((\KwSty{while}~B~\KwSty{do}~S)K) &=\\ 
    &\mkern-40mu \big\{
     \elb ((\KwSty{while}~B~\KwSty{do}~S)K) : \cd{skip} \ra \elb ((\KwSty{if}~B~\KwSty{then}~(S\, (\KwSty{while}~B~\KwSty{do}~S)))\, K),\\
     &\mkern-30mu \elb ((\KwSty{if}~B~\KwSty{then}~(S\, (\KwSty{while}~B~\KwSty{do}~S)))\, K): B \ra \elb(S(\KwSty{while}~B~\KwSty{do}~S)K),\\
      &\mkern-30mu \elb ((\KwSty{if}~B~\KwSty{then}~(S\, (\KwSty{while}~B~\KwSty{do}~S)))\, K): \neg B \ra \elb(K)\big\}\\ 
      &\mkern-55mu \cup \cC(S (\KwSty{while}~B~\KwSty{do}~S) K) 
      \cup \cC(K)
\end{align*}
\begin{align*}
    \cC((\KwSty{while}~B~\KwSty{do}~t)K) & =\\ 
    & \mkern-40mu \big\{
     \elb ((\KwSty{while}~B~\KwSty{do}~t)K) : \cd{skip} \ra \elb ((\KwSty{if}~B~\KwSty{then}~(t\, (\KwSty{while}~B~\KwSty{do}~t)))\, K),\\
     & \mkern-30mu\elb ((\KwSty{if}~B~\KwSty{then}~(t\, (\KwSty{while}~B~\KwSty{do}~t)))\, K): B \ra \elb(t(\KwSty{while}~B~\KwSty{do}~t)K),\\
      &\mkern-30mu\elb ((\KwSty{if}~B~\KwSty{then}~(t\, (\KwSty{while}~B~\KwSty{do}~t)))\, K): \neg B \ra \elb(K)\big\}\\ 
      &\mkern-55mu \cup \cC(t (\KwSty{while}~B~\KwSty{do}~t) K) 
      \cup \cC(K)
\end{align*}
We first show that
$\cC((\KwSty{while}~B~\KwSty{do}~t)K) \subseteq_{_{\!/\cong}} \extr_{hp_t}^{\GP}(\cC((\KwSty{while}~B~\KwSty{do}~S)K))$.
We consider the following label renaming: 
\begin{align*}
 \elb ((\KwSty{while}~B~\KwSty{do}~t)K) &\mapsto \ell_{-2}\\
  \elb ((\KwSty{if}~B~\KwSty{then}~(t\, (\KwSty{while}~B~\KwSty{do}~t)))\, K) &\mapsto \ell_{-1} \\
   \elb(t(\KwSty{while}~B~\KwSty{do}~t)K)  &\mapsto \ell_{0} 
\end{align*} 
so that it remains to 
show that $\cC(t (\KwSty{while}~B~\KwSty{do}~t) K) \subseteq_{_{\!/\cong}} \extr_{hp_t}^{\GP}(\cC((\KwSty{while}~B~\KwSty{do}~S)K))$. 
Since $t=c_1 t'$, with $t' = c_2 ...c_{n+1}$, let us analyze the five different cases for the first command $c_1$ of $t$. 
\begin{itemize}
\item[{\rm (i)}] $c_{1}\equiv x:=E;$. Thus, $S_0 \equiv x:=E; T (\KwSty{while}~B~\KwSty{do}~S)K$, $S_1 \equiv T (\KwSty{while}~B~\KwSty{do}~S)K$, 
$A_0 \equiv x:=E$. In this case, 
\begin{align*}
\cC(t (\KwSty{while}~B~\KwSty{do}~t) K) = \:
&\{\elb(t (\KwSty{while}~B~\KwSty{do}~t) K): x:=E \ra \elb(t' (\KwSty{while}~B~\KwSty{do}~t) K)\}\\ 
& \cup  \cC(t' (\KwSty{while}~B~\KwSty{do}~t) K).
\end{align*}
Hence, it is enough to consider the relabeling 
$\elb(t' (\KwSty{while}~B~\KwSty{do}~t) K) \mapsto \ell_{1}$ and to show that 
$\cC(t' (\KwSty{while}~B~\KwSty{do}~t) K) \subseteq_{_{\!/\cong}} \extr_{hp_t}^{\GP}(\cC((\KwSty{while}~B~\KwSty{do}~S)K))$. 

\item[{\rm (ii)}]  $c_{1}\equiv \KwSty{skip};$ and $S_0 \equiv \KwSty{skip}; T (\KwSty{while}~B~\KwSty{do}~S)K$. Thus,  
$S_1 \equiv T (\KwSty{while}~B~\KwSty{do}~S)K$,
so that
$A_0 \equiv \cd{skip}$. This case is analogous to the previous case~(i). 

\item[{\rm (iii)}]  $c_{1}\equiv \KwSty{skip};$ and $S_0 \equiv (\KwSty{while}~B'~\KwSty{do}~S') T (\KwSty{while}~B~\KwSty{do}~S)K$.
Thus, $S_1 \equiv (\KwSty{if}~B'~\KwSty{then}~(S' (\KwSty{while}~B'~\KwSty{do}~S'))) T (\KwSty{while}~B~\KwSty{do}~S)K$ and  
$A_0 \equiv \cd{skip}$. Here, we have that
\begin{align*}
\cC(t (\KwSty{while}~B~\KwSty{do}~t) K) = \:
&\{\elb(t (\KwSty{while}~B~\KwSty{do}~t) K):\cd{skip} \ra \elb(t' (\KwSty{while}~B~\KwSty{do}~t) K)\}\\ 
& \cup  \cC(t' (\KwSty{while}~B~\KwSty{do}~t) K).
\end{align*}
 Again, it is enough to consider the relabeling 
$\elb(t' (\KwSty{while}~B~\KwSty{do}~t) K) \mapsto \ell_{1}$ and to show that 
$\cC(t' (\KwSty{while}~B~\KwSty{do}~t) K) \subseteq_{_{\!/\cong}} \extr_{hp_t}^{\GP}(\cC((\KwSty{while}~B~\KwSty{do}~S)K))$. 

\item[{\rm (iv)}]  $c_{1}\equiv \KwSty{bail}~\neg B'~\KwSty{to}~(T (\KwSty{while}~B~\KwSty{do}~S)K)$, with 
$S_0 \equiv (\KwSty{if}~B'~\KwSty{then}~S') T (\KwSty{while}~B~\KwSty{do}~S)K$ and 
$\bsem\grasse{B'}\rho_0 = \textit{true}$, so that $S_1 \equiv S' T (\KwSty{while}~B~\KwSty{do}~S)K$ and $A_0 \equiv B'$.  
In this case:
\begin{align*}
\cC(t (\KwSty{while}~B~\KwSty{do}~t) K) = & 
\; \{\elb(t (\KwSty{while}~B~\KwSty{do}~t) K):\neg B'  \ra \elb(T (\KwSty{while}~B~\KwSty{do}~S)K), \\ 
&\;\;\; \elb(t (\KwSty{while}~B~\KwSty{do}~t) K): B'  \ra \elb(t' (\KwSty{while}~B~\KwSty{do}~t) K)
\}\\ 
&\mkern-10mu \cup \cC(T (\KwSty{while}~B~\KwSty{do}~S)K) \cup  \cC(t' (\KwSty{while}~B~\KwSty{do}~t) K),\\[5pt]  
\cC(S (\KwSty{while}~B~\KwSty{do}~S) K) = & 
\; \{\elb(S (\KwSty{while}~B~\KwSty{do}~S) K): B'  \ra \elb(S' T (\KwSty{while}~B~\KwSty{do}~S)K),\, \\ 
&\;\;\; \elb(S (\KwSty{while}~B~\KwSty{do}~S) K): \neg B'  \ra \elb(T (\KwSty{while}~B~\KwSty{do}~S) K)
\}\\ 
& \mkern-10mu \cup \cC(S' T (\KwSty{while}~B~\KwSty{do}~S)K) \cup 
\cC(T (\KwSty{while}~B~\KwSty{do}~S) K).
\end{align*} 
Hence, since $\elb(t(\KwSty{while}~B~\KwSty{do}~t)K)  \mapsto \ell_{0}$ and $A_0\equiv B'$, it is enough to consider 
the relabeling  $\elb(t'(\KwSty{while}~B~\KwSty{do}~t)K)  \mapsto \ell_{1}$ and to show that 
$\cC(t' (\KwSty{while}~B~\KwSty{do}~t) K) \subseteq_{_{\!/\cong}} \extr_{hp_t}^{\GP}(\cC((\KwSty{while}~B~\KwSty{do}~S)K))$. 

\item[{\rm (v)}] 
$c_{1}\equiv \KwSty{bail}~ B'~\KwSty{to}~(S' T (\KwSty{while}~B~\KwSty{do}~S)K)$, with 
$S_0 \equiv (\KwSty{if}~B'~\KwSty{then}~S') T (\KwSty{while}~B~\KwSty{do}~S)K$ and 
$\bsem\grasse{B'}\rho_0 = \textit{false}$, so that $S_1 \equiv T (\KwSty{while}~B~\KwSty{do}~S)K$ and 
$A_0 \equiv \neg B'$.  
In this case:
\begin{align*}
\cC(t (\KwSty{while}~B~\KwSty{do}~t) K) = & 
\; \{\elb(t (\KwSty{while}~B~\KwSty{do}~t) K): B'  \ra \elb(S' T (\KwSty{while}~B~\KwSty{do}~S)K), \\ 
&\;\;\; \elb(t (\KwSty{while}~B~\KwSty{do}~t) K): \neg B'  \ra \elb(t' (\KwSty{while}~B~\KwSty{do}~t) K)
\}\\ 
&\mkern-10mu \cup \cC(S' T (\KwSty{while}~B~\KwSty{do}~S)K) \cup  \cC(t' (\KwSty{while}~B~\KwSty{do}~t) K),
\end{align*}  
while $\cC(S (\KwSty{while}~B~\KwSty{do}~S) K)$ is the same as in the previous point~(iv). 
Hence, since $\elb(t(\KwSty{while}~B~\KwSty{do}~t)K)  \mapsto \ell_{0}$ and $A_0\equiv \neg B'$, it is enough to consider 
the relabeling  $\elb(t'(\KwSty{while}~B~\KwSty{do}~t)K)  \mapsto \ell_{1}$ and to show that 
$\cC(t' (\KwSty{while}~B~\KwSty{do}~t) K) \subseteq_{_{\!/\cong}} \extr_{hp_t}^{\GP}(\cC((\KwSty{while}~B~\KwSty{do}~S)K))$. 
\end{itemize}

\noindent
Thus, in order to prove this containment, it remains to show that 
$\cC(t' (\KwSty{while}~B~\KwSty{do}~t) K) \subseteq_{_{\!/\cong}} \extr_{hp_t}^{\GP}(\cC((\KwSty{while}~B~\KwSty{do}~S)K))$. 
If $t'=\epsilon$ then the containment boils down to 
$\cC((\KwSty{while}~B~\KwSty{do}~t) K) \subseteq_{_{\!/\cong}} \extr_{hp_t}^{\GP}(\cC((\KwSty{while}~B~\KwSty{do}~S)K))$ which is
therefore proved. Otherwise, $t' = c_2 t''$, so that the containment can be inductively proved by using the same five cases (i)-(v) above. 

\medskip
Let us now show the reverse containment, that is, $\extr_{hp_t}^{\GP}(\cC((\KwSty{while}~B~\KwSty{do}~S)K)) \subseteq_{_{\!/\cong}}
\cC((\KwSty{while}~B~\KwSty{do}~t)K)$.
For the trace $t=c_1 c_2 ...c_{n+1}$, we know by $(*)$ that each command $c_i$ either is in $\{\KwSty{skip};,\: x:=E;\}$ 
or is one of the two following bail commands (cf.~cases~(3) and (4) in $(*)$):  
$$\KwSty{bail}~\neg B'~\KwSty{to}~(T (\KwSty{while}~B~\KwSty{do}~S)K),\quad\qquad
\KwSty{bail}~ B'~\KwSty{to}~(S' T (\KwSty{while}~B~\KwSty{do}~S)K).$$
Furthermore, at least a bail command occurs in $t$ because  there exists at least a conditional command $H_k$ in $hp_t$. Let $c_k$, with 
$k\in [1,n+1]$, be the first bail command occurring in $t$. Thus, since the sequence $c_1...c_{k-1}$ consists of skip and assignment commands only, 
we have that $\cC(t (\KwSty{while}~B~\KwSty{do}~t) K) \supseteq \cC(c_k...c_{n+1} (\KwSty{while}~B~\KwSty{do}~t) K)$. 
Hence, either $\cC(c_k...c_{n+1} (\KwSty{while}~B~\KwSty{do}~t) K) \supseteq \cC(T (\KwSty{while}~B~\KwSty{do}~S)K)$ or 
$\cC(c_k...c_{n+1} (\KwSty{while}~B~\KwSty{do}~t) K) \supseteq \cC(T (\KwSty{while}~B~\KwSty{do}~S)K)$. In both cases, 
we obtain that $\cC(c_k...c_{n+1} (\KwSty{while}~B~\KwSty{do}~t) K) \supseteq \cC((\KwSty{while}~B~\KwSty{do}~S)K)$, so 
that $\cC((\KwSty{while}~B~\KwSty{do}~S)K)\subseteq \cC(t (\KwSty{while}~B~\KwSty{do}~t) K) \subseteq \cC((\KwSty{while}~B~\KwSty{do}~t) K)$. 
Thus, it remains to show that 
\begin{align*}
&\{\ell_{-2}: \cd{skip} \ra \ell_{-1},\, \ell_{-1}: B \ra \ell_0,\, \ell_{-1}: \neg B \ra \elb(K)\}\:\cup \\
&\{\ell_{i} : A_i \ra \ell_{\textit{next}(i)}~|~ i\in [0,n]\} \cup 
\{ \ell_i : \neg A_i \ra \elb(S_{\textit{next}(i)})^c~|~ i\in [0,n],\, A_i \in \BExp\}
\end{align*}
 is contained in 
$\cC(t (\KwSty{while}~B~\KwSty{do}~t) K)$. 
We consider the following label renaming: 
\begin{align*}
 \ell_{-2} &\mapsto  \elb ((\KwSty{while}~B~\KwSty{do}~t)K) \\
  \ell_{-1} &\mapsto  \elb ((\KwSty{if}~B~\KwSty{then}~(t\, (\KwSty{while}~B~\KwSty{do}~t)))\, K) \\
  \ell_{0}  &\mapsto  \elb(t(\KwSty{while}~B~\KwSty{do}~t)K)
\end{align*} 
so that it remains to check that for any $i\in [0,n]$, the commands $\ell_{i} : A_i \ra \ell_{\textit{next}(i)}$ and 
$\ell_i : \neg A_i \ra \elb(S_{\textit{next}(i)})^c$, when $A_i \in \BExp$, are in $\cC(t (\KwSty{while}~B~\KwSty{do}~t) K)$. 
We analyze the possible five cases listed in $(*)$ for the action $A_0$: 
\begin{itemize}
\item[{\rm (i)}] $A_0\equiv \cd{skip}$ because
$S_0\equiv \KwSty{skip}; T (\KwSty{while}~B~\KwSty{do}~S)K$. Here, $t= \KwSty{skip};t'$. 
Hence, $\elb(t (\KwSty{while}~B~\KwSty{do}~t) K): \cd{skip} \ra \elb(t' (\KwSty{while}~B~\KwSty{do}~t) K) \in 
\cC(t (\KwSty{while}~B~\KwSty{do}~t) K)$ and 
it is enough to use the relabeling $\ell_1 \mapsto \elb(t' (\KwSty{while}~B~\KwSty{do}~t)K)$. 

\item[{\rm (ii)}] $A_0 \equiv x:=e$ because
$S_0\equiv x:=E;\, T (\KwSty{while}~B~\KwSty{do}~S)K$, so that  $t\equiv x:=E;\,t'$. Analogous to case~(i).  

\item[{\rm (iii)}] $A_0 \equiv \cd{skip}$ because
$S_0 \equiv (\KwSty{while}~B'~\KwSty{do}~S') T (\KwSty{while}~B~\KwSty{do}~S)K$. Here, $t=\KwSty{skip};t'$. Here, again, 
$\elb(t (\KwSty{while}~B~\KwSty{do}~t) K):\cd{skip} \ra \elb(t' (\KwSty{while}~B~\KwSty{do}~t) K) \in 
\cC(t (\KwSty{while}~B~\KwSty{do}~t) K)$, so that 
it is enough to use the relabeling $\ell_1 \mapsto \elb(t' (\KwSty{while}~B~\KwSty{do}~t)K)$. 

\item[{\rm (iv)}] $A_0 \equiv B'$ because 
$S_0 \equiv (\KwSty{if}~B'~\KwSty{then}~S') T (\KwSty{while}~B~\KwSty{do}~S)K$ and $\bsem\grasse{B'}\rho_0 = \textit{true}$. 
Thus, $t= (\KwSty{bail}~\neg B'~\KwSty{to}~(T (\KwSty{while}~B~\KwSty{do}~S)K))t'$ and 
$S_1 \equiv   T (\KwSty{while}~B~\KwSty{do}~S)K$.  Note that $\elb(S_1)^c = \elb(T (\KwSty{while}~B~\KwSty{do}~S)K)$. 
Hence, 
\begin{align*}
&\elb(t (\KwSty{while}~B~\KwSty{do}~t) K):\neg B'  \ra \elb(T (\KwSty{while}~B~\KwSty{do}~S)K)\in 
\cC(t (\KwSty{while}~B~\KwSty{do}~t) K), \\ 
& \elb(t (\KwSty{while}~B~\KwSty{do}~t) K): B'  \ra \elb(t' (\KwSty{while}~B~\KwSty{do}~t) K) \in 
\cC(t (\KwSty{while}~B~\KwSty{do}~t) K).
\end{align*}
Once again, the relabeling $\ell_1 \mapsto \elb(t' (\KwSty{while}~B~\KwSty{do}~t)K)$ 
allows us to obtain that $\ell_{0} : B' \ra \ell_{1}$ and 
$\ell_0 : \neg B' \ra \elb(T (\KwSty{while}~B~\KwSty{do}~S)K)$ are in $\cC(t (\KwSty{while}~B~\KwSty{do}~t) K)$.

\item[{\rm (v)}] $A_0 \equiv \neg B'$ because
$S_0 \equiv (\KwSty{if}~B'~\KwSty{then}~S') T (\KwSty{while}~B~\KwSty{do}~S)K$ and 
$\bsem\grasse{B'}\rho_0 = \textit{false}$. Here, 
$t= (\KwSty{bail}~ B'~\KwSty{to}~(S' T (\KwSty{while}~B~\KwSty{do}~S)K))t'$ and 
$S_1 \equiv  T (\KwSty{while}~B~\KwSty{do}~S)K$. Note that $\elb(S_1)^c = \elb(S' T (\KwSty{while}~B~\KwSty{do}~S)K)$.
Hence, 
\begin{align*}
&\elb(t (\KwSty{while}~B~\KwSty{do}~t) K): B'  \ra \elb(S' T (\KwSty{while}~B~\KwSty{do}~S)K) \in 
\cC(t (\KwSty{while}~B~\KwSty{do}~t) K),\\
&\elb(t (\KwSty{while}~B~\KwSty{do}~t) K): \neg B'  \ra \elb(t' (\KwSty{while}~B~\KwSty{do}~t) K) \in 
\cC(t (\KwSty{while}~B~\KwSty{do}~t) K).
\end{align*}
Thus, through
the relabeling $\ell_1 \mapsto \elb(t' (\KwSty{while}~B~\KwSty{do}~t)K)$ we obtain that 
$\ell_{0} : \neg B' \ra \ell_{1}$ and 
$\ell_0 : B' \ra \elb(S' T (\KwSty{while}~B~\KwSty{do}~S)K)$ are in $\cC(t (\KwSty{while}~B~\KwSty{do}~t) K)$.

\end{itemize}

\noindent 
This case analysis (i)-(v) 
for the action $A_0$ can be iterated for all the other actions $A_i$, with $i\in [1,n]$, 
and this allows us to close the proof.
\end{proof}
 
Finally, we can also state the correctness of the GP trace extraction transform for the
store changes abstraction as follows. 

\begin{theorem}[\textbf{Correctness of GP trace extraction}]\label{GP-corr-th}
For any $P\in \Program$, $hp= C_0\cdots C_n
\in \alpha_{\mathit{hot}}^{\GP}( \mathbf{T}\grasse{P})$, we have that
$\alpha_{sc}(\mathbf{T}\grasse{\extr_{hp}^{\GP}(P)})=
\alpha_{sc}(\mathbf{T}\grasse{P})$.
\end{theorem}

The proof of Theorem~\ref{GP-corr-th} is omitted, since it is a conceptually straightforward 
adaptation of the proof technique for the analogous Theorem~\ref{corr-th} on the correctness of trace extraction. 
Let us observe that since $\alpha_{sc}$ is a stronger abstraction than $\alpha^{\rho}_{sc}$ and, by 
Theorem~\ref{bis-abs-th}, we know that $\alpha^{\rho}_{sc}$ characterizes bisimilarity, we obtain
the so-called Stitch lemma in \cite[Lemma~3.6]{palsberg}
as a straight consequence of Theorem~\ref{GP-corr-th}: $\alpha^{\rho}_{sc}(\mathbf{T}\grasse{\extr_{hp}^{\GP}(P)})=
\alpha^{\rho}_{sc}(\mathbf{T}\grasse{P})$.

\section{Conclusion and Further Work}\label{concl-sec}
This article put forward a formal model of tracing JIT compilation which allows: 
(1)~an easy definition of program hot paths---that is, most frequently executed
program traces; (2)~to prove the 
correctness of a hot path extraction transform of programs; (3)~to prove 
the correctness of dynamic optimizations confined to  hot paths, such as dynamic type specialization along a hot path.  
Our approach is based on two main ideas: the use of a standard trace semantics for 
modeling the behavior of programs and 
the use of abstract interpretation for defining the notion of hot path as an abstraction 
of the trace semantics and for proving the correctness of hot path extraction and optimization. 
We have shown that this framework is more flexible than \citeN{palsberg} model of tracing JIT compilation, 
which relies on a notion of correctness based on operational program bisimulations, and 
allows to overcome some limitations of \cite{palsberg} on selection and annotation of hot paths and 
on the correctness of optimizations such as dead store elimination. 
We expect that
most optimizations employed by
tracing JIT compilers  can be 
formalized
and proved correct using the proof methodology of our framework. 

We see a number of interesting avenues for
further work on this topic.  As a significant example of  optimization implemented by 
a practical tracing compiler, it would be worth to cast in our model the allocation
removal optimization for Python described by~\citeN{bolz2011} in order
to formally prove its correctness. Then, we think that our framework could be adapted
in order to provide a model of whole-method just-in-time compilation, as used, e.g.,
by IonMonkey \cite{ionmonkey}, the current JIT compilation scheme 
in the Firefox JavaScript engine. Finally, the main ideas of our model could be useful to
study and relate the foundational differences between traditional
static vs dynamic tracing compilation. 

\begin{acks}
We are grateful to the anonymous referees for their helpful comments.
\end{acks}

\bibliographystyle{ACM-Reference-Format-Journals}

\end{document}